\documentclass[11pt]{article}
\usepackage{amsfonts}
\usepackage{amsmath}
\usepackage{lipsum}
\usepackage{amssymb}
\usepackage{amsmath}
\usepackage{amsthm}
\usepackage{geometry}
\usepackage{multirow}
\usepackage{booktabs}
\usepackage{changepage}
\usepackage{xcolor}
\usepackage{xr}

\usepackage[authoryear]{natbib}
\renewcommand{\theequation}{\thesection.\arabic{equation}}
\newcommand{\myref}[2]{\hyperref[#1]{#2}}
\numberwithin{equation}{section}

\usepackage{latexsym}
\usepackage{graphicx} 
\usepackage{enumerate}
\usepackage{setspace}
\def\qed{\rule{2mm}{2mm}}

\usepackage[toc,title,titletoc,header]{appendix}
\usepackage[bottom]{footmisc} 
\usepackage[pdfborder={0 0 0}]{hyperref}
\hypersetup{
 colorlinks=true,linkcolor=magenta, citecolor=blue, filecolor=magenta, 
 linkbordercolor={0 1 1}, citebordercolor={1 0 0},
}

\usepackage{lineno}
\setpagewiselinenumbers

\allowdisplaybreaks

\newtheorem{theorem}{Theorem}
\newtheorem{lemma}{Lemma}
\newtheorem{assumption}{Assumption}
\newtheorem{remark}{Remark}
\newcounter{assumptionM}
\newcounter{assumptionA}
\def\theassumptionM{M.\arabic{assumptionM}}
\def\theassumptionA{A.\arabic{assumptionA}}
\usepackage{etoolbox} 
\AtEndEnvironment{remark}{~\qed}%

\begin{document}
\author{
Federico A. Bugni \\
Department of Economics\\
Northwestern University \\
\url{federico.bugni@northwestern.edu}
\and
Ivan A.\ Canay \\
Department of Economics\\
Northwestern University\\
\url{iacanay@northwestern.edu}
\and
Deborah Kim \\
Department of Economics\\
University of Warwick\\
\url{deborah.kim@warwick.ac.uk}
}

\title{Testing Conditional Stochastic Dominance at Target Points\footnote{We thank Xinran Li, David Kaplan, Matias Cattaneo, and participants at several conferences and seminars for helpful comments and discussion.}}

\maketitle

\begin{spacing}{1.1}
  \begin{abstract}
    This paper introduces a test for conditional stochastic dominance between two distributions at prespecified values of a conditioning covariate, referred to as target points. The test uses a one-sided Kolmogorov--Smirnov statistic computed from induced order statistics, the outcomes attached to the conditioning observations closest to the target point, and compares it to a critical value that, given the number of neighbors, requires no resampling, kernel smoothing, or parametric assumptions. The same procedure applies whether the outcomes are continuous, discrete, or mixed, and requires only continuity of the conditional distributions in the conditioning variable. We establish asymptotic validity under two frameworks: one in which the number of neighbors is held fixed, where the induced order statistics converge to independent draws from the conditional distributions at the target point; and one in which it grows with the sample size, where we obtain an explicit rate that accommodates both an estimated target point and a data-dependent choice of the number of neighbors. We connect the test to permutation-based inference, provide a refined critical value for discrete outcomes, propose a rule for selecting the tuning parameters, and illustrate the procedure in two empirical applications whose recorded outcomes exhibit mass points. Monte Carlo simulations confirm its strong finite-sample performance.
    \end{abstract}
\end{spacing}

\noindent KEYWORDS: stochastic dominance, regression discontinuity design, induced order statistics, rank tests, permutation tests.

\noindent JEL classification codes: C12, C14.

\maketitle

\thispagestyle{empty}

\newpage 

\maketitle

\section{Introduction}\setcounter{page}{1}
The concept of stochastic dominance has long been central to many areas of applied research. This paper studies a conditional version of the problem: testing conditional stochastic dominance (CSD) at one or more prespecified values of a conditioning covariate, which we call target points. Comparisons of this kind arise whenever the question of interest concerns a particular value of the conditioning variable rather than its entire range. Examples include the outcome distributions on either side of the cutoff in a regression discontinuity design (RDD), and the comparison of two groups at a fixed quantile of a covariate such as skill or experience.

Two empirical applications, developed in Section~\ref{app:empirical-application}, motivate and illustrate the procedure. In the first, we ask whether retirement shifts the distribution of household expenditure upward, following \citet{battistin/Brugianivi/Rettore/Weber:2009} and \citet{shen/zhang:16}. This application uses a regression discontinuity in the age of the household head, with the known retirement-age cutoff as the target point. In the second, we ask whether small classes stochastically dominate regular classes in student achievement at the median level of teacher experience, using the Project STAR data as in \citet{qu/yoon:2015}. Here, the target point is a quantile of the conditioning variable and must be estimated. In both applications, the recorded outcomes exhibit mass points, making atomlessness of the conditional outcome distributions an unattractive maintained assumption. Existing procedures for CSD at a target point typically impose this assumption, often together with smoothness conditions, and treat the target point as known. Our procedure accommodates continuous, discrete, and mixed outcomes and allows the target point to be either known or estimated.

There is a long literature on inference for unconditional stochastic dominance, which compares entire marginal distributions; see \citet{whang:2019} for a textbook treatment. A related and equally active literature studies inference for conditional stochastic dominance, comparing subgroups defined by covariates; see \citet[Chapter 4]{whang:2019} for a review. As that review makes clear, this literature addresses two distinct problems: dominance uniformly over a range of conditioning values, and dominance at one or more specific target points. Our paper concerns the latter, a problem that recurs across applied work; see, among others, \citet{shen:2019}, \citet{qu/yoon:2019}, \citet{barcellos/carvalho/turley:2023}, and \citet{Li/Sunder:2024}. We review both branches and position our contribution relative to them after describing the test.

The primary goal of this paper is to test whether the conditional cumulative distribution function (CDF) of one variable stochastically dominates that of the other at specific values of a conditioning variable. Formally, we consider the null hypothesis  
\[
H_0: F_Y(t | z) \leq F_X(t | z) \quad \text{for all } (t, z) \in \mathbf{R} \times \mathcal{Z}~,
\]
against the alternative that there exists some \((t, z)\) for which the reversed inequality holds strictly. Here, \( F_Y(\cdot | z) \) and \( F_X(\cdot | z) \) represent the conditional CDFs of the random variables \( Y \) and \( X \), respectively, given \( Z = z \). Importantly, we focus on situations where the set of target values,  $\mathcal{Z}$, is not the entire support of $Z$, but rather a finite collection of points (including the case of a singleton), which may be known or estimated from the data. To address this testing problem, we propose a novel procedure that leverages induced order statistics based on independent samples from \( Y \) and \( X \). Our test statistic, a Kolmogorov--Smirnov-type measure, captures the maximal deviation between the empirical CDFs of the two samples, conditional on observations near the target points. Crucially, the critical value we propose is derived in a deterministic, non-data-dependent manner once the tuning parameters are accounted for, ensuring computational simplicity. 

Our contributions in this paper are both methodological and theoretical. First, we introduce a novel test for CSD at target points. The test is built from induced order statistics: it orders the conditioning observations by their distance to the target point and compares the outcomes attached to the $q_y$ and $q_x$ closest observations in each of two independent samples through a one-sided Kolmogorov--Smirnov statistic. Like nonparametric methods generally, the procedure requires two tuning parameters: the numbers of neighbors $q_y$ and $q_x$, which play a role analogous to a bandwidth. What distinguishes the test is its scope. The same procedure applies without modification whether $Y$ and $X$ are continuous, discrete, or mixed, and it requires only that the conditional CDFs be continuous in $z$ at each target point, rather than differentiable or otherwise smooth. It uses no kernel weighting and imposes no parametric structure on the conditional distributions.

Second, we establish the asymptotic validity of our test under two alternative asymptotic frameworks. In the first framework, the numbers of effective ``local'' observations, $q_y$ and $q_x$, are held fixed as the sample size $n \to \infty$. In this setting, we show that the test statistic converges to a limiting experiment in which the induced order statistics behave as independent draws from the true conditional CDFs at the target point. This convergence yields a feasible critical value that remains valid without relying on resampling methods such as the bootstrap. Our regularity conditions in this fixed-$q$ framework are mild: the conditional CDFs at the target point may contain finitely many discontinuities; $Y$ and $X$ may be continuous, discrete, or mixed; and we require only that the conditional CDFs $F_Y(t \mid z)$ and $F_X(t \mid z)$ be continuous in $z$ at the target point, uniformly in $t$. This is a weaker condition than the smoothness assumptions, such as twice differentiability in $z$, typically imposed in nonparametric methods. In the second framework, the numbers of effective observations $q_y$ and $q_x$ are allowed to diverge as $n \to \infty$. Building on recent results on convergence rates for induced order statistics \citep{bugni/canay/kim:26b}, we derive an explicit, finite-sample bound on the rate at which the rejection probability of our test approaches its nominal level as a function of $q_y$, $q_x$, and $n$. We establish this result at a level of generality that mirrors how the test is used in practice: the conditioning point may be estimated from the data, and the tuning parameters $(q_y, q_x)$ may themselves be chosen in a data-dependent manner, with the known-point and deterministic-$q$ results each following as a special case. This second framework complements the fixed-$q$ results by demonstrating that the test continues to control size when the number of effective observations increases, and it provides the theoretical basis for the data-dependent rule we propose for selecting $q_y$ and $q_x$. A feature common to both frameworks is that the target point need not be known in advance: we allow it to be estimated from the data, for example, when it is a quantile of the conditioning variable, and show that this does not affect the asymptotic validity of the test. This size control can be conservative when the outcomes are discrete or mixed, in which case the achieved level may lie strictly below $\alpha$. As we discuss next, a refined critical value mitigates this when $Y$ and $X$ have few support points.

Third, we show that the proposed critical value aligns with the one obtained from a permutation-based approach when the random variables $Y$ and $X$ are both continuous, thus establishing a natural connection between our method and the broader literature on permutation-based inference. To the best of our knowledge, this result provides the first formal justification for the validity of permutation-based inference in testing unconditional stochastic dominance. We demonstrate that the critical value of our test cannot be improved when both $Y$ and $X$ are continuous. This connection is not merely technical: it situates our test alongside design-based approaches to treatment-effect inference, where permutation and randomization arguments are standard, and it explains why a data-independent critical value can deliver finite-sample control in our setting. We develop the connection in Section~\ref{sec:permutations}, where we also relate it to the design-based analysis of \citet{caughey/etal:2023}. However, this result on validity of permutation tests does not extend to the case when either $Y$ or $X$ is discretely distributed. For this latter case, we introduce a \emph{refined} critical value, which is typically smaller than the default one we propose and only a function of the support points for $Y$ and $X$. This refinement enhances power relative to the default critical value, though it comes at the cost of increased computational complexity. 

Finally, we examine the finite-sample performance of our test through Monte Carlo simulations and present two empirical application to illustrate its implementation. The Monte Carlo simulations and the empirical applications suggest that the data-dependent rule we propose for selecting the key tuning parameters performs well in practice.

\subsection*{Related literature}

Our work contributes to the extensive literature on inference for unconditional stochastic dominance. Early contributions include \cite{hodges:1958} and \cite{mcfadden:1989}, with seminal subsequent developments in \cite{anderson:1996}, \cite{davidson/duclos:2000}, \cite{abadie:2002}, \cite{barrett/donald:2003}, \cite{Linton/maasoumi/whang:2005}, \cite{linton/song/whang:2010}, \cite{davidson/duclos:2013}, and \cite{lok/tabri:2021}, among others. \citet{whang:2019} provides a comprehensive textbook treatment of this literature. This literature studies hypotheses of stochastic dominance between marginal distributions and predominantly relies on asymptotic approximations and resampling methods. As explained above, our paper instead studies conditional stochastic dominance at specific target points. Nevertheless, our paper is connected to the unconditional literature because, in the limit experiment underlying our procedure, the conditional stochastic dominance testing problem reduces to a finite-sample unconditional one.

Our paper lies within the literature on inference for conditional stochastic dominance. As explained in \citet[Chapter 4]{whang:2019}, this literature considers two distinct testing problems. The first concerns whether conditional stochastic dominance holds over a range or the entire support of the conditioning variable. The second concerns whether it holds at one or more specific points in that support. As explained earlier, our paper studies the latter problem. There is an extensive literature on the former, including \cite{delgado/escaciano:2013}, \cite{gonzalo/olmo:2014}, \cite{chang-lee-whang-2015-EJ}, \cite{linton/whang/yen:2016}, \cite{andrews-shi-2017-JOE}, \cite{linton/seo/whang:23}, and \cite{wu/kaplan:2025}, among others. These procedures test restrictions that hold almost surely over a range of conditioning values and therefore do not directly provide inference for the conditional distributions at an isolated target point.

Our paper lies squarely within the literature on testing stochastic dominance at specific target points. A growing body of related work develops methods for distributional comparisons in this setting, including \cite{donald/hsu/barrett:2012}, \cite{qu/yoon:2015}, \cite{shen/zhang:16}, \cite{shen:2019}, \citet[Sections 6.1-6.2]{goldman/kaplan:2018}, and \cite{qu/yoon:2019}. As mentioned earlier, the focus on particular target points may be dictated by the research design, as with the cutoff in an RDD, or by the substantive economic question, as with the median or another prespecified quantile of interest. Our method requires weaker assumptions than those developed in the previous studies. First, all existing methods assume that the conditional outcome distributions are continuous, and some impose additional smoothness conditions on the corresponding density functions. Our method, by contrast, accommodates continuous, discrete, and mixed data without requiring any modification to the inferential procedure. We view this flexibility with respect to the distributional form of the data as an important advantage of our approach. Second, some methods, such as that of \cite{qu/yoon:2019}, are based on conditional quantiles and require the quantile indices to be bounded away from the tails of the distribution. We impose no such restriction. Finally, all existing methods treat the target point as fixed and known, whereas our method also permits it to be estimated, as occurs, for example, when the target point is defined as a quantile of the distribution of the conditioning variable. These differences are not merely incremental refinements. To our knowledge, none of the existing methods covers the two-sample problem at a target point while allowing one or both conditional outcome distributions to contain atoms. 

The remainder of the paper is organized as follows. Section~\ref{sec:testing-problem} defines the testing problem and introduces notation. Section~\ref{sec:our-test} presents the induced order statistics that form the basis of our procedure and introduces the proposed test. Section~\ref{sec:results} establishes its asymptotic validity under two alternative asymptotic frameworks. Section~\ref{sec:extensions} discusses several refinements and extensions, including results on rates of convergence, a data-dependent rule for selecting the tuning parameters, and a refinement of the test that increases power when $Y$ and $X$ are discrete. Importantly, Section~\ref{sec:extensions} also provides formal results on the validity of permutation tests for stochastic dominance with continuously distributed random variables, offering novel arguments for the validity of permutation-based inference in this setting. Section~\ref{sec:simulations} evaluates the finite-sample performance of our test through Monte Carlo simulations. Finally, Section~\ref{sec:conclusion} concludes. All proofs are contained in the Appendix. 

\section{Testing problem}\label{sec:testing-problem}
Let $(Y,Z)$ and $(X,Z)$ denote random pairs, each supported on $\mathbf{R}\times\mathbf{R}$. Let $P_Y$ and $P_X$ denote the joint distributions of $(Y,Z)$ and $(X,Z)$, respectively. For $t\in\mathbf{R}$ and $z$ in the support of $Z$, define the conditional distribution functions
\[
F_Y(t | z) := P_Y\{Y \le t \mid Z=z\}
\qquad\text{and}\qquad
F_X(t | z) := P_X\{X \le t \mid Z=z\}~.
\]

We are interested in testing the null hypothesis:
\begin{equation}\label{eq:H0}
  H_0: F_Y(t|z) \le F_X(t|z) \text{ for all } (t, z) \in \mathbf{R} \times \mathcal{Z}
\end{equation}
versus the alternative hypothesis:
\begin{equation*}\label{eq:H1}
  H_1: F_Y(t|z) > F_X(t|z) \text{ for some } (t, z) \in \mathbf{R} \times \mathcal{Z}~,
\end{equation*}
where $\mathcal{Z} = \{z_1, \dots, z_L\}$ is a finite set of target values of the conditioning variable. When $L=1$, we write $\mathcal Z=\{z_0\}$. The case where $L = 1$ and $Z$ is a continuously distributed random variable is both simpler and particularly relevant in empirical applications. To minimize notational clutter, we focus on this case for most of the remainder of the paper, unless otherwise stated, with Section \ref{app:multiple-target} addressing the case where $L > 1$.

To test this hypothesis, we assume the analyst observes two independent samples:
\begin{equation}\label{eq:obs-data}
\{(Y_i,Z_i):1\le i\le n_y\}
\qquad\text{and}\qquad
\{(X_j,Z_j):1\le j\le n_x\}~.
\end{equation}
We refer to the first sample as the $Y$-sample, which consists of $n_y$ i.i.d.\ draws from $P_Y$, and to the second sample as the $X$-sample, which consists of $n_x$ i.i.d.\ draws from $P_X$. A leading example, and the one that motivates much of our analysis, is the sharp RDD. There, an outcome $\tilde{Y}$ is observed together with a running variable $\tilde{Z}$, treatment is assigned to units on one side of a known cutoff $z_0$, and interest centers on comparing the outcome distributions just above and just below $z_0$. The two samples in~\eqref{eq:obs-data} arise by splitting the data at the cutoff,
\begin{gather*}
    \{(Y_i, Z_i): 1 \le i \le n_y\} = \{(\tilde{Y}_i, \tilde{Z}_i): \tilde{Z}_i \le z_0\}, \\
    \{(X_j, Z_j): 1 \le j \le n_x\} = \{(\tilde{Y}_j, \tilde{Z}_j): \tilde{Z}_j > z_0\}~,
\end{gather*}
so that testing \eqref{eq:H0} at the single target point $z_0$ asks whether one outcome distribution conditionally dominates the other at the threshold. In this design, $z_0$ is a boundary point of the support of $Z$ within each subsample; our conditions accommodate this case, and we return to it in the Monte Carlo simulations of Section~\ref{sec:simulations} and the empirical application of Section~\ref{app:empirical-application}.

Independence of the samples implies that the data-generating distribution is the product measure
\[
P := P_Y \otimes P_X~.
\]
All probability and expectation statements in what follows are taken with respect to $P$ unless otherwise noted. We denote by $\mathbf P$ the class of data-generating distributions $P$ satisfying the regularity conditions imposed later, and let 
\begin{equation}\label{eq:P_0}
  \mathbf P_{0} := \{P \in \mathbf P: \eqref{eq:H0}~ \text{holds}\}
\end{equation}
denote the subset of distributions $P\in \mathbf P$ satisfying the null hypothesis in \eqref{eq:H0}. 

\begin{remark}\label{rem:one-sample}
  Our analysis assumes two independent samples. If instead one observes i.i.d.\ draws
  \[
  \{(Y_i,X_i,Z_i):1\leq i\leq n\}~,
  \]
  the nearest-neighbor step can still be applied by keeping the observations whose values of $Z_i$ are closest to $z_0$. The resulting local experiment, however, consists of draws from the joint conditional law $\mathcal L\big((Y,X)\mid Z=z_0\big)$ rather than independent draws from $P_{Y}$ and $P_{X}$, and since the null hypothesis in \eqref{eq:H0} only restricts the two marginal conditional distributions, the within-pair dependence is left unrestricted. This substantially changes the formal analysis and would likely lead to a different critical value, which we leave for future work.
  \end{remark}

\section{Test based on induced ordered statistics}\label{sec:our-test}

Let the observed data be those given in \eqref{eq:obs-data}. Let $(q_y,q_x)$ be two small positive integers (relative to $(n_y,n_x)$) and consider the point $z_{0}\in \mathcal Z$. The test we propose is based on the following two samples: 
\begin{itemize}
    \item The $q_y$ values of $\{Y_i:1\le i\le n_y\}$ associated with the $q_y$ values of $\{Z_i:1\le i\le n_y\}$ closest to $z_{0}$, and 
    \item The $q_x$ values of $\{X_j:1\le j\le n_x\}$ associated with the $q_x$ values of $\{Z_j:1\le j\le n_x\}$ closest to $z_{0}$.
\end{itemize}

To define these samples formally, we introduce $g$-order statistics for the conditioning variable $Z$, where $g(Z) := |Z - z_{0}|$; see \citet[Section 2.1]{reiss:89} and \cite{kaufmann/reiss:92}. For any two values \( z, z'\) in the support of $Z$, we define the ordering \(\le_{g}\) as follows:
\[
z \le_{g} z' \quad \text{if and only if} \quad g(z) \le g(z')~.
\]
This defines a $g$-ordering on the support of $Z$. The $g$-order statistics \( Z_{g,(i)} \) are then the values of \( Z \) ordered according to this criterion:
\[
Z_{g,(1)} \le_{g} Z_{g,(2)} \le_{g} \cdots \le_{g} Z_{g,(n)}~.
\]
If there are ties in the $g$-ordering, they can be resolved arbitrarily—for example, by relying on the original sample order.

We then take the values of $\{Y_i:1\le i\le n_y\}$ associated with the $q_y$ smallest $g$-ordered statistics of $Z$ in the $Y$-sample, denoted by 
\begin{equation}\label{eq:ios-Y}
  Y_{n_y,[1]},Y_{n_y,[2]}, \dots,Y_{n_y,[q_y]}~.
\end{equation}
That is, $Y_{n_y,[j]}=Y_k$ if $Z_{g,(j)}=Z_k$ for $k=1,\dots,q_y$. Similarly, we take the values of $\{X_i:1\le i\le n_x\}$ associated with the $q_x$ smallest $g$-ordered statistics of $Z$ in the $X$-sample, denoted by  
\begin{equation}\label{eq:ios-X}
  X_{n_x,[1]},X_{n_x,[2]}, \dots,X_{n_x,[q_x]}~.
\end{equation}
The random variables in \eqref{eq:ios-Y} and \eqref{eq:ios-X} are referred to as \emph{induced order statistics} or \emph{concomitants of order statistics}; see \cite{david/galambos:74,bhattacharya:74,canay/kamat:18}. Intuitively, we view these samples as independent samples of \(Y\) and \(X\), conditional on \(Z\) being `close' to \(z_{0}\). A key feature of our test is that it relies solely on these induced order statistics. Letting \(n := \min\{n_y,n_x\}\) and \(q := q_y + q_x\), the effective (pooled) sample is
\begin{equation}\label{eq:Sn}
  S_n = (S_{n,1}, \dots, S_{n,q}) := (Y_{n_y,[1]}, \dots, Y_{n_y,[q_y]}, X_{n_x,[1]}, \dots, X_{n_x,[q_x]})~.
\end{equation}
It is also important to note that the first \(q_y\) elements of \(S_n\) are associated with the \(Y\)-sample, while the remaining \(q_x\) elements come from the \(X\)-sample.

Having defined the induced order statistics, we now define our test statistic as 
\begin{equation}\label{eq:ks-stat}
  T(S_n) = \sup_{t\in \mathbf R} \left( \hat F_{n,Y}(t) - \hat F_{n,X}(t) \right) = \max_{k \in \{1,\dots,q_y\}} \left( \hat F_{n,Y}(S_{n,k}) - \hat F_{n,X}(S_{n,k}) \right) ~,
\end{equation}
where the empirical CDFs are 
\begin{equation*}
  \hat F_{n,Y}(t) := \frac{1}{q_y} \sum_{j=1}^{q_y} I\{S_{n,j} \le t\} \quad \text{ and }\quad \hat F_{n,X}(t) := \frac{1}{q_x} \sum_{j=1}^{q_x} I\{S_{n,q_y+j} \le t\}~.
\end{equation*}
The test statistic in \eqref{eq:ks-stat} is a two-sample one-sided Kolmogorov--Smirnov (KS) test statistic, see \citet[][p.\ 99]{hajek/sidak/sen:99}, and the test we propose rejects the null hypothesis in \eqref{eq:H0} when $T(S_n)$ exceeds a critical value, defined next.

To introduce the critical value, let $\alpha\in(0,1)$ be given and $\{U_j: 1\le j\le q\}$ be a sequence of uniform random variables i.i.d., that is, $U_j\sim U[0,1]$. Define $\Delta_{q_y,q}(u)$ as
\begin{equation}\label{eq:Delta-u}
    \Delta_{q_y,q}(u) := \frac{1}{q_y}\sum_{j=1}^{q_y} I\{U_{j} \le u\}-\frac{1}{q-q_y}\sum_{j=q_y+1}^{q} I\{U_{j} \le u\} ~,
\end{equation}
and
\begin{equation}\label{eq:our-cv}
  c_{\alpha}(q_y,q) := \inf_{x\in\mathbf R}\left\lbrace P\left\lbrace \sup_{u\in (0,1)} \Delta_{q_y,q}(u) \le x \right\rbrace\ge 1-\alpha \right\rbrace ~.
\end{equation}
The test we propose for the null hypothesis in \eqref{eq:H0} is 
\begin{equation}\label{eq:our-test}
  \phi(S_n) := I\{T(S_n)>c_{\alpha}(q_y,q)\}~.
\end{equation}
We reiterate that \eqref{eq:our-test} corresponds to our test with a single target point ($L = 1$). Section \ref{app:multiple-target} extends this framework to the general case with multiple target points ($L > 1$). Furthermore, in addition to our default critical value in \eqref{eq:our-cv}, we provide a refined critical value tailored for discrete $Y$ and $X$ with finite support in Section \ref{sec:discrete}.

The choice of the two-sample one-sided KS statistic in \eqref{eq:ks-stat} is crucial for accommodating discrete or mixed outcomes. Alternative statistics entail a familiar power trade-off: the KS statistic is naturally sensitive to pronounced positive peaks in $F_Y(t| z_0)-F_X(t| z_0)$ over a relatively narrow range, whereas integrated statistics such as Cram\'er--von Mises may have greater power against smaller but more diffuse departures, with no uniform ranking between them. In the continuous case, the Cram\'er--von Mises analogue remains asymptotically valid; with discrete or mixed outcomes, however, the analogous Cram\'er--von Mises and Anderson--Darling tests need not control size. This is why our default procedure relies on the KS statistic; see the discussion following Theorem~\ref{thm:SD-AsySz} and Section~\ref{app:other-stats}.

\begin{remark}\label{rem:c-approx}
From a computational perspective, a notable feature of the test $\phi(S_n)$ is that the supremum over $\mathbf{R}$ in \eqref{eq:ks-stat} can be replaced by a maximum over the set $\{1, \dots, q_y\}$. This simplification arises because the KS test statistic increases only when evaluated at points corresponding to the $Y$-sample, which are the first $q_y$ elements in the vector $S_n$. Consequently, the supremum in $c_{\alpha}(q_y, q)$ can also be replaced with a maximum over $\{1, \dots, q_y\}$. This allows us to compute the critical value $c_{\alpha}(q_y, q)$ with arbitrary accuracy through simulation.
\end{remark}

\begin{remark}\label{rem:conservative}
  When both $Y$ and $X$ are continuous, the test $\phi(S_n)$ is essentially exact; see Lemma~\ref{lem:limit-exact}. By contrast, in the discrete or mixed case, the test may be conservative because the default critical value in \eqref{eq:our-cv} is an upper bound on, rather than the relevant quantile. Consequently, the achieved level may remain strictly below $\alpha$ even asymptotically. We revisit this issue after Theorem~\ref{thm:SD-AsySz} and, in Section~\ref{sec:discrete}, develop a refined critical value that is less conservative for discrete outcomes.
\end{remark}

\begin{remark}\label{rem:choice-q}
     The values of $q_y$ and $q_x$ are tuning parameters chosen by the researcher. Section~\ref{sec:results} shows that the test is asymptotically valid when $(q_y, q_x)$ are either held fixed or allowed to diverge slowly with the sample size. In Section~\ref{sec:tuning-parameters}, we develop data-dependent guidelines for selecting $(q_y, q_x)$ that are compatible with these asymptotic requirements, and we assess their practical performance through simulations in Section~\ref{sec:simulations} and in the empirical applications in Section \ref{app:empirical-application}.
\end{remark}

\begin{remark}\label{rem:rdd}
  The RDD introduced in Section~\ref{sec:testing-problem} is a running example for our test. With observed sample $\{(\tilde{Y}_i, \tilde{Z}_i): 1 \le i \le n\}$ and cutoff $z_0$, the two samples in \eqref{eq:obs-data} are obtained by splitting the data at $z_0$, and the test compares the induced order statistics on either side of the threshold. The cutoff $z_0$ is a boundary point of the support of $Z$ within each subsample, a case our assumptions accommodate. 
\end{remark}

\begin{remark}\label{rem:estimated-z0}
In many applications, the conditioning point $z_0$ is not a fixed value but a feature of the distribution of $Z$ that the researcher estimates from the same data---the leading example being the median, or another quantile, of $Z$. The results below, including the RDD setting discussed in Remark \ref{rem:rdd}, are stated to formally account for this case: they allow $z_0$ to be replaced by a consistent estimator $\hat z_n$ that depends on the $Z$-observations alone, so our conclusions apply directly to such settings.
\end{remark}

\section{Asymptotic framework and formal results}\label{sec:results}

In this section, we derive the asymptotic properties of the test in \eqref{eq:our-test} using two alternative asymptotic frameworks. The first one requires $q := q_y+q_x$ to be fixed as $n\to \infty$ and represents a finite sample
situation where the effective number of observations used by the test is too small to credibly invoke approximations that require a ``large'' value of $q$. The second framework requires $q\to \infty$ slowly as $n\to \infty$, and represents a finite sample situation where the effective number of observations used by the test is large enough to invoke approximations for ``large'' $q$. Throughout, we allow the conditioning point $z_0$ to be either known or estimated from the data, as anticipated in Remark \ref{rem:estimated-z0}. Let $\mathcal A$ denote the $\sigma$-algebra generated by the $Z$-observations from both samples, and let $\hat z_n$ be the estimator of $z_0$ mentioned above. To avoid clutter, in the main text we continue to write $S_n$ for the induced order statistics in \eqref{eq:Sn}, with the understanding that they are constructed using the nearest neighbors of either $z_0$ or $\hat z_n$, depending on the context; the appendix makes this distinction explicit.

\subsection{Asymptotic results for fixed \emph{q}}\label{sec:fixed-q}

In this section, we examine the asymptotic properties of the test in \eqref{eq:our-test} within a framework where $q := q_y+q_x$ is fixed and $n:=\min\{n_y,n_x\}\to \infty$. We first derive the asymptotic properties of induced order statistics in \eqref{eq:Sn}, and then present our main theorem. 

We start by deriving a result on the induced order statistics collected in the vector $S_n$ in \eqref{eq:Sn}. To do so, we make the following assumptions. We state the assumptions for all $z\in\mathcal Z$ so that they also apply to the multiple-target case considered in Section \ref{app:multiple-target}.

\begin{assumption}\label{ass:zhat}
  For each $z\in\mathcal Z$, there exists an $\mathcal A$-measurable estimator $\hat z_n$ such that $\hat z_n \overset{p}{\to} z$.
  \end{assumption}

\begin{assumption}\label{ass:Z}
For any $\varepsilon >0$ and $z\in \mathcal{Z}$, $P\{Z\in (z-\varepsilon ,z+\varepsilon )\}>0$.
  \end{assumption}

\begin{assumption}\label{ass:cond-continuous}
    For any $z\in\mathcal Z$ and any sequence $\{z_k\}$ in the support of $Z$ such that $z_k\to z$,
    \[
    \sup_{t\in\mathbf R}\big|F_Y(t| z_k)-F_Y(t| z)\big|\to 0
    \quad\text{and}\quad
    \sup_{t\in\mathbf R}\big|F_X(t| z_k)-F_X(t| z)\big|\to 0~.
    \]
  \end{assumption}

The $\mathcal A$-measurability requirement in Assumption \ref{ass:zhat} restricts $\hat z_n$ to be a function of the $Z$-observations alone---as with the empirical median or any empirical quantile of $Z$---and rules out estimators that depend on the outcomes $Y$ or $X$. Taking $\hat z_n:= z_0$ is allowed for, so this assumption nests the known-$z_0$ case as a special case. Assumption \ref{ass:Z} requires that the distribution of $Z$ is locally dense at each of the points in $\mathcal{Z}$. This includes the case where $Z$ has a mass point at $z\in \mathcal Z$. Assumption \ref{ass:cond-continuous} is a smoothness assumption required to guarantee that conditioning on observations close to $z$ is informative about the distribution conditional on $Z=z$.

\begin{theorem}\label{thm:limit-Sn}
Let Assumptions \ref{ass:zhat}--\ref{ass:cond-continuous} hold. Then,
    \begin{equation}\label{eq:limit-Sn}
          S_n \stackrel{d}{\to} S=(S_1,\dots,S_{q})~,
    \end{equation}
    where for any $s:=(s_1,\dots,s_{q})\in \mathbf{R}^{q}$, the random vector $S$ satisfies
  \begin{equation*}
   P\{S\le s\} = \prod_{j=1}^{q_y} F_Y(s_j|z_{0}) \cdot \prod_{j=q_y+1}^{q} F_X(s_j|z_{0})~.
  \end{equation*}
\end{theorem}

Theorem \ref{thm:limit-Sn} is a special case of Theorem \ref{thm:limit-Sn-general} in the appendix with $L=1$, which generalizes \citet[][Theorem 4.1]{canay/kamat:18} in two directions: first, it accommodates multiple conditioning values ($L>1$); and second, it allows the conditioning point $z_0$ to be estimated by $\hat z_n$ rather than known, in which case $S_n$ is understood to be constructed using neighbors of $\hat z_n$. It establishes that the limiting distribution of the induced order statistics in the vector $S_n$ is such that the elements of the vector, denoted by $S$, are mutually independent. Specifically, the first $q_y$ elements of this vector follow the distribution $F_Y(\cdot|z_{0})$, while the remaining $q_x$ elements follow $F_X(\cdot|z_{0})$. The proof leverages the fact that the induced order statistics $S_n$ in \eqref{eq:Sn} are conditionally independent given $(Z_1,\dots,Z_n)$, with conditional CDFs $$F_Y(\cdot|{Z_{n_y,(1)}}),\dots,F_Y(\cdot|{Z_{n_y,(q_y)}}),F_X(\cdot|{Z_{n_x,(1)}}),\dots,F_X(\cdot|{Z_{n_x,(q_x)}})~.$$ The result then follows by showing that $Z_{n_y,(j)}\overset{p}{\to}z_0$ and $Z_{n_x,(j)}\overset{p}{\to}z_0$ for all $j\in\{1,\dots,q\}$, and invoking standard properties of weak convergence. Theorem \ref{thm:limit-Sn} plays a crucial role in our asymptotic validity result presented in Theorem \ref{thm:SD-AsySz}. 

In addition to Assumptions \ref{ass:zhat} to \ref{ass:cond-continuous}, we also require that, conditional on $Z=z$, the random variables $Y$ and $X$ have distributions with a finite number of discontinuity points. To state this assumption formally, let $\mathcal{D}_{Y}(z)$ and $\mathcal{D}_{X}(z)$ denote the sets of discontinuity points of the CDFs of ${ Y|Z=z }$ and ${ X|Z=z }$, respectively.

\begin{assumption}\label{ass:FiniteDisc}
For any $z \in \mathcal{Z}$, $|\mathcal{D}_{Y}(z)|$ and $|\mathcal{D}_{X}(z)|$ are finite.
\end{assumption}

It is important to note that Assumption \ref{ass:FiniteDisc} allows both $Y$ and $X$ to be continuous, discrete, or mixed random variables. However, it excludes cases where these variables have countably many discontinuities conditional on $z\in \mathcal Z$. We also point out that Theorem \ref{thm:limit-Sn} does not require Assumption \ref{ass:FiniteDisc}. 

We now formalize our main result in Theorem \ref{thm:SD-AsySz}, which shows that the test defined in \eqref{eq:our-test} is asymptotically level $\alpha$ under the assumptions we just introduced. Below, we denote by $E_P[\cdot]$ the expected value with respect to the distribution $P\in \mathbf P$.

\begin{theorem}\label{thm:SD-AsySz}
  Let $\mathbf P$ be the space of distributions that satisfy Assumptions \ref{ass:zhat} to \ref{ass:FiniteDisc}, and let $\mathbf P_{0}$ be as in \eqref{eq:P_0}. Let $\alpha\in(0,1)$ be given, $T$ be the KS test statistic in \eqref{eq:ks-stat}, and $\phi(\cdot)$ be the test in \eqref{eq:our-test}. Then, 
  \begin{equation} \label{eq:Asysize1}
    \limsup_{n \to \infty} E_{P}[\phi(S_n)]\le \alpha \qquad\text{for all }P \in \mathbf P_{0}~.
  \end{equation}
\end{theorem}

Theorem \ref{thm:SD-AsySz} establishes the asymptotic validity of the test in \eqref{eq:our-test}. There are three main reasons why the inequality in \eqref{eq:Asysize1} may be strict, resulting in the limiting rejection probability strictly below $\alpha$. First, when the inequality in \eqref{eq:H0} is strict, the test may reject with probability below $\alpha$, with the magnitude depending on the distance of $P$ from the least-favorable boundary of $\mathbf P_0$. Second, in cases where the support of $Y$ and $X$ take finitely many points, the critical value defined in \eqref{eq:our-cv} may be a strict upper bound for the desired quantile, as discussed further in Section \ref{sec:discrete}. Finally, the test statistic $\Delta_{q_y,q}(u)$ in \eqref{eq:Delta-u} is discretely distributed, taking only a limited number of distinct values. Consequently, the achieved significance level
\begin{equation}\label{eq:bar-alpha}
    \bar{\alpha} := P\left\lbrace \sup_{u \in (0,1)} \Delta_{q_y,q}(u) > c_{\alpha}(q_y, q) \right\rbrace
\end{equation}
satisfies $\bar{\alpha} \leq \alpha$ by definition, but may be strictly less than $\alpha$. Whether $\bar{\alpha} = \alpha$ occurs or not depends on whether the critical value $c_\alpha(q_y, q)$ aligns exactly with one of the discrete jumps in the CDF of $\sup_{u \in (0,1)} \Delta_{q_y,q}(u)$, which depends on $\alpha$, $q_y$, and $q$.

We derive Theorem~\ref{thm:SD-AsySz} by linking the weak convergence of the induced order statistics \( S_n \) to the limit variable \( S \) in \eqref{eq:limit-Sn} with the finite-sample validity of the test \( \phi(S) \) in the limit experiment. This connection becomes nontrivial when the data are not continuously distributed. The KS statistic plays a central role in addressing these challenges. First, weak convergence of $S_n$ to $S$ does not by itself imply convergence of the rejection probability, because the test is discontinuous at configurations involving ties. Our proof constructs an almost-sure representation under which $S_n\to S$ and shows that, with probability approaching one, all pairwise order relations, including ties, agree between the two vectors. On this event, the empirical CDF comparisons and therefore $\phi(S_n)$ and $\phi(S)$ coincide. This yields convergence of the rejection probability to that in the limit experiment. Second, the structure of the KS statistic implies that our test controls size in the limit experiment over our class of null distributions---including those that are discrete or mixed---thereby establishing asymptotic validity. By contrast, as shown in Section \ref{app:other-stats} in the appendix, analogous results generally fail when using the one-sided versions of the Cramér--von Mises or Anderson--Darling statistics, which are commonly used in the stochastic dominance testing. Nonetheless, in the special case where the data are continuously distributed, we show in Theorem \ref{thm:SD-AsySz_CvM} that the Cramér--von Mises-based test remains valid.

Two features are worth highlighting. First, the conclusion of Theorem \ref{thm:SD-AsySz} continues to hold when the conditioning point is estimated. The argument relies on only two properties of the neighbor selection underlying $S_n$: that it is measurable with respect to the $Z$-observations, so that conditionally on them the selected outcomes remain independent with conditional CDFs $F_Y(\cdot|Z_i)$ and $F_X(\cdot|Z_i)$; and that the selected $Z$-values converge in probability to $z_0$. Both properties survive when the neighborhoods are centered at an estimator $\hat z_n$: the first because $\hat z_n$ is $\mathcal A$-measurable, as required by Assumption \ref{ass:zhat}, and the second because consistency of $\hat z_n$, combined with Assumption \ref{ass:Z}, places $q$ observations in any fixed neighborhood of $z_0$ with probability approaching one. As a result, replacing $z_0$ by $\hat z_n$ leaves the limiting behavior of the test, and hence its asymptotic validity, unchanged. Second, the fixed-$q$ framework cannot, in general, deliver consistency against all fixed alternatives. As $n\to\infty$, the selected observations become increasingly well localized around $z_0$, but the limiting experiment continues to contain only $q_y+q_x$ observations. Hence, under a fixed alternative, the rejection probability converges to the power of the test in the corresponding fixed-$q$ ideal experiment, which is generally strictly below one. The relevant sample sizes for power in this framework are therefore $q_y$ and $q_x$, rather than the original sample sizes $n_y$ and $n_x$. Consistency is obtained in the large-$q$ framework of Section~\ref{sec:large-q}; see Theorem~\ref{thm:large-q-general}.

\begin{remark}
\citet{goldman/kaplan:2018} develop an inference framework for the two-sample one-sided hypothesis test in \eqref{eq:H0}, interpreting it as a multiple-testing procedure and with applications to the RDD. As in this section, they adopt an asymptotic framework with fixed $q$ and construct critical values by simulating i.i.d.\ $U(0,1)$ random variables given a test statistic. Our approach to stochastic dominance testing, however, differs from theirs in several important respects. First, we focus on the KS statistic, whereas they advocate for a different statistic based on the so-called Dirichlet approach, which they argue may offer power advantages. This distinction allows us to accommodate discrete or mixed distributions, while their analysis is restricted to continuously distributed data; see Appendix~\ref{app:other-stats} for details.\footnote{Although \citet[][p.~146]{goldman/kaplan:2018} state that their test remains asymptotically valid—albeit conservative—for discrete data, they do not provide a formal proof. By contrast, Theorem~\ref{thm:SD-AsySz} establishes this property for the KS statistic. Interestingly, asymptotic validity does not extend to other test statistics: Appendix~\ref{app:other-stats} shows that it fails for tests based on the Cramér--von Mises and Anderson--Darling statistics.} Second, we also derive results under an asymptotic framework where $q \to \infty$ as $n \to \infty$. This extension allows us to obtain data-dependent rules for choosing the tuning parameters that satisfy the required rate-of-convergence conditions, and we establish the asymptotic validity of our test when these data-dependent rules are used. Third, we allow the conditioning point $z_0$ to be estimated from the data (e.g., a quantile of $Z$) rather than treated as known, and show that this leaves the validity of our test unaffected.
    \end{remark}

\subsection{Asymptotic results for large \emph{q}}\label{sec:large-q}

In this section, we examine the asymptotic properties of the test in \eqref{eq:our-test} within a framework where $q:=q_y+q_x\to\infty$ as $n\to\infty$. Motivated by the practical considerations discussed in Section~\ref{sec:tuning-parameters}, we do so at a level of generality that accommodates two features of how the test is used: the conditioning point may be estimated from the data, and the tuning parameters $(q_y,q_x)$ may be data dependent. The known-point, deterministic-$q$ version of our result then follows as a special case; see Remark~\ref{rem:known-special}. The assumptions in this section are stronger than those in Section~\ref{sec:fixed-q}. This is natural: the weaker assumptions of Section~\ref{sec:fixed-q} ensure convergence to the limit experiment but are silent about its rate, whereas the results below deliver explicit rates and therefore require local smoothness and local mass conditions at $z_0$.

We first introduce the required notation. Let
\begin{equation*}
 \mathcal A := \sigma\big( \{Z_i : 1\le i\le n_y\} \cup \{Z_j : 1\le j\le n_x\} \big)
\end{equation*}
denote the $\sigma$-algebra generated by the $Z$ observations from both samples. Given an estimator $\hat z_n$ of $z_0$ and (possibly random) positive integers $(q_y,q_x)$, define $\hat S_n$ exactly as $S_n$ in \eqref{eq:Sn}, with the $q_y$ and $q_x$ nearest neighbors in each sample selected by their distance to $\hat z_n$ rather than to $z_0$, and with distance ties broken by any $\mathcal A$-measurable rule. For $W\in\{Y,X\}$ and $z$ such that the conditional distribution is well defined, we denote by $P_{W,z}$ the conditional distribution of $W$ given $Z=z$, whose CDF is $F_W(\cdot|z)$; the target laws are $P_{Y,z_0}$ and $P_{X,z_0}$. The ideal local experiment is defined conditionally on $(q_y,q_x)$: let
\begin{equation}\label{eq:ideal-def}
 S^{\circ} := S^{\circ}(q_y,q_x) := (Y^{\circ}_1,\dots,Y^{\circ}_{q_y},X^{\circ}_1,\dots,X^{\circ}_{q_x})~,
\end{equation}
where, conditional on $(q_y,q_x)$, the components are mutually independent and independent of the data, with $Y^{\circ}_i\sim P_{Y,z_0}$ and $X^{\circ}_j\sim P_{X,z_0}$. When $(q_y,q_x)$ is deterministic, $S^{\circ}$ has the same distribution as the limit vector $S$ in Theorem~\ref{thm:limit-Sn}. Since the dimension of $\hat S_n$ is random, the objects we compare are the joint laws of $(q_y,q_x,\hat S_n)$ and $(q_y,q_x,S^{\circ})$; when $(q_y,q_x)$ is deterministic this reduces to comparing $\mathcal L(\hat S_n)$ and $\mathcal L(S^{\circ})$. Finally, for probability measures $P$ and $Q$ on $\mathbf R$ with densities $p$ and $q$ with respect to a common $\sigma$-finite measure $\mu$, the Hellinger distance is defined by $H^2(P,Q):=\frac12\int(\sqrt{p}-\sqrt{q})^2d\mu$, the total variation distance is $\mathrm{TV}(P,Q):=\tfrac12 \int |p-q|~d\mu$ and they satisfy $\mathrm{TV}(P,Q)\le\sqrt{2}\,H(P,Q)$.

We work under the following four assumptions.

\begin{assumption}\label{ass:A-meas}
The estimator $\hat z_n$, the integers $(q_y,q_x)$, the selected index sets, and the tie-breaking rule are $\mathcal A$-measurable. In addition, there exist deterministic sequences $\bar q_y\to\infty$ and $\bar q_x\to\infty$ with $\bar q_y/n_y\to0$ and $\bar q_x/n_x\to0$, and a sequence $\eta_n\downarrow0$, such that
\begin{equation*}
 P\{q_y\le\bar q_y~\text{and}~q_x\le\bar q_x\}\ge 1-\eta_n~.
\end{equation*}
\end{assumption}

\begin{assumption}\label{ass:zhat-rate}
There exist sequences $\delta_n\downarrow0$ and $\varepsilon_n\downarrow0$ such that $P\{|\hat z_n-z_0|>\delta_n\}\le\varepsilon_n$.
\end{assumption}

\begin{assumption}\label{ass:local-mass}
There exist constants $c_Y,c_X>0$ and $r_0>0$ such that, for all $r\in(0,r_0)$ and $W\in\{Y,X\}$,
\begin{equation*}
 P\{|Z-z_0|\le r\}\ge c_W\,r~,
\end{equation*}
where $Z$ denotes the covariate of the sample associated with $W$.
\end{assumption}

\begin{assumption}\label{ass:hellinger}
There exist constants $L_Y,L_X<\infty$ and a neighborhood $\mathcal N$ of $z_0$ such that, for all $z\in\mathcal N$ and $W\in\{Y,X\}$,
\begin{equation*}
 H\big(P_{W,z},P_{W,z_0}\big)\le L_W\,|z-z_0|~.
\end{equation*}
\end{assumption}

The four assumptions play distinct roles. Assumption~\ref{ass:A-meas} is new to this section: it disciplines the data-dependent choices by requiring the estimator $\hat z_n$, the tuning parameters $(q_y,q_x)$, the selected index sets, and the tie-breaking rule to be $\mathcal A$-measurable, and by imposing deterministic envelopes $(\bar q_y,\bar q_x)$ that cap $(q_y,q_x)$ with probability approaching one. It holds in our setting: for deterministic $(q_y,q_x)$ it is satisfied trivially with $\bar q_y=q_y$, $\bar q_x=q_x$, and $\eta_n=0$, while for the data-dependent rule proposed in Section~\ref{sec:tuning-parameters} we verify there that the requirement holds. Assumption~\ref{ass:zhat-rate} is the same requirement on $\hat z_n$ as Assumption~\ref{ass:zhat} in the fixed-$q$ analysis, now stated quantitatively in terms of a rate $\delta_n$ because this rate appears in the finite-sample bound of Theorem~\ref{thm:large-q-general}; it holds with $\delta_n=Mn^{-1/2}$ for any $\sqrt{n}$-consistent estimator, the sample median of the pooled $Z$ observations being the leading example. Finally, Assumptions~\ref{ass:local-mass} and~\ref{ass:hellinger} strengthen Assumptions~\ref{ass:Z} and~\ref{ass:cond-continuous}, respectively: local denseness of $Z$ at $z_0$ is strengthened to mass proportional to the window length, and sup-norm continuity of the conditional CDFs is strengthened to a local Hellinger--Lipschitz condition. Assumption~\ref{ass:hellinger} holds, in particular, whenever the conditional densities of $W$ given $Z=z$ are differentiable in quadratic mean at $z_0$ with finite Fisher information; see \citet{bugni/canay/kim:26b}. Both strengthened conditions accommodate $z_0$ at the boundary of the support of $Z$.

\begin{theorem}\label{thm:large-q-general}
Let Assumptions~\ref{ass:A-meas}--\ref{ass:hellinger} hold. Then there exist constants $C<\infty$ and $c>0$ such that,
\begin{equation}\label{eq:tv-general}
 \mathrm{TV}\Big( \mathcal L\big(q_y,q_x,\hat S_n\big),\, \mathcal L\big(q_y,q_x,S^{\circ}\big) \Big) \;\le\; C\left( \frac{\bar q_y^{3/2}}{n_y}+\frac{\bar q_x^{3/2}}{n_x} \right) \;+\; C\left( \sqrt{\bar q_y}+\sqrt{\bar q_x} \right)\delta_n \;+\; {\rm Rem}_n ~,
\end{equation}
where ${\rm Rem}_n := \varepsilon_n+\eta_n+e^{-c\,\bar q_y}+e^{-c\,\bar q_x}$. Consequently, 
\begin{equation}\label{eq:size-general}
 E_P[\phi(\hat S_n)] \;\le\; \alpha + C\left( \frac{\bar q_y^{3/2}}{n_y}+\frac{\bar q_x^{3/2}}{n_x} \right) \;+\; C\left( \sqrt{\bar q_y}+\sqrt{\bar q_x} \right)\delta_n \;+\; {\rm Rem}_n \quad\text{for all } P\in\mathbf P_0~.
\end{equation}
In particular, the test is asymptotically level $\alpha$ whenever
\begin{equation}\label{eq:growth}
 \bar q_y=o(n_y^{2/3})~,\quad \bar q_x=o(n_x^{2/3})~,\quad \sqrt{\bar q_y}\,\delta_n\to0~,\quad \sqrt{\bar q_x}\,\delta_n\to0~.
\end{equation}
If, in addition, $\min\{q_y,q_x\}\overset{p}{\to}\infty$, then $\lim_{n\to\infty}E_{P'}[\phi(\hat S_n)]=1$ for any $P'\in\mathbf P_1$.
\end{theorem}


In \eqref{eq:size-general}, $\phi(\hat S_n)$ denotes the test in \eqref{eq:our-test} computed with the realized pair $(q_y,q_x)$; that is, the critical value is $c_\alpha(q_y,q)$ with $q=q_y+q_x$. The three terms in the bound \eqref{eq:tv-general} have distinct interpretations. The first, $C\big(\bar q_y^{3/2}/n_y+\bar q_x^{3/2}/n_x\big)$, is the cost of approximating the local experiment by the ideal one at a \emph{known} conditioning point, and coincides with the rates for induced order statistics in \citet{bugni/canay/kim:26b} under quadratic mean differentiability. The second, $C\big(\sqrt{\bar q_y}+\sqrt{\bar q_x}\big)\delta_n$, is the price of \emph{estimating} the conditioning point. The last, ${\rm Rem}_n$, collects the low-probability events on which the localization of $\hat z_n$ or the envelopes for $(q_y,q_x)$ fail, together with exponential terms that vanish as $\bar q_y,\bar q_x\to\infty$; under Assumptions~\ref{ass:A-meas} and~\ref{ass:zhat-rate} it vanishes automatically, so the growth conditions in \eqref{eq:growth} are what remains to ensure the bound tends to zero.

\begin{remark}\label{rem:known-special}
When $\hat z_n=z_0$ and $(q_y,q_x)$ are deterministic, the estimation term $C\big(\sqrt{\bar q_y}+\sqrt{\bar q_x}\big)\delta_n$ vanishes ($\delta_n=0$) and the failure probabilities satisfy $\varepsilon_n=\eta_n=0$, so the bound \eqref{eq:tv-general} reduces to
\begin{equation}\label{eq:tv-known}
 \mathrm{TV}\big(\mathcal L(S_n),\mathcal L(S)\big) \;\le\; C\left(\frac{q_y^{3/2}}{n_y}+\frac{q_x^{3/2}}{n_x}\right)+e^{-cq_y}+e^{-cq_x}~,
\end{equation}
so the test is asymptotically level $\alpha$ provided $q_y=o(n_y^{2/3})$ and $q_x=o(n_x^{2/3})$. These are the growth conditions we rely on in Section~\ref{sec:tuning-parameters}. 
\end{remark}

\begin{remark}\label{rem:proof-method}
The proof of Theorem~\ref{thm:large-q-general} does not rely on the conditional representations for nearest neighbors at a fixed center that underlie the rates in \citet{bugni/canay/kim:26b}; those arguments require the center to be nonrandom relative to the sample used to select neighbors. 
Instead, we use that, conditional on $\mathcal A$, the selected outcomes are independent with conditional laws $P_{W,Z_i}$, regardless of how $\hat z_n$ and $(q_y,q_x)$ were constructed from the $Z$ observations. This approach is well suited under quadratic-mean-differentiability, where the pointwise rate in Assumption~\ref{ass:hellinger} is sharp; see \citet{bugni/canay/kim:26b}.
\end{remark}

\begin{remark}
In a framework where $q \to \infty$, one could alternatively define a test using a critical value based on the $1-\alpha$ quantile of the limiting distribution of the (scaled) KS statistic in \eqref{eq:ks-stat}. Denote this limiting critical value by $c_{\alpha}(\infty)$. Numerical evaluations show that our scaled critical value is smaller than the limiting critical value: $\sqrt{q_yq_x/q}\,c_{\alpha}(q_y,q) \le c_{\alpha}(\infty)$ for $\alpha \in \{0.1, 0.05, 0.01\}$ and $(q_y,q_x) \in \{10,20,\ldots,500\}^2$, with strict inequality in most cases. For example, when $\alpha=5\%$ and $(q_y,q_x) = (70,70)$---a value consistent with our simulations---we obtain $\sqrt{q_yq_x/q}\,c_{\alpha}(q_y,q) = 1.1832 < 1.2239 = c_{\alpha}(\infty)$. The smaller critical value associated with our test translates into an improvement in power of roughly $4\%$ under a local-power approximation for location-shift alternatives. We omit the details for brevity, but they are available upon request.
\end{remark}

\section{Discussion and extensions}\label{sec:extensions}

\subsection{Data-dependent choice of tuning parameters}\label{sec:tuning-parameters}
We now discuss the practical considerations for implementing our test and propose a data-dependent rule for choosing the two tuning parameters $q_y$ and $q_x$. Our goal is to provide practical guidance based on the data, rather than claiming optimality of any sort. The rule is constructed in two steps that mirror the structure of our asymptotic analysis: the rates of convergence in Section~\ref{sec:large-q} determine how fast $(q_y,q_x)$ may grow with $(n_y,n_x)$, and a bias--standard deviation calculation for the estimators of the conditional CDFs, in the spirit of \citet{armstrong/kolesar:18} and similar to the approach in \citet{bugni/canay:21}, determines the constants multiplying these rates. We examine the performance of this rule via Monte Carlo simulations in Section~\ref{sec:simulations} and use it in the empirical application in Section~\ref{app:empirical-application}.

We start with the rate, which is pinned down by results that do not involve any tuning constants. By \eqref{eq:size-general}, the size distortion of the feasible test obeys a finite-sample bound in which the restriction on the tuning parameters is governed by the coupling term $q_y^{3/2}/n_y + q_x^{3/2}/n_x$; for a $\sqrt{n}$-consistent conditioning point, the growth conditions in \eqref{eq:growth} reduce to $q_y = o(n_y^{2/3})$ and $q_x = o(n_x^{2/3})$. This rate is sharp at conditioning points on the boundary of the support of $Z$ by the results in \cite{bugni/canay/kim:26b} (the relevant case in regression discontinuity applications), so no faster bound is available for the class of distributions we allow for. Any rule of the form $q_y \propto n_y^{\gamma}$ with $\gamma < 2/3$ is thus consistent with asymptotic validity, and the only remaining question is which exponent to use within this range. To choose one, we balance two approximation errors that move in opposite directions as $q_y$ grows. Taking $q_y$ larger worsens the coupling between $\hat S_n$ and the ideal experiment $S^{\circ}$, contributing a size distortion of order $q_y^{3/2}/n_y$. Taking $q_y$ smaller degrades the experiment itself: $S^{\circ}$ consists of $q_y$ i.i.d.\ draws from $F_Y(\cdot|z_0)$, so the Gaussian approximation to its rejection probability under local alternatives carries an error of order $q_y^{-1/2}$; see \citet{bugni/canay/kim:26b} for this decomposition in a general setting. The sum $q_y^{3/2}/n_y + q_y^{-1/2}$ is minimized at $q_y \propto n_y^{1/2}$, and we accordingly set $\gamma = 1/2$. This exponent lies strictly inside the range required for validity.

Having fixed the rate, we write the rule as $q_y^* = c_y\, n_y^{1/2}$ and calibrate the constant $c_y$ through the bias and standard deviation of the implicit estimator of $F_Y(\cdot|z_0)$. Assume that the random variable $Z$ is continuous with a density function $f_Z(\cdot)$ satisfying
\begin{equation*}
 |f_Z(z_1) - f_Z(z_2)| \le C_Z |z_1 - z_2| \quad \text{for a Lipschitz constant } C_Z < \infty~,
\end{equation*}
and any values $z_1, z_2 \in \mathcal Z$. In addition, suppose that the conditional CDF of $Y$ satisfies
\begin{equation*}
 \left| \frac{\partial F_Y(t|z)}{\partial z} \right| \le C_Y \quad \text{for a constant } C_Y < \infty~,
\end{equation*}
and that the conditional CDF of $X$ satisfies the same condition with a constant $C_X$. It can be shown that the standardized bias $B_{n_y,q_y}$ associated with the estimator of the conditional CDF $F_Y(\cdot|z_0)$ satisfies
\begin{equation}\label{eq:rot-bias}
 |B_{n_y,q_y}| \;\le\; \frac{q_y^{3/2}}{n_y}\cdot \frac{F_Y(t|z_0)\, C_Z + f_Z(z_0)\, C_Y}{4 f_Z^2(z_0)}~.
\end{equation}
Here, $B_{n_y,q_y}$ is the bias multiplied by $\sqrt{q_y}$, so it is expressed on the same scale as the sampling fluctuations of the empirical CDFs. The two terms in the numerator capture the two sources of local bias: the curvature of the density of $Z$ near $z_0$, and the variation of the conditional CDF in $z$, scaled by the density at the target point. For calibration, we compare this standardized bias with the sampling variability of the two-sample CDF difference that underlies the test. Under the least favorable configuration $F_Y(\cdot|z_0) = F_X(\cdot|z_0) =: F$, $\Delta_{q_y,q}(u)$ in \eqref{eq:Delta-u} has variance $F(1-F)(1/q_y + 1/q_x)$. Under the balanced convention $q_x \approx q_y$, its standard deviation on the same $\sqrt q_y$ scale is $\sqrt{2F(1-F)}$. We calibrate the constant by the simplest bias-sd trade-off: requiring the bias to be no larger than one standard deviation, i.e. $|B_{n_y,q_y}| \le \sqrt{2F(1-F)}$. Solving for $q_y$ yields the benchmark
\begin{equation}\label{eq:rot-benchmark}
 \tilde q_y \;=\; n_y^{2/3}\, \big(2F(1-F)\big)^{1/3} \left( \frac{4 f_Z^2(z_0)}{F\, C_Z + f_Z(z_0)\, C_Y} \right)^{2/3}~.
\end{equation}
The benchmark $\tilde q_y$ grows exactly at the $n_y^{2/3}$ rate. Our rule therefore keeps the constant in \eqref{eq:rot-benchmark} and combines it with the exponent $\gamma = 1/2$ derived above.

It remains to choose the value of $F$ at which to evaluate \eqref{eq:rot-benchmark}, since the factor $(2F(1-F))^{1/3}$ has no lower bound as $t$ varies over $\mathbf R$. We evaluate it at $F = 1/2$, for two reasons. First, the least favorable configuration within the null hypothesis is $F_Y(\cdot|z_0) = F_X(\cdot|z_0)$, as discussed in Section~\ref{sec:permutations}, so the relevant comparison is between two estimators of a common conditional CDF. Second, under this configuration the two-sample empirical process $\Delta_{q_y,q}(u)$ in \eqref{eq:Delta-u} has variance proportional to $u(1-u)$, which is maximal at $u = 1/2$: the KS supremum concentrates around the median, so the median is where the comparison between bias and sampling variability is relevant for the behavior of the test. Setting $F = 1/2$ in \eqref{eq:rot-benchmark} and imposing the $n_y^{1/2}$ rate, we obtain
\begin{equation}\label{eq:rot-general}
 q_y^* \;=\; n_y^{1/2} \left(\frac{1}{2}\right)^{1/3} \left( \frac{4 f_Z^2(z_0)}{\tfrac{1}{2} C_Z + f_Z(z_0)\, C_Y} \right)^{2/3}~.
\end{equation}

The constants $C_Z$ and $C_Y$ cannot be chosen adaptively: it is impossible to estimate them in a way that delivers uniformly valid inference for testing \eqref{eq:H0} (see, e.g., \citealp{armstrong/kolesar:18}). Since any data-dependent rule for $q_y$ must take a stand on $C_Z$ and $C_Y$, we prioritize simplicity and compute them under a bivariate normal working model for $(Y,Z)$, which gives
\begin{equation*}
 C_Z = \frac{1}{\sigma_Z^2 \sqrt{2\pi e}}~, \qquad C_Y = \frac{|\rho_Y|}{\sigma_Z \sqrt{1-\rho_Y^2}} \frac{1}{\sqrt{2\pi}}~, \qquad f_Z(z_0) = \phi_{\mu_Z,\sigma_Z}(z_0)~,
\end{equation*}
where $\phi_{\mu_Z,\sigma_Z}(\cdot)$ denotes the density of a normal distribution with mean $\mu_Z$ and variance $\sigma_Z^2$, and $\rho_Y$ is the correlation between $Y$ and $Z$. The first constant is the maximal slope of the normal density, attained at $\mu_Z \pm \sigma_Z$; the second is the maximal slope of the conditional CDF in $z$. We propose choosing $q_y$ and $q_x$ using the following data-dependent rules:
\begin{equation}\label{eq:rot-qy}
 q_y^* := n_y^{1/2} \left(\frac{1}{2}\right)^{1/3} \left( \frac{4\, \phi^2_{\hat\mu_Z,\hat\sigma_Z}(\hat z_n)} {\dfrac{1}{2\hat\sigma_Z^2\sqrt{2\pi e}} + \phi_{\hat\mu_Z,\hat\sigma_Z}(\hat z_n)\,\dfrac{|\hat\rho_Y|}{\hat\sigma_Z\sqrt{1-\hat\rho_Y^2}} \dfrac{1}{\sqrt{2\pi}}} \right)^{2/3}
\end{equation}
and
\begin{equation}\label{eq:rot-qx}
 q_x^* := n_x^{1/2} \left(\frac{1}{2}\right)^{1/3} \left( \frac{4\, \phi^2_{\hat\mu_Z,\hat\sigma_Z}(\hat z_n)} {\dfrac{1}{2\hat\sigma_Z^2\sqrt{2\pi e}} + \phi_{\hat\mu_Z,\hat\sigma_Z}(\hat z_n)\,\dfrac{|\hat\rho_X|}{\hat\sigma_Z\sqrt{1-\hat\rho_X^2}} \dfrac{1}{\sqrt{2\pi}}} \right)^{2/3}~,
\end{equation}
where $\hat\mu_Z$ is the sample median of $Z$, $\hat\sigma_Z$ is the interquartile range of $Z$ divided by $2\Phi^{-1}(3/4) \approx 1.349$, and $\hat\rho_Y := \sin(\pi \hat\tau_Y/2)$, with $\hat\tau_Y$ denoting Kendall's rank correlation between $Y$ and $Z$ in the first sample; $\hat\rho_X$ is defined analogously using the second sample. Each estimator is consistent for the corresponding parameter under the working model --- in particular, $\tau = \frac{2}{\pi}\arcsin(\rho)$ for the bivariate normal, so $\sin(\pi\hat\tau/2)$ recovers the correlation exactly --- while being insensitive to heavy tails and outliers, which is in line with treating normality as a reference only.

Two implementation details are worth making explicit. First, we round $q_y^*$ and $q_x^*$ up to the nearest integer and impose a small floor, so that the rule always returns a well-defined, non-degenerate neighborhood. Second, we compute the rank correlations $\hat\rho_Y$ and $\hat\rho_X$ from observations lying outside the test neighborhood, a leave-out rule, rather than from the full sample. This is what makes the selected $(q_y,q_x)$ compatible with the measurability requirement of Theorem~\ref{thm:large-q-general} (Assumption~\ref{ass:A-meas}), so that the large-$q$ validity result applies to the feasible test based on the data-dependent rule. Both samples share the same conditioning point, and the leave-out construction and the precise measurability statement are described in Appendix~\ref{app:simulations}.

The constant multiplying $n_y^{1/2}$ in \eqref{eq:rot-general} is intuitive for two reasons. First, it reflects that 
a density that varies more steeply with $z$, or a steeper conditional CDF in $z$, calls for smaller $q_y$: in such cases, nearby observations provide a poorer approximation of the quantities at $z_0$. Since the maximal slopes are determined by $C_Z$ and $C_Y$, the rule is decreasing in both constants. Second, the rule accounts for low density at $z_0$. When $f_Z(z_0)$ is small, the $q_y$ nearest observations are likely to be farther from $z_0$, again requiring smaller $q_y$.

\begin{remark}
The rules in \eqref{eq:rot-qy} and \eqref{eq:rot-qx} are invariant to affine transformations of $Z$ and, because they depend on the outcomes only through Kendall's rank correlation, to strictly increasing transformations of $Y$ and $X$. The test itself shares both invariance properties, as the nearest-neighbor sets are unchanged by affine transformations of $Z$ and the KS statistic is rank-based. A rule based on the Pearson correlation would instead select different values of $(q_y, q_x)$ for, say, $Y$ and $\log Y$, even though the test treats the two problems identically.
\end{remark}

\subsection{Refined critical value for discrete data}\label{sec:discrete}
While our default critical value in \eqref{eq:our-cv} is valid as long as the distributions of random variables $Y$ and $X$ have at most finitely many discontinuities, it is possible to construct a smaller \textit{refined} critical value when both variables are discretely distributed with a limited number of support points. Our test with the refined critical value still maintains the asymptotic validity, though it comes at the cost of additional computational complexity.

To motivate the refined critical value, let $\mathbf Y$ and $\mathbf X$ denote the support of $Y$ and $X$, respectively. 
Our proof for the asymptotic validity (specifically, Theorem \ref{thm:limit-cont} in the appendix) relies on the following inequalities:
\begin{equation*}
P\left\lbrace \sup_{u\in \mathcal{U}}\Delta_{q_y,q}(u) > c_{\alpha}(q_y,q) \right\rbrace ~\leq ~ P\left\lbrace \sup_{u\in\mathbf (0,1)}\Delta_{q_y,q}(u) > c_{\alpha}(q_y,q) \right\rbrace~\leq ~ \alpha.
\end{equation*}
where $\mathcal{U} := \cup_{t\in \mathbf Y}\{u =  F_X(t|z_{0}) \}$ is the set of values that $F_X(t|z_{0})$ takes as $t$ varies over $\mathbf Y$ and $\Delta_{q_y,q}(u) $ is given in \eqref{eq:Delta-u}. Once we replace the set $ \mathcal{U}$ with the interval $(0,1)$, our default critical value $c_{\alpha}(q_y,q)$, defined as a quantile of $\sup_{u\in(0,1)}\Delta_{q_y,q}(u)$ in \eqref{eq:our-cv}, ensures the probability is bounded below $\alpha.$ However, replacing the set $ \mathcal{U}$ with $(0,1)$ could be unnecessarily conservative when $ \mathcal{U}$ contains only a few points---either because $\mathbf Y$ contains only a few points or because $F_X(t|z_{0})$ takes few distinct values as $t$ varies. In fact, the cardinality of the set $ \mathcal{U}$ is determined by the smaller support size of $Y$ and $X$.

To define our refined critical value, let $r$ denote the smaller support size of $Y$ and $X$, 
\begin{equation*}
    r := \min \{|\mathbf Y|,|\mathbf X|\}~,
\end{equation*}
and let $\mathbf{U}_r$ denote the collection of all ordered $r$-tuples of distinct points in (0,1),
\begin{equation*}
   \mathbf{U}_r := \{(u_1,u_2,\dots,u_r)\in (0,1)^r : u_1 < u_2 < \cdots < u_r \}~.
\end{equation*}
We denote an arbitrary element of $\mathbf{U}_r$ by $\mathcal{U}_r$. Our refined critical value is defined as 
\begin{equation}\label{eq:our-cv-r}
    c^{r}_{\alpha}(q_y,q) := \inf_{x\in\mathbf R} \left\{\inf_{\mathcal{U}_r\in  \mathbf{U}_r }
    P\left\lbrace \sup_{u\in\mathcal{U}_r} \Delta_{q_y,q}(u)  \le x \right\rbrace \ge 1-\alpha \right\}~,
\end{equation} 
with $\Delta_{q_y,q}(u)$ as in \eqref{eq:Delta-u}, and the refined test for the null hypothesis in \eqref{eq:H0} when either $Y$ or $X$  is discretely distributed with a limited number of support points is thus 
\begin{equation}\label{eq:our-test-r}
  \phi^{r}(S_n) := I\{T(S_n)>c^{r}_{\alpha}(q_y,q)\}~.
\end{equation}
Section \ref{app:multiple-target} presents the general version of this test for the case $L>1$.

The power gains of using $c_{\alpha}^{r}(q_y, q)$ over $c_{\alpha}(q_y, q)$ are most pronounced when $r$ is small. Our numerical analysis shows the largest gains for $r \le 10$. Thus, this refinement is most effective for discrete data with a limited number of support points, rather than all discrete settings. Moreover, the computational cost of $c_{\alpha}^{r}(q_y, q)$ increases with $r$, and so as $r$ grows, the gains diminish while the cost rises. 

We propose to compute $c_{\alpha}^{r}(q_y, q)$ numerically, by solving the optimization problem
\begin{equation}\label{eq:c-r-optimization}
    c_{\alpha}^{r}(q_y, q) =  \min \left\{ x\in [c_{\rm lb}, c_{\rm ub}] \cap \mathbf{T}: \inf_{\mathcal{U}_r\in  \mathbf{U}_r }
    P\left\lbrace \max_{u\in\mathcal{U}_r} \Delta_{q_y,q}(u)   \le x \right\rbrace~ \ge 1-\alpha \right\}~,
\end{equation}
where $\mathbf{T}$ is the support of $\Delta_{q_y,q}(u)$ in \eqref{eq:Delta-u}, $c_{\rm ub} = c_{\alpha}(q_y, q)$, and $c_{\rm lb}$ is given by
\begin{equation*}
    c_{\rm lb} := \min_{x\in\mathbf T} \left\{
    P\left\lbrace \max_{u\in\left\{\frac{1}{1+r},\frac{2}{1+r},\dots,\frac{r}{1+r} \right\}} \Delta_{q_y,q}(u)   \le x \right\rbrace \ge 1-\alpha \right\}~.
\end{equation*}
That $c_{\rm ub} = c_{\alpha}(q_y, q)$ is a valid upper bound is unsurprising given the preceding discussion, while $c_{\rm lb}$ is a valid lower bound because $\left\{\frac{1}{1+r},\dots,\frac{r}{1+r}\right\}$ is a specific element of $\mathbf U_{r}$. In our numerical evaluations we often found $c_{\alpha}^{r}(q_y, q) = c_{\rm lb}$, but not always, so the additional optimization in \eqref{eq:c-r-optimization} cannot in general be avoided.

\begin{remark}\label{rem:support}
    The support $\mathbf{T}$ of $\Delta_{q_y,q}(u)$ is a discrete subset of $[-1,1]$ that can be enumerated for modest values of $q$: its cardinality is bounded by $(q_y+1)(q_x+1)$, and it is independent of the realizations of the random variables and of the value $u \in (0,1)$.
\end{remark}

The optimization in \eqref{eq:c-r-optimization} is inexpensive for a given application. The refined critical value depends on the data only through $(q_y,q_x)$ and the support sizes, not on the observed outcome values, so it need not be recomputed as the data change and can be tabulated in advance for the relevant $(q_y,q,r,\alpha)$. The outer search ranges over $[c_{\rm lb},c_{\rm ub}]\cap\mathbf T$, which in our experience contains only a handful of points (recall $|\mathbf T|\le(q_y+1)(q_x+1)$ by Remark~\ref{rem:support}); the cost that grows with the support size $r$ comes instead from the inner minimization over $r$-tuples in $\mathbf U_r$. This is why we recommend the refinement only for small $r$ (say $r\le10$), where the power gains are largest and the computation is light. The burden therefore becomes noticeable only when $c^r_\alpha(q_y,q)$ must be recomputed across many Monte Carlo replications; for a single empirical application it is negligible. Code that computes $c^r_\alpha(q_y,q)$ is provided in the replication package.

\subsection{Properties of our test in the limit experiment}\label{sec:permutations}

In this section, we study the properties of the test $\phi(\cdot)$ in \eqref{eq:our-test} in the limit experiment associated with the asymptotic framework in Section \ref{sec:fixed-q}. By Theorem \ref{thm:limit-Sn}, this test is equivalent to $\phi(S)$, where
\begin{equation}\label{eq:limit-S}
    S = (S_1,\dots, S_q), \quad S_j \sim F_Y(\cdot|z_{0}) \text{ for } j \leq q_y \text{ and } S_j \sim F_X(\cdot|z_{0}) \text{ for } j > q_y~.
\end{equation}
In words, in the limit experiment, we observe one random sample of size $q_y$ from the distribution $F_Y(\cdot|z_{0})$ and the other independent random sample of size $q_x$ from the distribution $F_X(\cdot|z_{0})$. The KS test statistic in \eqref{eq:ks-stat} is a function of $S$, and the critical value $c_{\alpha}(q_y, q)$ remains unchanged. We begin our discussion by focusing on the case where $S$ is \emph{continuously distributed}. 

The finite-sample properties of the two-sample one-sided KS statistic have been extensively studied in the literature of testing equality of two continuous (unconditional) distributions. Early works established that the test statistic's finite-sample distribution is pivotal under the null and developed algorithms for its computation (\cite{gnedenko/korolyuk:1951,korolyuk:1955,blackman:1956,hodges:1958,hajek/sidak:1967}, and \cite{durbin:1973}). Our critical value is obtained from this pivotal finite-sample distribution, despite the different null hypothesis. 

This connection to the literature of testing equality of two distributions arises from the observation that the distribution \(F_Y(\cdot|z_{0}) = F_X(\cdot|z_{0})\) is the least favorable within the set of null distributions \(\mathbf{P}_0\) in \eqref{eq:P_0} satisfying stochastic dominance, a point first made by \cite{lehmann:1951} and later reiterated by \cite{hodges:1958,mcfadden:1989}, and \cite{goldman/kaplan:2018}. We define the subset of continuous distributions in \(\mathbf{P}_0\) that satisfy \(F_Y(\cdot|z_{0}) = F_X(\cdot|z_{0})\) as \(\mathbf{P}^*_0\), and denote a generic element in \(\mathbf{P}^*_0\) by \(P^*\). Although these papers studied the distribution of KS statistic under \(P^*\), they did not provide a formal proof that \(P^*\) determines the size of the test under the null hypothesis in \eqref{eq:H0}. For completeness, we formally state this result in Lemma \ref{lem:limit-exact} below, and provide its proof.

\begin{lemma}\label{lem:limit-exact}
Let \(\mathbf{P}^*_0 \subset \mathbf{P}_0\) be the subset of distributions \(P^*\) of the random variable \(S\) in \eqref{eq:limit-S}, such that for a continuous CDF $F$, \(S_j \sim F\) for all \(j = 1, \ldots, q\). Let \(\phi(\cdot)\) be the test defined in \eqref{eq:our-test}. Then, for any $P^*\in \mathbf P^*_0$, we have
\[
\sup_{P \in \mathbf{P}_0} E_P[\phi(S)] = E_{P^*}[\phi(S)] = \bar{\alpha}~,
\]
where $\bar{\alpha}\le \alpha $ is defined in \eqref{eq:bar-alpha}. 
\end{lemma}

Lemma \ref{lem:limit-exact} shows that when $S \sim P^*\in \mathbf P^*_0$, our test is `exact' in that it achieves the closest possible rejection rate to $\alpha$, defined as $\bar{\alpha}$ in \eqref{eq:bar-alpha}. In this case, we have 
  $$ T(S) \stackrel{d}{=} T(U) \quad \text{where}\quad \{U_j \sim U[0,1]: 1 \leq j \leq q\} \; \text{are i.i.d.}$$
This demonstrates that the analytical (finite-sample) critical value $c_{\alpha}(q_y, q)$ for the one-sided KS test can be accurately approximated by simulating uniform random variables. However, when $S \sim P\in \mathbf P_0$ is such that $P\{S_i = S_{j}: i\not=j\} > 0$, the connection $T(S) \stackrel{d}{=} T(U)$ need not hold. In this case, $c_{\alpha}(q_y, q)$ is no longer the finite-sample \emph{analytical} quantile of $T(S)$, but rather a valid upper bound.

When \(S\) is continuously distributed, the proposed test \(\phi(S)\) is equivalent to a non-randomized permutation test. This connection establishes the validity of permutation tests for testing stochastic dominance. While \cite{hodges:1958} and \cite{mcfadden:1989} suggested that a permutation test could be used for this purpose, they did not provide a formal justification. To the best of our knowledge, our proof of this result is novel.

To formally define a permutation test, we introduce the following notation. Let \(\mathbf{G}\) denote the set of all permutations \(\pi = (\pi(1), \dots, \pi(q))\) of \(\{1, \dots, q\}\). The permuted values of \(S\) are given by
\[
S^{\pi} = (S_{\pi(1)}, \dots, S_{\pi(q)})~.
\]
The (non-randomized) permutation test is then defined as follows:
\begin{align}
    \phi^{\rm p}(S) &:= I\left \{ T(S) > c^{\rm p}_{\alpha}(S) \right \} \notag \\
    c^{\rm p}_{\alpha}(S) &:= \inf_{x \in \mathbf{R}} \left \{ \frac{1}{|\mathbf{G}|} \sum_{\pi \in \mathbf{G}} I\{T(S^{\pi}) \leq x\} \geq 1 - \alpha \right \}~.\label{eq:perm-cv}
\end{align}
It follows from standard arguments (see, e.g., \citet[][Ch.\ 15]{lehmann/romano:05}) that when $S$ is invariant to permutations, i.e., $S \stackrel{d}{=} S^{\pi}$, the randomized version of the test $\phi^{\rm p}(S)$ is exact in finite samples. However, under the null hypothesis in \eqref{eq:H0}, we have that $S \stackrel{d}{\not =} S^{\pi}$ for some $P \in \mathbf{P}_0$, and so invariance (or the so-called randomization hypothesis) fails. Therefore, the traditional finite-sample arguments for validity no longer apply. Alternative arguments that claim validity of permutation tests when invariance does not hold typically require $q \to \infty$; see \cite{chung/romano:13,canay/romano/shaikh:17}, and \cite{bugni/canay/shaikh:18}, among others. In our current setting, where $q$ is fixed, such arguments do not apply.

We contribute to this literature by demonstrating that, when \(S\) is continuously distributed, our test is equivalent to a non-randomized permutation test. The formal statement of this result follows.

\begin{lemma}\label{lem:limit-permutation}
For any random variable $\tilde{S}\in \mathbf{R}^{q}$ with $P\{\tilde{S}_{i}\neq \tilde{S}_{j} : i\not=j\}=1$, we have
\begin{equation}\label{eq:same-cv}
P\{ c_{\alpha }^{\rm p}(\tilde{S}) =c_{\alpha }( q_{y},q)\} =1~.
\end{equation}
Moreover, \eqref{eq:same-cv} no longer holds if $\tilde{S}$ is such that $P\{\tilde{S}_{i}=\tilde{S}_{j}: i\not=j\}>0$.
\end{lemma}

Lemma \ref{lem:limit-permutation} shows that our data-independent critical value in \eqref{eq:our-cv} is equivalent to the critical value of a permutation test when the random variable $\tilde{S}$ has no ties. Lemma \ref{lem:limit-permutation} immediately implies when $S_n$ in \eqref{eq:Sn} is continuously distributed, we have
\begin{equation}\label{eq:ours-same-permutations}
\phi^{\rm p}(S_n) \stackrel{a.s}{=} \phi(S_n)~.
\end{equation}
It follows from \eqref{eq:ours-same-permutations} and Theorem \ref{thm:SD-AsySz} that, when $S$ is continuously distributed, a non-randomized permutation test controls the limiting rejection probability under the null hypothesis in \eqref{eq:H0}. Importantly, this result holds even though invariance does not hold for all $P \in \mathbf{P}_0$.

\begin{remark}\label{rem:xinran}
Lemmas \ref{lem:limit-exact} and \ref{lem:limit-permutation} establish the validity of permutation tests for testing stochastic dominance in finite-sample settings. This result illustrates an instance where permutation tests can provide finite-sample valid inference even in settings where the randomization hypothesis does not hold. The only similar result we are aware of is that of \cite{caughey/etal:2023}, who consider a design-based framework with the null hypothesis $\tau_i \le 0$, where $\tau_i$ represents a unit-level treatment effect. Like our work, their result is valid under the condition that the random variables are continuously distributed or that a random tie-breaking rule is applied to handle ties.
\end{remark}

The results in Lemmas \ref{lem:limit-exact} and \ref{lem:limit-permutation} reveal interesting and novel connections between our test and classical arguments involving finite-sample critical values and permutation tests. However, these results critically depend on the random variable $S$ being continuously distributed and do not extend to cases where $S$ is discretely distributed.

When \( S \) is discrete and ties occur with positive probability, i.e., \( P\{S_i = S_{j}\} > 0: i\not=j \), the finite-sample distribution of the KS test statistic \( T(S) \) depends on the number and location of these ties. A natural approach to handle ties is to redefine the test statistic to randomly break them, effectively making the test a randomized one. While this would allow us to establish an analog of Lemma \ref{lem:limit-permutation} for discrete data, we do not pursue such an extension, as randomized tests are rarely used in practice. Despite our best efforts, we were unable to demonstrate that a permutation test could control size under the null hypothesis of stochastic dominance in \eqref{eq:H0} when $S$ is discrete, without relying on random tie-breaking rules.

\section{Simulations}\label{sec:simulations}

In this section, we evaluate the finite-sample performance of the test in \eqref{eq:our-test} for $L=1$ or the test proposed in Section \ref{app:multiple-target} for $L>1$ through a simulation study. We present a variety of data-generating processes to illustrate both the strengths and potential limitations of our test.

We consider seven distinct designs with four cases, (a) to (d), in each design as follows: 
\begin{itemize}
    \item {\bf Case (a)}: the null hypothesis holds with equality and $L=1$
    \item {\bf Case (b)}: the null hypothesis holds with strict inequality and $L=1$
    \item {\bf Case (c)}: the null hypothesis holds with equality and $L=2$
    \item {\bf Case (d)}: the null hypothesis is violated and $L=1$
\end{itemize}

The first three designs are based on the following location–scale model:
\begin{equation}\label{eq:loc-scale}
Y = \mu_Y(Z) + \sigma_Y(Z) U \quad \text{and} \quad X = \mu_X(Z) + \sigma_X(Z) V~,
\end{equation}
where $U$ and $V$ are random variables with specified distributions, and the conditioning variable $Z$ follows a non-negative Beta$(2,2)$ distribution. Under this location–scale model, whenever $\sigma_Y(z_{\ell}) = \sigma_X(z_{\ell})$, the null hypothesis in \eqref{eq:H0} holds as long as
\begin{equation}\label{eq:mus-fsd}
    \mu_Y(z_{\ell}) \ge \mu_X(z_{\ell})~.
\end{equation}
Design 1 satisfies \eqref{eq:mus-fsd} for all $z\in(0,1)$. Design 2 satisfies \eqref{eq:mus-fsd} at $z_{\ell} = 0.5$, but violates the inequality for $z > 0.5$, which may affect the performance of our test in finite samples due to its reliance on induced order statistics. Design 3 is such that $U$ and $V$ are $U[0,1]$, which guarantees that \eqref{eq:H0} holds even when $\sigma_Y(z_{\ell}) < \sigma_X(z_{\ell})$, provided $\mu_Y(z_{\ell}) - \mu_X(z_{\ell}) \ge \sigma_X(z_{\ell}) - \sigma_Y(z_{\ell})$. 

Design 4 is a slight variation of the location-scale model to accommodate an RDD: 
\begin{equation}\label{eq:rdd-design}
Y = \mu_Y(Z) + \sigma_Y(Z) U \quad \text{if } Z \geq 0, \quad \text{and} \quad X = \mu_X(Z) + \sigma_X(Z) V \quad \text{if } Z < 0~.
\end{equation}
Following \cite{shen/zhang:16}, we take $U$ and $V$ to be independent $N(0,1)$ random variables and $Z \sim 2\mathrm{Beta}(2,2)-1$. Both mean functions share the common baseline
\begin{equation}\label{eq:SZ-mu}
\mu(z) = 0.61 - 0.02 z + 0.06 z^{2} + 0.17 z^{3}~,
\end{equation}
with $\mu_Y(z)$ and $\mu_X(z)$ each equal to $\mu(z)$ or a constant shift of it, as specified in Table~\ref{tab:Simu-Designs-14}.

Design 5 is such that $U$ and $V$ follow log-normal distributions with the bottom 20\% censored. This setup reflects features of wage distributions, which often exhibit a point mass at the minimum wage. Design 6 defines conditional probabilities \( P\{X = k | Z\} \) and \( P\{Y = k | Z\} \) as
$$ P\{X = k | Z\}   = \frac{e^{\theta^x_k(\frac{3}{2}-Z)}}{\sum_{j=1}^3 e^{\theta^x_j(\frac{3}{2}-Z)}} ~~\text{and}~~ 
    P\{Y = k | Z\} = \frac{e^{\theta_k^y(\frac{3}{2}-Z)}}{\sum_{j=1}^3 e^{\theta^y_j(\frac{3}{2}-Z)}} ~\text{for }k=1,2,3, $$
where, for $\mu_Y(z) \in [-1,1]$,
\begin{equation*}
  \theta_1^y = \theta^x_1 - \mu_Y(z), \quad \theta_2^y = \theta^x_2 + \mu_Y(z), \quad\text{and}\quad  \theta_3^y = \theta^x_3~. 
\end{equation*}
The parameter  $ \theta^x_k = (-0.5, -1.5, -2)$ controls the baseline log-odds of each category for $X$, while the factor  $(3/2 - Z)$  introduces a monotonic dependence on  $Z$. Finally, Design 7 considers 
$$ X|Z \sim B\Big([25Z], \frac{1}{2}\Big)\quad \text{and}\quad Y|Z \sim B\Big([25Z] + \mu_Y(Z), \frac{1}{2}\Big)~,$$
where $B(\cdot, \cdot)$ denotes a binomial distribution and $[x]$ represents the nearest integer to $x$.

\begin{figure}[ht!]
\begin{center}
    \begin{adjustwidth}{-20mm}{-20mm} 
        \centering
        \includegraphics[width=0.9\textwidth]{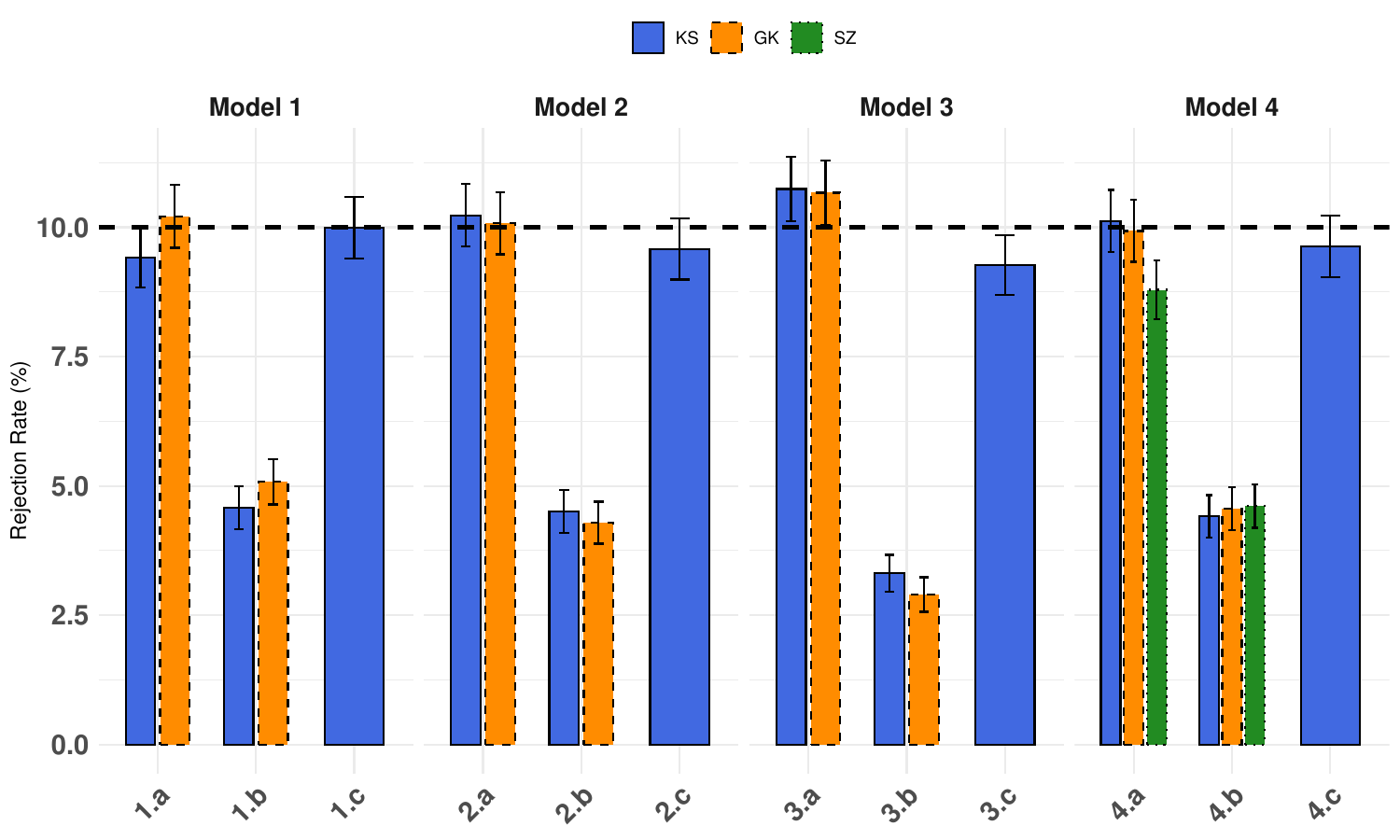} 
    \end{adjustwidth}
    \caption{\small Rejection rates under $H_0$ for Designs 1-4 in cases (a)-(c): $n=1,000$, $\alpha=10\%$, $MC=10,000$.}
    \label{fig:rejection-rates-14}
\end{center}
\end{figure}

\begin{figure}[ht!]
\begin{center}
    \begin{adjustwidth}{-20mm}{-20mm} 
        \centering
        \includegraphics[width=0.8\textwidth]{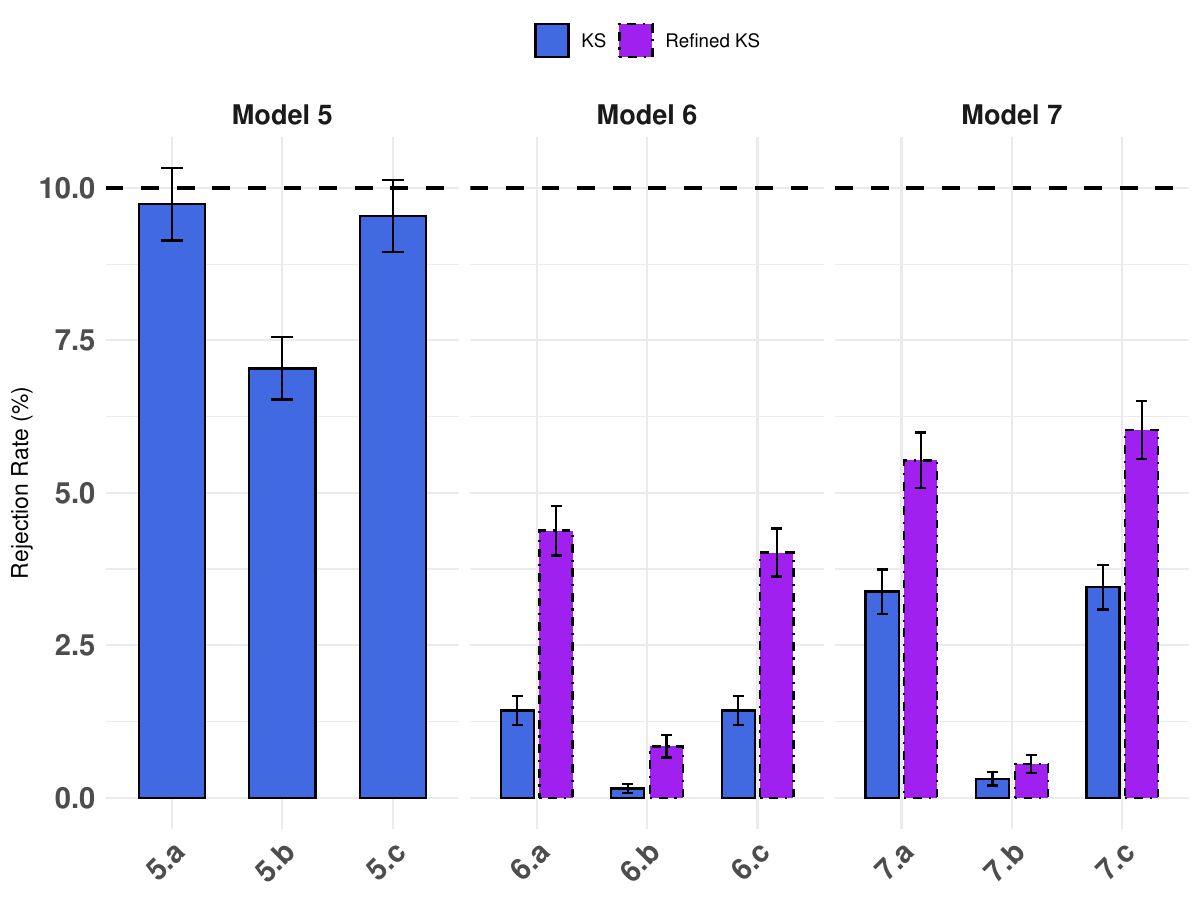} 
    \end{adjustwidth}
    \caption{\small Rejection rates under $H_0$ for Designs 5-7 in cases (a)-(c): $n=1,000$, $\alpha=10\%$, $MC=10,000$.}
    \label{fig:rejection-rates-57}
\end{center}
\end{figure}

The parameter values for all Designs are reported in Appendix \ref{app:simulations}, where we also report the mean values of \(q_y^*\) and \(q_x^*\) across simulations. The results in this section compute the data dependent rule for $q$ using the leave-out rule described in Appendix \ref{app:simulations}. The results based on the alternative rules dicussed in the appendix are quite similar and are not reported here.

We report results for sample size $n = 1,000$ and nominal level $\alpha = 10\%$, based on $10,000$ Monte Carlo simulations to test the null hypothesis in \eqref{eq:H0}. To implement the test $\phi(\cdot)$ in \eqref{eq:our-test}, denoted by `KS' in the figures, we select the tuning parameters $(q_y, q_x)$ using the data-dependent rules in \eqref{eq:rot-qy} and \eqref{eq:rot-qx}. For Designs 1-3 with $L=1$, we compare our test with the one proposed in \citet[][GK]{goldman/kaplan:2018}. For the RDD Design 4 with $L=1$, we compare our test with the method in \citet[][SZ]{shen/zhang:16}.\footnote{SZ is implemented using the authors’ rule of thumb, whereas for GK we use our own rule of thumb, as the original paper does not provide guidance on how to select their tuning parameter. Note that neither of these tests is defined when $L>1$.} For the mixed and discrete designs, we present results for our test in \eqref{eq:our-test} as well as for its refined version in \eqref{eq:our-test-r}.

Figure \ref{fig:rejection-rates-14} reports rejection probabilities under the null hypothesis for the continuously distributed designs. When the data-generating process satisfies the null in \eqref{eq:H0} with equality (case (a)), the KS test’s rejection probabilities closely align with the nominal level and consistently outperform the GK and SZ tests. When the null holds with strict inequality (case (b)), rejection rates fall below \(\alpha\), consistent with our critical value serving as a valid upper bound for the true quantile. Case (c), where \(L = 2\), shows behavior similar to case (a), where \(L = 1\). Overall, the KS test exhibits excellent size control.

\begin{figure}[ht!]
\begin{center}
    \begin{adjustwidth}{-10mm}{-10mm} 
        \centering
        \includegraphics[width=0.9\textwidth]{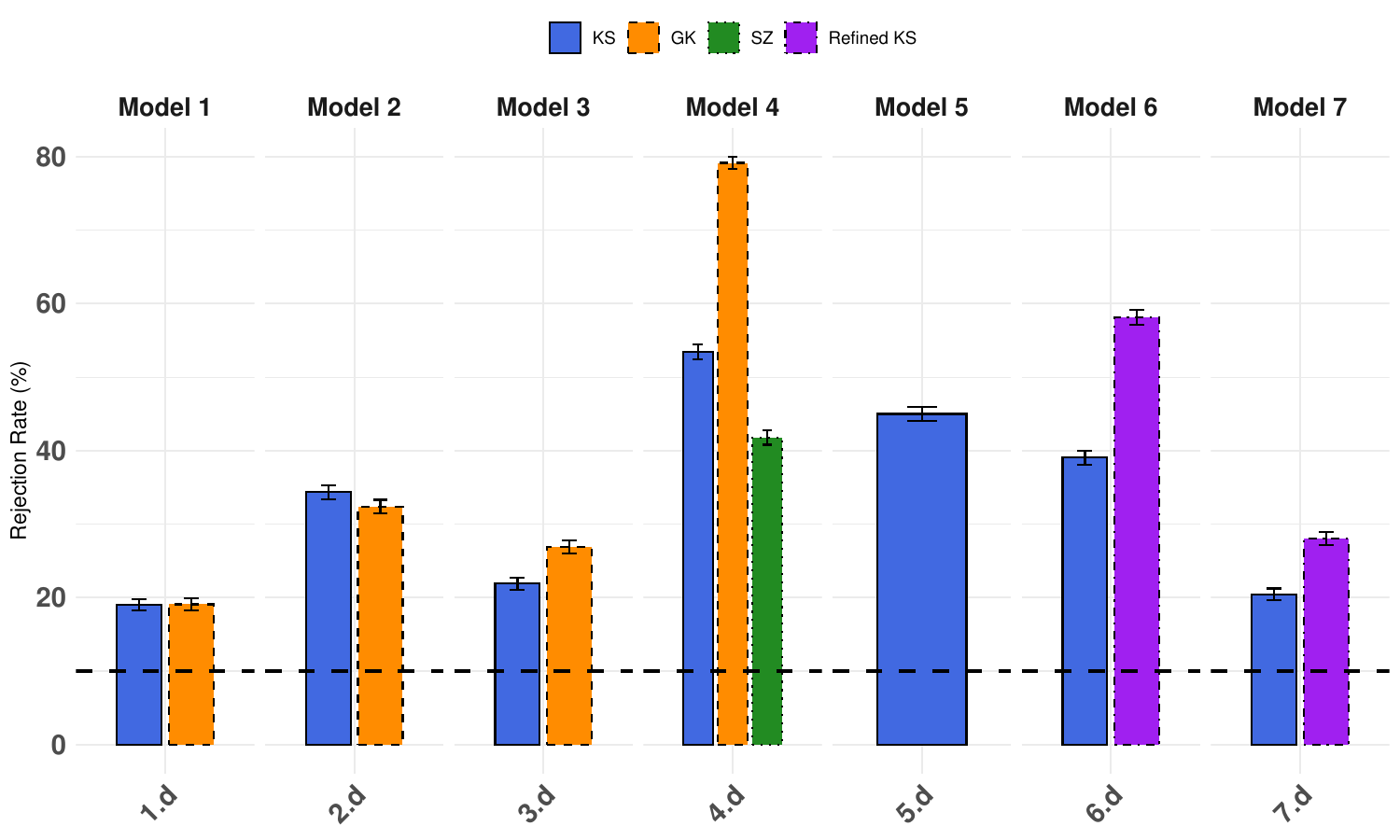} 
    \end{adjustwidth}
    \caption{\small Rejection rates under $H_1$ for Designs 1-7: $n=1,000$, $\alpha=10\%$, $MC=10,000$.}
        \label{fig:power-rates}
\end{center}
\end{figure}

Figure \ref{fig:rejection-rates-57} reports rejection probabilities under the null hypothesis for the mixed and discretely distributed designs. When the data-generating process satisfies the null in \eqref{eq:H0} with equality, the KS test we propose in \eqref{eq:our-test} works well when the data is mixed, but it is conservative when the data is discrete with few support points. Designs 6 and 7 illustrate that the refined critical value \(c^{\rm r}_{\alpha}(q, q_y)\) in \eqref{eq:our-cv-r} offers a more accurate approximation of the true quantile of the KS test statistic, though it may still be somewhat conservative. 

Figure \ref{fig:power-rates} reports rejection probabilities under the alternative hypothesis for all designs. The power of the KS test is similar to that of GK and could be above, below, or roughly the same. The power of the KS test is much higher than that of SZ for this design. This is worth noting given that the KS test uses about $90$ observations (for both $q_y$ and $q_x$) while the SZ test uses $276$ effective observations on each side of the threshold. The power results also demonstrate that the refined critical value \(c^{\rm r}_{\alpha}(q, q_y)\) in \eqref{eq:our-cv-r} enhances power in discrete cases. 

\section{Concluding remarks}\label{sec:conclusion}

This paper introduces a novel test for conditional stochastic dominance (CSD) at target points, offering a flexible, nonparametric approach that avoids kernel smoothing while ensuring computational efficiency. By leveraging induced order statistics, our method constructs empirical CDFs using observations closest to the target conditioning point. 
We establish asymptotic validity of our test under two frameworks: one in which the number of neighbors is held fixed, where the induced order statistics converge to independent draws from the conditional distributions at the target point; and one in which it grows with the sample size, where we obtain an explicit rate that accommodates both an estimated target point and a data-dependent choice of the number of neighbors. Additionally, we extend our framework to better handle discrete data, proposing a refined critical value that enhances the power of the test with minimal additional information. Monte Carlo simulations align with our theoretical results and suggest that our test performs well in finite samples, making our test readily applicable to empirical research in economics, finance, and public policy.

An important feature of our test is its simplicity. Once the key tuning parameters are computed, the test only requires a standard test statistic with a deterministic critical value, without the need for kernels, local polynomials, bias correction, bandwidth selection or resampling techniques such as the bootstrap. Furthermore, our test admits a clear interpretation in the limit experiment, which allows us to connect it with classical analytical critical values and permutation-based tests. In this sense, our findings contribute to the broader literature on stochastic dominance testing by refining conditional inference methods and establishing new links between permutation-based and rank-based approaches. One open question we did not address in this paper concerns the validity of permutation-based tests for the hypothesis of stochastic dominance when both random variables, $Y$ and $X$, are discrete. Despite attempts to formalize this result, we were unable to prove or disprove it. Extensive Monte Carlo simulations (not reported here) suggest that the test may be valid, and this is an area we plan to explore further. 


\appendices
\renewcommand{\theequation}{\Alph{section}-\arabic{equation}}
\setcounter{equation}{0}
\small 

\section{Proof of the main results}\label{app:main-proofs}


\subsection{Proof of Theorem \ref{thm:SD-AsySz}}
By Theorem \ref{thm:limit-Sn}, 
\begin{equation*}
S_{n}\overset{d}{\to }S=(S_{1},\cdots ,S_{q})
\end{equation*}
where the elements of $S$ are independent, and $S_{j}\sim F_{Y}(\cdot |z_{0})$ for $j=1,\cdots ,q_{y}$ and $S_{j}\sim F_{X}(\cdot |z_{0})$ for $ j=q_{y}+1,\cdots
,q$. By the almost-sure representation theorem, we have a sequence of random vectors $\{\Tilde{S}_{n}: 1\le i\le \infty \}$ and a random vector $\Tilde{S}$ defined on a
common probability space $(\Omega ,\mathcal{A },\Tilde{P})$ such that
\begin{equation*}
\Tilde{S}_{n}\overset{d}{=}S_{n},~\Tilde{S}\overset{d}{=}S,\text{ and } \Tilde{S}_{n}\overset{a.s.}{\to }\Tilde{S}~.
\end{equation*}
Define the vector of pairwise weak-order relations of $q$,  $R:\mathbf R^q\mapsto \{0,1\}^{q^2}$ such that $$R(s):=\big(I\{s_i\leq s_j\}: 1\leq i,j\leq q\big)~.$$
Define the event that the rank of the two vectors coincides as follows,
\begin{equation*}
    F_{n}=\{R(\tilde{S}_{n})=R(\tilde{S})\}~.
\end{equation*}
To reach the conclusion, it suffices to show that
\begin{equation}\label{eq:rank-equival_2}
\tilde{P}\{F_{n}\}\to 1~.  
\end{equation}
To see this, consider the following argument,
\begin{align*}
E_{P}[\phi (S_{n})]& \overset{(1)}{=} E_{\tilde{P}}[\phi (\tilde{S}_{n})]
 \overset{(2)}{=} E_{\tilde{P}}[\phi (\tilde{S})I\{F_{n}\}+\phi (\tilde{S} _{n})I\{F_{n}^{c}\}] 
 \overset{(3)}{\to} E_{\tilde{P}}[\phi (\tilde{S})]  \overset{(4)}{\leq} \alpha~,
\end{align*}
where (1) holds by $\Tilde{S}_{n}\overset{d}{=}S_{n}$, (2) because $R(\tilde{S}_n)=R(\tilde{S})$ implies $T(\tilde{S}_n)=T(\tilde{S})$, (3) by \eqref{eq:rank-equival_2} and $\phi(\cdot)\in \{0,1\}$, and (4) by $P \in {\bf P}_0$ and Theorem \ref{thm:limit-cont}.

We devote the remainder of the proof to establishing \eqref{eq:rank-equival_2}. Define \(\mathcal{D} = \mathcal{D}_{X}(z_{0}) \cup \mathcal{D}_{Y}(z_{0})\), where \(\mathcal{D}_{X}(z_{0})\) and \(\mathcal{D}_{Y}(z_{0})\) denote the sets of discontinuity points as specified in Assumption \ref{ass:FiniteDisc}. For any \(\varepsilon > 0\), let  
\begin{align*}
    E_{1}(\varepsilon) &:= \big\{ \{ |\tilde{S}_{i} - \tilde{S}_{j}| > \varepsilon \} \cup \{ \tilde{S}_{i} = \tilde{S}_{j} \in \mathcal{D} \} : i \neq j = 1, \dots, q \big\}, \\
    E_{n,2}(\varepsilon) & := \big\{ |\tilde{S}_{n,k} - \tilde{S}_{k}| < \varepsilon /2 : k = 1, \dots, q \big\}, \\
    E_{n,3} & := \big\{ \{ \tilde{S}_{k} \in \mathcal{D} \} \subseteq \{ \tilde{S}_{k} = \tilde{S}_{n,k} \} : k = 1, \dots, q \big\}~.
\end{align*}
Observe that  
\begin{align}\label{eq:rank-equival_3}
    E_{1}(\varepsilon) \cap E_{n,2}(\varepsilon) \cap E_{n,3} \subseteq F_{n}. 
\end{align}
To establish this, consider the following argument. For any \(i,j = 1, \dots, q\), there are three possible cases:  
(i) \(\tilde{S}_{i} < \tilde{S}_{j}\),  
(ii) \(\tilde{S}_{i} > \tilde{S}_{j}\), or  
(iii) \(\tilde{S}_{i} = \tilde{S}_{j}\). 
First, consider case (i), where \(\tilde{S}_{i} < \tilde{S}_{j}\). Under \(E_{1}(\varepsilon)\), this implies \(\tilde{S}_{i} < \tilde{S}_{j} - \varepsilon\). Under \(E_{n,2}(\varepsilon)\), we have \(\tilde{S}_{n,i} - \varepsilon/2 < \tilde{S}_{i}\) and \(\tilde{S}_{j} < \tilde{S}_{n,j} + \varepsilon/2\). Combining these inequalities yields \(\tilde{S}_{n,i} < \tilde{S}_{n,j}\), as required. Case (ii) follows identically by reversing the roles of \(i\) and \(j\). Finally, consider case (iii), where \(\tilde{S}_{i} = \tilde{S}_{j}\). Under \(E_{1}(\varepsilon)\), this implies \(\tilde{S}_{i} = \tilde{S}_{j} \in \mathcal{D}\). By \(E_{n,3}\), it follows that \(\tilde{S}_{n,i} = \tilde{S}_{n,j} \in \mathcal{D}\). Since this argument holds for all \(i, j = 1, \dots, q\), we conclude that \(F_{n}\) follows, as desired.

By \eqref{eq:rank-equival_3}, \eqref{eq:rank-equival_2} follows that there exits $\varepsilon >0$ such that
\begin{equation}
\tilde{P}\{E_{1}(\varepsilon) \cap E_{n,2}(\varepsilon) \cap E_{n,3}\}\to 1. \label{eq:rank-equival_4}
\end{equation}

For arbitrary $\delta >0$, \eqref{eq:rank-equival_4} follows from finding $ \varepsilon =\varepsilon (\delta )>0$ and $N(\delta) $ such that $\tilde{P}\{ E_{1}(\varepsilon )\cap E_{n,2}(\varepsilon )\cap E_{n,3}\} \geq 1-\delta $ for all $n\geq N(\delta) $. Let $\varepsilon_{1}=\inf \{ \Vert \tilde{d}-d\Vert /2:d<\tilde{ d}\in \mathcal{D}\} >0$. By Lemma \ref{lem:aux_lemma4}, $\exists \varepsilon _{2}>0$ such that, for $i\not=j=1,\ldots ,q$, 
\begin{align}
    \tilde{P}\{\{ |{ \tilde{S}}_{i}-{\tilde{S}}_{j}|<\varepsilon _{2}\} \cap \{ {\tilde{S}}_{i},{\tilde{S}}_{j}\in \mathcal{D} ^{c}\} \} &<\delta /(9q(q-1))~, \label{eq:delta-bound-1} \\
    \tilde{P}\{\cup _{d\in \mathcal{D}}\{|{\tilde{S}}_{i}-{d}|< \varepsilon _{2}\}\cap \{{ \tilde{S}}_{i}\in \mathcal{D}^{c}\mathcal{\}}\}&<\delta /(9q(q-1))~.\label{eq:delta-bound-2}
\end{align}
Finally, set $\varepsilon =\min \{ \varepsilon _{1},\varepsilon _{2}\} >0$ for the remainder of the proof. By elementary arguments, it suffices to show that: (i) $\tilde{P}\{ E_{1}(\varepsilon) ^{c}\} \leq \delta /3$, (ii) $\exists N_{2}(\delta) \in \mathbf{N}$ s.t. $ \tilde{P}\{ E_{n,2}(\varepsilon) ^{c}\} \leq \delta /3$ for all $n\geq N_{2}(\delta) $, and (iii)\ $\exists N_{3}(\delta) \in \mathbf{N}$ s.t. $\tilde{P}\{ E_{n,3}(\varepsilon) ^{c}\} \leq \delta /3$ for all $ n\geq N_{3}( \delta
) $. We divide the rest of the proof into three results.


First, we show that $\tilde{P}\{E_{1}(\varepsilon)^{c}\}\le\delta/3$. Pick $i\neq j$
arbitrarily, and set $B_{k}:=\big(\cup_{d\in\mathcal{D}}\{|\tilde{S}_{k}-d|<\varepsilon\}\big)\cap\{\tilde{S}_{k}\in\mathcal{D}^{c}\}$.
On $\{\tilde{S}_{i}=\tilde{S}_{j}\in\mathcal{D}\}^{c}$, if $\tilde{S}_{i},\tilde{S}_{j}\in\mathcal{D}$
they must differ, so $|\tilde{S}_{i}-\tilde{S}_{j}|\ge2\varepsilon_{1}\ge\varepsilon$ and the event
$\{|\tilde{S}_{i}-\tilde{S}_{j}|<\varepsilon\}$ is empty; hence
\begin{equation*}
\{|\tilde{S}_{i}-\tilde{S}_{j}|<\varepsilon\}\cap\{\tilde{S}_{i}=\tilde{S}_{j}\in\mathcal{D}\}^{c}
\subseteq
\big(\{|\tilde{S}_{i}-\tilde{S}_{j}|<\varepsilon\}\cap\{\tilde{S}_{i},\tilde{S}_{j}\in\mathcal{D}^{c}\}\big)\cup B_{i}\cup B_{j}~.
\end{equation*}
Each of the three sets has probability at most $\delta/(9q(q-1))$ by \eqref{eq:delta-bound-1} and \eqref{eq:delta-bound-2}, both of which follow from Lemma \ref{lem:aux_lemma4} and require only independence (not identical
distributions) so they apply to cross-sample pairs $i\le q_{y}<j$. Summing the resulting bound
$\delta/(3q(q-1))$ over $i\neq j$ gives $\tilde{P}\{E_{1}(\varepsilon)^{c}\}\le\delta/3$, as desired.

Second, we show that $\exists N_{2}(\delta) \in \mathbf{N}$ such that $\tilde{P}\{ E_{n,2}(\varepsilon) ^{c}\} \leq \delta /3$ for all $n\geq N_{2}(\delta) $. To see this, note that
\begin{equation}
\tilde{P}\{ E_{n,2}(\varepsilon) ^{c}\} \leq \sum_{k=1}^{q}P\{|\tilde{S}_{n,k}-\tilde{S}_{k}|>\varepsilon /2\}~.
\end{equation}
By $\Tilde{S}_{n}\overset{a.s.}{\to }\Tilde{S}$, $\exists N_{2}(\delta) $ such that the right-hand side is less than $ \delta /3$, as desired. Finally, we show that $\exists N_{3}(\delta) \in \mathbf{N}$ such that $\tilde{P}\{ E_{n,3}(\varepsilon) ^{c}\} \leq \delta /3$ for all $n\geq N_{3}(\delta) $. To see this, note that
\begin{align*}
\tilde{P}\{ E_{n,3}(\varepsilon) ^{c}\} 
&=\tilde{P} \{ \exists k=1,\ldots ,q:\{ \{ \tilde{S}_{k}\in \mathcal{D}\} \subseteq \{ \tilde{S}_{k}=\tilde{S} _{n,k}\} \} ^{c}\} \\
&\leq \sum_{k=1}^{q}\tilde{P}\{ \{ \tilde{S}_{k}\in \mathcal{D}\} \cap\{ \tilde{S}_{k}\not=\tilde{S}_{n,k}\} \}~.
\end{align*}
By Lemma \ref{lemma:equal-value-at-limit-S}, $\exists N_{3}(\delta) $ such that the right-hand side is less than $\delta /3$, as desired. This completes the proof of \eqref{eq:rank-equival_2} and the theorem.

\subsection{Proof of Theorem \ref{thm:large-q-general}}
We focus on the case $L=1$ for simplicity. Let
\begin{equation*}
 \mathbb S:=\bigsqcup_{(a,b)\in\mathbf N^2}\{(a,b)\}\times\mathbf R^{a+b}~,
 \qquad \mathbf N:=\{1,2,\dots\}~,
\end{equation*}
be equipped with the sigma-algebra
\begin{equation*}
 \mathcal B:=\left\{B\subseteq\mathbb S: B_{a,b}:=\{s\in\mathbf R^{a+b}:((a,b),s)\in B\}\in\mathcal B(\mathbf R^{a+b})
 \text{ for every }(a,b)\in\mathbf N^2\right\}~,
\end{equation*}
where $\mathcal B(\mathbf R^{a+b})$ denotes the Borel sigma-algebra of $\mathbf R^{a+b}$. A function
$\psi:\mathbb S\to\mathbf R$ is $\mathcal B$-measurable if and only if the section $s\mapsto\psi((a,b),s)$ is
Borel measurable for every $(a,b)\in\mathbf N^2$; in particular, the test in \eqref{eq:our-test},
viewed as the map $((a,b),s)\mapsto I\{T(s)>c_\alpha(a,a+b)\}$, is $\mathcal B$-measurable.

For each $(a,b)\in\mathbf N^2$ with $a\le n_y$ and $b\le n_x$, let $\hat M_y(a)$ denote the index set
of the $a$ nearest neighbors of $\hat z_n$ in the first sample and $\hat M_x(b)$ that of the $b$
nearest neighbors of $\hat z_n$ in the second sample, and let $\hat S_n^{a,b}$ denote the vector of
induced order statistics in \eqref{eq:Sn} formed from $\hat M_y(a)$ and $\hat M_x(b)$; for
all remaining $(a,b)\in\mathbf N^2$, set $\hat S_n^{a,b}:=0\in\mathbf R^{a+b}$, a choice that is immaterial in
what follows. The sets $\hat M_y(a)$ and $\hat M_x(b)$ are nested in $a$ and $b$ and, being
functions of the $Z$ observations, $\hat z_n$, and the tie-breaking rule alone, they are
$\mathcal A$-measurable for every fixed $(a,b)$ by Assumption \ref{ass:A-meas}; this is the property
used in Step 1 below. Similarly, let $\{S^{\circ}(a,b):(a,b)\in\mathbf N^2\}$ be a family of random
vectors, independent of the data, whose components are mutually independent with
$Y^{\circ}_i\sim P_{Y,z_0}$ for $i\le a$ and $X^{\circ}_j\sim P_{X,z_0}$ for $j\le b$. With
these definitions, $\hat S_n=\hat S_n^{q_y,q_x}$, and $S^{\circ}=S^{\circ}(q_y,q_x)$ by \eqref{eq:ideal-def}. We also
record two conventions. First, $1\le q_y\le n_y$ and $1\le q_x\le n_x$ almost surely, as is
automatic from the construction of the test; this is used when writing finite sums below.
Second, since each $\mathbf R^{a+b}$ is a standard Borel space, regular conditional distributions of
$\hat S_n^{a,b}$ given $\mathcal A$ exist, and their total variation distance to any fixed probability
measure can be computed as a supremum over a countable field generating $\mathcal B(\mathbf R^{a+b})$; all
conditional total variation and Hellinger quantities appearing below are therefore well
defined and $\mathcal A$-measurable.

We proceed in four steps: a conditional product representation at fixed $(a,b)$, a bound on
the conditional Hellinger distance on a suitable event, the removal of the conditioning, and
the implications for size and power.

\medskip
\noindent{\bf Step 1:} Fix $(a,b)\in\mathbf N^2$ with $a\le n_y$ and $b\le n_x$. By Assumption
\ref{ass:A-meas} and the preliminaries, conditionally on $\mathcal A$, the quantities $\hat z_n$, the
index sets $(\hat M_y(a),\hat M_x(b))$, and the ordering of the selected indices are
deterministic. Since observations are independently drawn within each sample, and the two samples are mutually independent, and each outcome is
independent of all remaining $Z$ observations conditional on its own, it follows that
\begin{equation}\label{eq:cond-product}
 \mathcal L\big(\hat S_n^{a,b}\,\big|\,\mathcal A\big)
 \;=\;\bigotimes_{i\in\hat M_y(a)}P_{Y,Z_i}\;\otimes\;\bigotimes_{j\in\hat M_x(b)}P_{X,Z_j}~,
\end{equation}
which repeats the conditional representation used in the proof of Theorem \ref{thm:limit-Sn},
the only difference being that the selection is governed by $\hat z_n$; the $\mathcal A$-measurability
imposed in Assumption \ref{ass:A-meas} is what guarantees that this difference is immaterial.
Since the family $\{S^{\circ}(a,b)\}$ is independent of the data,
\begin{equation}\label{eq:cond-ideal}
 \mathcal L\big(S^{\circ}(a,b)\,\big|\,\mathcal A\big)
 \;=\;\big(P_{Y,z_0}\big)^{a}\otimes\big(P_{X,z_0}\big)^{b}~.
\end{equation}
Next, recall that the Hellinger affinity $1-H^2$ is multiplicative across independent
components, so for product measures with components indexed by $\ell$,
\begin{equation*}
 H^2\Big(\bigotimes_\ell P_\ell,\bigotimes_\ell Q_\ell\Big)
 =1-\prod_\ell\big(1-H^2(P_\ell,Q_\ell)\big)
 \overset{(1)}{\le}\sum_\ell H^2(P_\ell,Q_\ell)~,
\end{equation*}
where (1) follows by induction from $1-(1-u)(1-v)\le u+v$ for $u,v\in[0,1]$. Applying this to \eqref{eq:cond-product} and \eqref{eq:cond-ideal} yields, for every $(a,b)$
with $a\le n_y$ and $b\le n_x$,
\begin{equation}\label{eq:hell-sum}
 H^2\Big( \mathcal L\big(\hat S_n^{a,b}|\mathcal A\big),\,
 \big(P_{Y,z_0}\big)^{a}\otimes\big(P_{X,z_0}\big)^{b} \Big)
 \;\le\;\sum_{i\in\hat M_y(a)}H^2\big(P_{Y,Z_i},P_{Y,z_0}\big)
 +\sum_{j\in\hat M_x(b)}H^2\big(P_{X,Z_j},P_{X,z_0}\big)~.
\end{equation}

\medskip
\noindent{\bf Step 2:} Define the radii
\begin{equation}\label{eq:radii}
 r_y:=\frac{2}{c_Y}\cdot\frac{\bar q_y}{n_y}~,
 \qquad
 r_x:=\frac{2}{c_X}\cdot\frac{\bar q_x}{n_x}~,
\end{equation}
and the counts
\begin{equation*}
 N_y:=\sum_{i=1}^{n_y}I\{|Z_i-z_0|\le r_y\}~,
 \qquad
 N_x:=\sum_{j=1}^{n_x}I\{|Z_j-z_0|\le r_x\}~.
\end{equation*}
Since $\bar q_y/n_y\to0$, we have $r_y<r_0$ for all $n$ large enough with $r_0>0$ as in Assumption \ref{ass:local-mass}. Thus, for all $n$ large enough, Assumption \ref{ass:local-mass} yields $p_y:=P\{|Z_i-z_0|\le r_y\}\ge c_Y r_y=2\bar q_y/n_y$ and hence $E[N_y]=n_yp_y\ge2\bar q_y$. As $N_y$ is a sum of $n_y$ i.i.d.\ indicators, the multiplicative Chernoff bound (see, e.g., \citealp[Theorem 1.1]{dubhashi/panconesi:09}) gives, for all $n$ large enough,
\begin{equation}\label{eq:chernoff}
 P\{N_y<\bar q_y\}\le P\{N_y\le E[N_y]/2\}\le \exp\big(-E[N_y]/8\big)\le e^{-\bar q_y/4}~.
\end{equation}
An analogous argument yields $P\{N_x<\bar q_x\}\le e^{-\bar q_x/4}$ for all $n$ large enough. Define the event
\begin{equation}\label{eq:good-event}
 \mathcal{G}_n:=\big\{|\hat z_n-z_0|\le\delta_n\big\}\cap\big\{q_y\le\bar q_y,\,q_x\le\bar q_x\big\}
 \cap\big\{N_y\ge\bar q_y,\,N_x\ge\bar q_x\big\}~,
\end{equation}
and note that $\mathcal{G}_n\in\mathcal A$ and, by Assumptions
\ref{ass:A-meas}--\ref{ass:zhat-rate} and \eqref{eq:chernoff},
\begin{equation}\label{eq:good-prob}
 P\{\mathcal{G}_n^c\}\le\varepsilon_n+\eta_n+e^{-\bar q_y/4}+e^{-\bar q_x/4}~.
\end{equation}
On $\mathcal{G}_n$, at least $\bar q_y$ observations of the first sample lie within $r_y$ of
$z_0$, and hence within $r_y+\delta_n$ of $\hat z_n$; consequently, for every $a\le\bar q_y$, each of
the $a$ nearest neighbors of $\hat z_n$ satisfies $|Z_i-\hat z_n|\le r_y+\delta_n$, and therefore
\begin{equation}\label{eq:selected-radius}
 |Z_i-z_0|\le|Z_i-\hat z_n|+|\hat z_n-z_0|\le r_y+2\delta_n~,
 \qquad i\in\hat M_y(a)~,\quad a\le\bar q_y~.
\end{equation}
An analogous argument yields
\begin{equation}\label{eq:selected-radius2}
|Z_j-z_0|\le r_x+2\delta_n~,
 \qquad j\in\hat M_x(b)~,\quad b\le\bar q_x~.
\end{equation}
For all $n$ large enough, $r_y+2\delta_n$ and $r_x+2\delta_n$ lie inside the neighborhood $\mathcal N$ of Assumption \ref{ass:hellinger}, so, on $\mathcal{G}_n$ and for every $(a,b)$ with $a\le\bar q_y$ and $b\le\bar q_x$, we obtain
\begin{align}
 H^2\Big( \mathcal L\big(\hat S_n^{a,b}|\mathcal A\big),\,
 \big(P_{Y,z_0}\big)^{a}\otimes\big(P_{X,z_0}\big)^{b} \Big)
 &\;\overset{(1)}{\le}\; L_Y^2\,a\,(r_y+2\delta_n)^2+L_X^2\,b\,(r_x+2\delta_n)^2\notag\\
 &\;\overset{(2)}{\le}\; C_1\left\{a\left(\frac{\bar q_y}{n_y}+\delta_n\right)^{2}
 +b\left(\frac{\bar q_x}{n_x}+\delta_n\right)^{2}\right\}~,\label{eq:cond-hell-good}
\end{align}
where $C_1$ depends only on $(L_Y,L_X,c_Y,c_X)$, (1) holds by \eqref{eq:hell-sum}, \eqref{eq:selected-radius}, \eqref{eq:selected-radius2}, and Assumption \ref{ass:hellinger}, and (2) by \eqref{eq:radii}. By increasing the constant $C_1$ if necessary, \eqref{eq:cond-hell-good} can be extended to all $n$.

\medskip
\noindent{\bf Step 3:} Since the events $\{q_y=a,q_x=b\}$ partition the sample space (by the
convention $1\le q_y\le n_y$, $1\le q_x\le n_x$) and $\hat S_n=\hat S_n^{a,b}$ on $\{q_y=a,q_x=b\}$, for
every $B\in\mathcal B$,
\begin{equation*}
 I\{((q_y,q_x),\hat S_n)\in B\}
 =\sum_{a\le n_y}\sum_{b\le n_x} I\{q_y=a,q_x=b\}\,I\{\hat S_n^{a,b}\in B_{a,b}\}~,
\end{equation*}
and the same identity holds with $S^{\circ}(a,b)$ in place of $\hat S_n^{a,b}$. Therefore,
\begin{align}
 P\{((q_y,q_x),\hat S_n)\in B\mid\mathcal A\}
 &=\sum_{a\le n_y}\sum_{b\le n_x} I\{q_y=a,q_x=b\}\,
 P\{\hat S_n^{a,b}\in B_{a,b}\mid\mathcal A\}~,\label{eq:cond-prob-hat}\\
 P\{((q_y,q_x),S^{\circ})\in B\mid\mathcal A\}
 &=\sum_{a\le n_y}\sum_{b\le n_x} I\{q_y=a,q_x=b\}\,
 \big(P_{Y,z_0}\big)^{a}\otimes\big(P_{X,z_0}\big)^{b}(B_{a,b})~,\label{eq:cond-prob-S}
\end{align}
where \eqref{eq:cond-prob-hat} and \eqref{eq:cond-prob-S} use that $I\{q_y=a,q_x=b\}$ is $\mathcal A$-measurable and that the sums have finitely many terms, and the latter also uses \eqref{eq:cond-ideal}. Letting
\begin{equation*}
 \Delta_{a,b}(B_{a,b}):=P\big\{\hat S_n^{a,b}\in B_{a,b}\mid\mathcal A\big\}
 -\big(P_{Y,z_0}\big)^{a}\otimes\big(P_{X,z_0}\big)^{b}(B_{a,b})~,
\end{equation*}
we obtain
\begin{align}
 \mathrm{TV}\Big(\mathcal L\big(q_y,q_x,\hat S_n\big),&\mathcal L\big(q_y,q_x,S^{\circ}\big)\Big)
 \overset{(1)}{=}\sup_{B\in\mathcal B}\Big\vert
 E\Big[P\big\{((q_y,q_x),\hat S_n)\in B\mid\mathcal A\big\}
 -P\big\{((q_y,q_x),S^{\circ})\in B\mid\mathcal A\big\}\Big]\Big\vert\notag\\
 &\overset{(2)}{=}\sup_{B\in\mathcal B}\Big\vert
 E\Big[\sum_{a\le n_y}\sum_{b\le n_x}I\{q_y=a,q_x=b\}\,
 \Delta_{a,b}(B_{a,b})\Big]\Big\vert\notag\\
 &\overset{(3)}{\le}
 E\Bigg[\sup_{B\in\mathcal B}\Big\vert\sum_{a\le n_y}\sum_{b\le n_x}I\{q_y=a,q_x=b\}\,
 \Delta_{a,b}(B_{a,b})\Big\vert\Bigg]\notag\\
 &\overset{(4)}{=}
 E\Bigg[\sum_{a\le n_y}\sum_{b\le n_x}I\{q_y=a,q_x=b\}\,
 \mathrm{TV}\Big(\mathcal L\big(\hat S_n^{a,b}\mid\mathcal A\big),
 \big(P_{Y,z_0}\big)^{a}\otimes\big(P_{X,z_0}\big)^{b}\Big)\Bigg]~,
 \label{eq:decondition}
\end{align}
where (1) holds by the definition of total variation on $(\mathbb S,\mathcal B)$ and the law of iterated expectations, (2) by \eqref{eq:cond-prob-hat} and \eqref{eq:cond-prob-S}, (3) by the triangle inequality and moving the supremum inside the expectation, and (4) because, for each realization, exactly one indicator is nonzero, so the inner supremum equals the total variation distance between the corresponding pair of laws on $\mathbf R^{a+b}$, attained by optimizing $B_{a,b}$ at the realized $(a,b)=(q_y,q_x)$.

Next, split the expectation in \eqref{eq:decondition} over $\mathcal{G}_n$ and its complement:
\begin{align}
 E\Bigg[\sum_{a\le n_y}\sum_{b\le n_x}&I\{q_y=a,q_x=b\}\,
 \mathrm{TV}\Big(\mathcal L\big(\hat S_n^{a,b}\mid\mathcal A\big),
 \big(P_{Y,z_0}\big)^{a}\otimes\big(P_{X,z_0}\big)^{b}\Big)\Bigg]\notag\\
 &\overset{(5)}{\le}
 E\Bigg[I\{\mathcal{G}_n\}\sum_{a\le \bar q_y}\sum_{b\le \bar q_x}I\{q_y=a,q_x=b\}\,
 \sqrt{2}\,H\Big(\mathcal L\big(\hat S_n^{a,b}\mid\mathcal A\big),
 \big(P_{Y,z_0}\big)^{a}\otimes\big(P_{X,z_0}\big)^{b}\Big)\Bigg]
 +P\{\mathcal{G}_n^c\}\notag\\
 &\overset{(6)}{\le}
 \sqrt{2C_1}\left\{\sqrt{\bar q_y}\left(\frac{\bar q_y}{n_y}+\delta_n\right)
 +\sqrt{\bar q_x}\left(\frac{\bar q_x}{n_x}+\delta_n\right)\right\}
 +P\{\mathcal{G}_n^c\}\notag\\
 &\overset{(7)}{\le}
 C\left(\frac{\bar q_y^{3/2}}{n_y}+\frac{\bar q_x^{3/2}}{n_x}
 +\sqrt{\bar q_y}\,\delta_n+\sqrt{\bar q_x}\,\delta_n\right)
 +\varepsilon_n+\eta_n+e^{-\bar q_y/4}+e^{-\bar q_x/4}~,
 \label{eq:tv-final}
\end{align}
where (5) holds because on $\mathcal{G}_n$ the realized pair satisfies $q_y\le\bar q_y$ and
$q_x\le\bar q_x$, $\mathrm{TV}\le\sqrt{2}\,H$, and on $\mathcal{G}_n^c$ the sum is bounded by one, since it contains a single nonzero total variation term, (6) by \eqref{eq:cond-hell-good} applied at the (unique) realized pair $(a,b)=(q_y,q_x)$, which satisfies $a\le\bar q_y$ and $b\le\bar q_x$ on $\mathcal{G}_n$, together with $\sqrt{u+v}\le\sqrt{u}+\sqrt{v}$ and $a\le\bar q_y$, $b\le\bar q_x$, and (7) by \eqref{eq:good-prob} and collecting constants. This establishes \eqref{eq:tv-general} with $c=1/4$.

\medskip
\noindent{\bf Step 4:} The test $\phi$ in \eqref{eq:our-test}, computed
with the realized pair $(q_y,q_x)$, is the indicator of a set in $\mathcal B$. Hence, directly from
the definition of total variation on $(\mathbb S,\mathcal B)$,
\begin{equation}\label{eq:phi-transfer}
 \Big\vert E_P[\phi(\hat S_n)]-E_P[\phi(S^{\circ})] \Big\vert
 \le \mathrm{TV}\Big( \mathcal L\big(q_y,q_x,\hat S_n\big),\mathcal L\big(q_y,q_x,S^{\circ}\big) \Big)~,
\end{equation}
which parallels the argument in \eqref{eq:decondition}. Conditional on $(q_y,q_x)$, the vector
$S^{\circ}$ consists of $q_y$ i.i.d.\ draws from $P_{Y,z_0}$ and $q_x$ i.i.d.\ draws from
$P_{X,z_0}$, mutually independent, and the critical value is $c_\alpha(q_y,q)$. Therefore, for
any $P\in\mathbf P_0$, Theorem \ref{thm:limit-cont} applied at the realized pair $(q_y,q_x)$ yields
\begin{equation}\label{eq:ideal-size}
 E_P\big[\phi(S^{\circ})\,\big|\,q_y,q_x\big]\le\alpha
\end{equation}
Taking expectations in \eqref{eq:ideal-size} and combining with \eqref{eq:phi-transfer} and \eqref{eq:tv-final} proves \eqref{eq:size-general}, and the asymptotic level claim follows under \eqref{eq:growth} since every term on the right-hand side of \eqref{eq:tv-final} then vanishes as $n\to \infty$.

For the power claim, fix $P'\in\mathbf P_1$, and let $t^*$, $\varepsilon$, and $\kappa_\alpha$ be as in Lemma \ref{lem:ideal-power}. Since the family $\{S^{\circ}(a,b)\}$ is independent of the data and $(q_y,q_x)$ is $\mathcal A$-measurable,
\begin{equation*}
 E_{P'}\big[\phi(S^{\circ})\,\big|\,q_y,q_x\big]=h(q_y,q_x)~,
 \qquad\text{where}\qquad
 h(a,b):=E_{P'}\big[\phi(S^{\circ}(a,b))\big]~.
\end{equation*}
Fix any $m\in\mathbf N$ with $2\kappa_\alpha/\sqrt{m}<\varepsilon/3$. By Lemma \ref{lem:ideal-power}(ii), $h(a,b)\ge 1-2e^{-2m\varepsilon^2/9}$ whenever $\min\{a,b\}\ge m$, and therefore
\begin{equation*}
 E_{P'}[\phi(S^{\circ})]
 \;\ge\; E\big[h(q_y,q_x)\,I\{\min\{q_y,q_x\}\ge m\}\big]
 \;\ge\;\big(1-2e^{-2m\varepsilon^2/9}\big)\,P\{\min\{q_y,q_x\}\ge m\}~.
\end{equation*}
Since $\min\{q_y,q_x\}\overset{p}{\to}\infty$ by hypothesis, $P\{\min\{q_y,q_x\}\ge m\}\to1$ for every fixed $m$, so $\liminf_{n\to\infty}E_{P'}[\phi(S^{\circ})]\ge 1-2e^{-2m\varepsilon^2/9}$. Letting $m\to\infty$ yields $E_{P'}[\phi(S^{\circ})]\to1$. Combining with \eqref{eq:phi-transfer} and \eqref{eq:tv-final}, whose right-hand side vanishes under \eqref{eq:growth}, gives $E_{P'}[\phi(\hat S_n)]\to1$, as desired.

\subsection{Proof of Lemma \ref{lem:limit-exact}}

Note that
\begin{equation*}
E_{P^{\ast }}[\phi (S)]~\overset{(1)}{\leq }~\sup_{P\in \mathbf{P}_{0}}E_{P}[\phi (S)]~\overset{(2)}{\leq }~\bar{\alpha},
\end{equation*}
where (1) holds by $P^{\ast }\in \mathbf{P}_{0}$ and (2) by Theorem \ref{thm:limit-cont}. To complete the proof, it suffices to show that $E_{P^{\ast }}[\phi (S)]=\bar{\alpha}$. To this end, consider the following argument: 
\begin{align}
 &E_{P^{\ast }}[\phi (S)]\notag\\
 &~\overset{(1)}{=}~P^*\left\{ \sup_{t\in \mathbf{R} }\left( \frac{1}{q_{y}}\sum_{j=1}^{q_{y}}I\{S_{j}\leq t\}-\frac{1}{q_{x}} \sum_{j=q_{y}+1}^{q}I\{S_{j}\leq t\}\right) >c_{\alpha }(q_{y},q)\right\}  \notag\\
&~\overset{(2)}{=}~P^*\left\{ \sup_{t\in \mathbf{R}}\left( \frac{1}{q_{y}} \sum_{j=1}^{q_{y}}I\{U_{j}\leq F(t)\}-\frac{1}{q_{x}}\sum_{j=q_{y}+1}^{q}\{U_{j}\leq F(t)\}\right) >c_{\alpha }(q_{y},q)\right\}   \notag\\
&~\overset{(3)}{=}~P^*\left\{ \sup_{u\in (0,1)}\left( \frac{1}{q_{y}} \sum_{j=1}^{q_{y}}I\{U_{j}\leq u\}-\frac{1}{q_{x}}\sum_{j=q_{y}+1}^{q}I \{U_{j}\leq u\}\right) >c_{\alpha }(q_{y},q)\right\}  \notag\\
&~ \overset{(4)}{=}~\bar{\alpha}~,\label{eq:limit-exact_0}
\end{align}
where (1) holds by \eqref{eq:ks-stat}, (2) holds by \citet[][Eq.\ 36]{pollard:02} and the same arguments used in the proof of Theorem \ref{thm:limit-cont}, (3) follows from the continuity of $F$ guaranteeing that 
\begin{equation*}
 \sup_{t\in \mathbf{R}} \Delta_{q_y,q}\left(F(t)\right) = \sup_{u\in (0,1)} \Delta_{q_y,q}(u)
\end{equation*} 
for $ \Delta_{q_y,q}(u) := \frac{1}{q_y}\sum_{j=1}^{q_y} I\{U_{j} \le u\}-\frac{1}{q-q_y}\sum_{j=q_y+1}^{q} I\{U_{j} \le u\}$ as defined in \eqref{eq:Delta-u}, and (4) by definition of $\bar{\alpha}$ in \eqref{eq:bar-alpha}.

\subsection{Proof of Lemma \ref{lem:limit-permutation}}

Let $Q:=\{Q_{1},\dots ,Q_{q}\}=\{1,2,\dots ,q\}$ and denote by $Q^{\pi }:=\{Q_{\pi (1)},Q_{\pi (2)}\dots ,Q_{\pi (q)}\}$ the permutation $\pi =(\pi (1),\pi (2),...,\pi \left( q\right) )$ of $Q$. Let $R(s)$ denote the rank of $s$, which maps $s$ to a permutation of $Q$. Since the KS statistic $T(\cdot )$ in \eqref{eq:ks-stat} is a rank statistic, it follows that for any $s$
\begin{equation}
T(s)~=~T^{\ast }(R(s))~,  \label{eq:limit-permutation1}
\end{equation}%
where $T^{\ast }$ is a known function; see \citet[page 99]{hajek/sidak/sen:99}. That is, the KS test statistic depends on $S$ only through $R(S)$. Define 
\begin{equation}
c_{\alpha }^{\mathrm{p}}~:=~\inf_{x\in \mathbf{R}}\left\{ \frac{1}{ \left\vert \mathbf{G}\right\vert }\sum_{\pi \in \mathbf{G}}I\left\{ T(Q^{\pi })\leq x\right\} \geq 1-\alpha \right\} ~.  \label{eq:limit-permutation2}
\end{equation}

We divide the rest of the argument into four steps.

\noindent \underline{Step 1.} For any $s\in \mathbf{R}^{q}$ with $s_{i}\not=s_{j} $ for $i\not=j$, and $c_{\alpha }^{\mathrm{p}}(s)$ as in \eqref{eq:perm-cv},
\begin{equation*}
c_{\alpha }^{\mathrm{p}}(s)~=~c_{\alpha }^{\mathrm{p}}~.
\end{equation*}
To establish this, consider the following derivation,
\begin{align}
c_{\alpha }^{\mathrm{p}}(s)& ~=~\inf_{x\in \mathbf{R}}\left\{ \frac{1}{ \left\vert \mathbf{G}\right\vert }\sum_{\pi \in \mathbf{G}}I\left\{ T(s^{\pi })\leq x\right\} \geq 1-\alpha \right\}   \notag \\
& ~\overset{(1)}{=}~\inf_{x\in \mathbf{R}}\left\{ \frac{1}{\left\vert \mathbf{G}\right\vert }\sum_{\pi \in \mathbf{G}}I\left\{ T^{\ast }(R(s^{\pi }))\leq x\right\} \geq 1-\alpha \right\}   \notag \\
& ~\overset{(2)}{=}~\inf_{x\in \mathbf{R}}\left\{ \frac{1}{\left\vert \mathbf{G}\right\vert }\sum_{\pi \in \mathbf{G}}I\left\{ T^{\ast }((Q^{\bar{ \pi}})^{\pi })\leq x\right\} \geq 1-\alpha \right\}   \notag \\
& ~\overset{(3)}{=}~\inf_{x\in \mathbf{R}}\left\{ \frac{1}{\left\vert \mathbf{G}\right\vert }\sum_{{\pi }\in \mathbf{G}}I\left\{ T^{\ast }(Q^{{\pi }})\leq x\right\} \geq 1-\alpha \right\}\notag \\
&~\overset{(4)}{=}~c_{\alpha }^{\mathrm{p}}~.  \label{eq:limit-permutation3}
\end{align}
Here, (1) follows by \eqref{eq:limit-permutation1}, (2) follows since $ s_{i}\not=s_{j}$ for $i\not=j$ implies that $R(s^{\pi })=(R(s))^{\pi }$ and $ R(s)=Q^{\bar{\pi}(s)}$ for some $\bar{\pi}\left( s\right) \in \mathbf{G}$, (3) by $\mathbf{G}=\tilde{\mathbf{G}}:=\{\pi \circ \bar{\pi}(s):\pi \in \mathbf{G}\}$, which follows from the fact that $\mathbf{G}$ is a group, and (4) by \eqref{eq:limit-permutation2}.

\noindent \underline{Step 2.} $c_{\alpha }^{\mathrm{p}}=c_{\alpha }\left( q_{y},q\right) $, where $c_{\alpha }\left( q_{y},q\right)$ is defined in \eqref{eq:our-cv}.

Let $\{U_i : 1\le i\le q\}$ be i.i.d.\ with $U_{i}\sim U\left( 0,1\right) $ and let $\hat{\pi}$ be a uniformly chosen permutation from $ \mathbf{G}$, independent of $U$. Note that $c_{\alpha }\left( q_{y},q\right) $ is the $\left( 1-\alpha \right) $-quantile of $T(U)$ and, by \eqref{eq:limit-permutation2}, $ c_{\alpha }^{\mathrm{p}}$ is the $\left( 1-\alpha \right) $-quantile of the CDF $\frac{1}{\left\vert \mathbf{G}\right\vert }\sum_{\pi \in \mathbf{G} }I\left\{ T^{\ast }(Q^{\pi })\leq x\right\} $. The desired result then follows from noting that $\frac{1}{\left\vert \mathbf{G}\right\vert }\sum_{\pi \in \mathbf{G}}I\left\{ T^{\ast }(Q^{\pi })\leq x\right\} $ is the CDF of $T(U)$, as we show next.

Let $E:=\left\{ U_{i}\not=U_{j}\text{ for }i\not=j\right\} $. For any $x\in \mathbf{R}$, our desired result follows from this derivation:
\begin{align*}
P\left\{ T(U)\leq x\right\} & ~\overset{(1)}{=}~P\{T(U^{\hat{\pi}})\leq x\} \\
& ~\overset{(2)}{=}~P\left\{ \left\{ T(U^{\hat{\pi}})\leq x\right\} \cap E\right\}  \\
& ~\overset{(3)}{=}~\int_{s\in E}\frac{1}{\left\vert \mathbf{G}\right\vert } \sum_{\pi \in \mathbf{G}}I\left\{ T(u^{\pi })\leq x\right\} dP_{U}^{\ast }(u) \\
& ~\overset{(4)}{=}~\int_{u\in E}\frac{1}{\left\vert \mathbf{G}\right\vert } \sum_{\pi \in \mathbf{G}}I\left\{ T^{\ast }(Q^{\pi })\leq x\right\} dP_{U}^{\ast }(u) \\
& ~\overset{(5)}{=}~\frac{1}{\left\vert \mathbf{G}\right\vert }\sum_{\pi \in \mathbf{G}}I\left\{ T^{\ast }(Q^{\pi })\leq x\right\} ~.
\end{align*}%
Here, (1) holds by $U\overset{d}{=}U^{\hat{\pi}}$, (2) and (5) by $P\{E\}=1$, (3) by $\hat{\pi}\perp U$ and $\hat{\pi}$ uniformly chosen in $\mathbf{G}$, and (4) by repeating the arguments used to derive \eqref{eq:limit-permutation3}.

\noindent \underline{Step 3.} By Step 1, $\{S_{i}\not=S_{j}$ for $ i\not=j\}\subseteq \left\{ c_{\alpha }^{\mathrm{p}}(S)=c_{\alpha }^{\mathrm{p}}\right\} $, and so $ P\left\{ c_{\alpha }^{\mathrm{p}}(S)=c_{\alpha }^{\mathrm{p}}\right\} =1$ holds by our assumption. By Step 2, $c_{\alpha }^{\mathrm{p}}=c_{\alpha }\left( q_{y},q\right) $. The desired result follows from combining these points.

\noindent \underline{Step 4.} To show the last statement, consider $S=\{1 :1\le j \le q\}$. Then, $T(S^{\pi })=T(S)=0$ for all $\pi \in \mathbf{G}$, and so $c_{\alpha }^{\mathrm{p}}(S)=0$. On the
other hand, Step 2 implies $c_{\alpha }^{\mathrm{p}}=c_{\alpha }\left( q_{y},q\right) $, which are positive for typical values of $\left( q_{y},q,\alpha \right) $. For example, $q_{y}=1$, $q=2$, and $\alpha =0.1$ yield $c_{\alpha }^{\mathrm{p}}=c_{\alpha }\left( q_{y},q\right) =1.$

\section{Supporting technical results}\label{app:auxiliary}
For any $\varepsilon>0$, we use $o_{\varepsilon}(1)$ to denote an expression that converges to zero as $\varepsilon \to 0$. Analogously, for any $n \in \mathbf{N}$, we use $o_{n}(1)$ to denote an expression that converges to zero as $n \to \infty$.

\subsection{Auxiliary theorems}

\begin{theorem}\label{thm:limit-Sn-general}
  Let Assumptions \ref{ass:zhat}, \ref{ass:Z}, and \ref{ass:cond-continuous} hold with $\mathcal{Z}=\{z_{1},z_{2},\dots,z_{L}\}$, and let $\hat S^{\ell}_n$ be defined as in \eqref{eq:Sn_ell} based on $\hat z_{n,\ell}$ instead of $z_{\ell}$, with distance ties broken by any $\mathcal A$-measurable rule. Then,
  \begin{equation}
      (({\hat S_{n}^{1}})^{\prime},\ldots,({\hat S_{n}^{L}})^{\prime})\overset{d}{\rightarrow}
      (({S^{1}})^{\prime},\ldots,({S^{L}})^{\prime})~,
      \label{eq:limit_randomz_L0}
      \end{equation}
  where, for any $(s_{1}^{\ell},\dots,s_{q_y^\ell}^{\ell},s_{q_y^\ell+1}^{\ell},\dots,
  s_{q^\ell}^{\ell})\in\mathbf{R}^{q^\ell}$ for each $\ell=1,\dots,L$ with $q^\ell=q_y^\ell+q_x^\ell$ fixed, $(({S^{1}})^{\prime},\ldots,({S^{L}})^{\prime})$ has the following distribution:
  \begin{equation}
      P\left\{\bigcap_{\ell=1}^{L}\bigcap_{i=1}^{q^{\ell}}\left\{S_{i}^{\ell}\leq s_{i}^{\ell}
      \right\}\right\}~=~\prod_{\ell=1}^{L}\prod_{i=1}^{q_{y}^{\ell}}F_{Y}(s_{i}^{\ell}|z_{\ell})
      \prod_{j=1}^{q_{x}^{\ell}}F_{X}(s_{j+q_{y}^{\ell}}^{\ell}|z_{\ell})~.
      \label{eq:limit-Sn-general_limit_distribution}
  \end{equation}
  \end{theorem}
\begin{proof}
  For each $\ell=1,\dots,L$, let $\hat g_\ell(Z):=|Z-\hat z_{n,\ell}|$, let $M_y^{\ell}$ denote the subset of the indices $i=1,\dots,n_y$ corresponding to the $q_y^{\ell}$ first $\hat g_\ell$-order statistics $(\hat Z^{\ell}_{n_y,(1)},\dots,\hat Z^{\ell}_{n_y,(q_y^{\ell})})$, and let $M_x^{\ell}$ denote the subset of the indices $j=1,\dots,n_x$ corresponding to the $q_x^{\ell}$ first $\hat g_\ell$-order statistics $(\hat Z^{\ell}_{n_x,(1)},\dots,\hat Z^{\ell}_{n_x,(q_x^{\ell})})$. Since each $\hat z_{n,\ell}$ is $\mathcal A$-measurable by Assumption \ref{ass:zhat} and distance ties are broken by an $\mathcal A$-measurable rule, the sets $M_y^{\ell}$, $M_x^{\ell}$ and the orderings within them are $\mathcal A$-measurable. Let $E_n$ denote the following event:
  \begin{equation*}
  E_{n}~=E_{n_y,y}\cap E_{n_x,x}~,~\text{where}~~
  E_{n_y,y}:=\hspace{-6pt}\bigcap_{1\le\ell<\ell'\le L}\hspace{-6pt}
  \left\{M_y^{\ell}\cap M_y^{\ell'}=\emptyset\right\}
  ~\text{and}~
  E_{n_x,x}:=\hspace{-6pt}\bigcap_{1\le\ell<\ell'\le L}\hspace{-6pt}
  \left\{M_x^{\ell}\cap M_x^{\ell'}=\emptyset\right\}.
  \end{equation*}
  In words, $E_n$ means that the subsets of the data used in each of the $L$ tests are pairwise disjoint within each sample. Note that $E_n\in\mathcal A$. We begin by showing that
  \begin{equation}
  P\{E_{n}\}\rightarrow 1~.\label{eq:limit_randomz_L1}
  \end{equation}
  
  Set $\bar\varepsilon:=\frac{1}{4}\min\left\{\left\vert z_{\ell}-z_{\ell'}\right\vert: \ell\not=\ell',~\ell,\ell'=1,\ldots,L\right\}>0$ and fix $\varepsilon\in(0,\bar\varepsilon)$ arbitrarily. Set 
  \[G_{n}:=\left\{\max\left\vert\hat z_{n,\ell}-z_{\ell}\right\vert\leq \varepsilon/2:\ell=1,\ldots,L\right\}~.\]
  For each $\ell=1,\ldots,L$, let $B_{\ell,y}:=\sum_{i=1}^{n_y}I\{|Z_{i}-z_{\ell}|\leq\varepsilon/2\}$ and $\hat B_{\ell,y}:=\sum_{i=1}^{n_y}I\{|Z_{i}-\hat z_{n,\ell}|\leq\varepsilon\}$, and define $B_{\ell,x}$ and $\hat B_{\ell,x}$ analogously for the second sample. Note that
    \begin{equation}
  G_{n}\bigcap\bigcap_{\ell=1}^{L}\{B_{\ell,y}\geq q_{y}^{\ell}\}
  ~\overset{(1)}{\subseteq}~\bigcap_{\ell=1}^{L}\bigcap_{i=1}^{q_y^\ell}
  \{|\hat Z_{n_{y},(i)}^{\ell}-z_{\ell}|\leq 3\varepsilon/2\}
  ~\overset{(2)}{\subseteq}~E_{n_y,y}~,\label{eq:limit_randomz_L2}
  \end{equation}
  where (1) holds by the fact that, for every $\ell=1,\dots,L$, $G_{n}$ and $\{B_{\ell,y}\geq q_{y}^{\ell}\}$ imply that $|\hat z_{n,\ell}-z_{\ell}|\leq \varepsilon/2$ and $|\hat Z_{n_{y},(i)}^{\ell}-\hat z_{n,\ell}|\leq \varepsilon$ for $i=1,\dots,q_y^{\ell}$, and (2) by the fact that for $\ell\not=\ell'$ we have $|z_\ell-z_{\ell'}|\geq 4\bar\varepsilon>3\varepsilon$, so no sample observation can lie within $3\varepsilon/2$ of both $z_\ell$ and $z_{\ell'}$, and therefore $M_y^{\ell}\cap M_y^{\ell'}=\emptyset$ for all $\ell\not=\ell'$, which implies $E_{n_y,y}$.

The same argument applied to the second sample yields the analogue of \eqref{eq:limit_randomz_L2} with $(B_{\ell,x},q_x^\ell,\hat Z_{n_x,(j)}^{\ell},E_{n_x,x})$ in place of $(B_{\ell,y},q_y^\ell,\hat Z_{n_y,(i)}^{\ell},E_{n_y,y})$. From here, we have that \eqref{eq:limit_randomz_L1} follows if we show that
  \begin{align}
  P\{G_{n}^{c}\}&\rightarrow 0~,\label{eq:limit_randomz_L3}\\
  P\{B_{\ell,y}<q_{y}^{\ell}\}\rightarrow 0~\text{ and }~
  P\{B_{\ell,x}<q_{x}^{\ell}\}&\rightarrow 0~\text{ for all }\ell=1,\ldots,L~,
  \label{eq:limit_randomz_L4}
  \end{align}
  since then $P\{E_n^c\}\leq P\{G_n^c\}+\sum_{\ell=1}^{L}\left(P\{B_{\ell,y}<q_y^\ell\} +P\{B_{\ell,x}<q_x^\ell\}\right)\to 0$. Note that this union bound does not rely on the independence of the two samples: $G_n$ may involve the $Z$ observations from both samples, as is the case when $\hat z_{n,\ell}$ is computed from the pooled data.

  Note that \eqref{eq:limit_randomz_L3} holds by Assumption \ref{ass:zhat} and $L$ being finite. For \eqref{eq:limit_randomz_L4}, we only prove the first convergence, as the second is analogous. Since the count $B_{\ell,y}$ is centered at the non-random point $z_\ell$, $B_{\ell,y}\sim Bi(n_y,P\{|Z_{i}-z_{\ell}|\leq\varepsilon/2\})$, where $P\{|Z_{i}-z_{\ell}|\leq\varepsilon/2\}>0$ by Assumption \ref{ass:Z}. By this and $\max_{\ell=1,\ldots,L}q^{\ell}$ being bounded, $\exists N(\varepsilon)$ s.t.\ for all $n_y\geq N(\varepsilon)$,
  \begin{equation}
  0<\frac{1}{2}P\{|Z_{i}-z_{\ell}|\leq\varepsilon/2\}\leq
  P\{|Z_{i}-z_{\ell}|\leq\varepsilon/2\}-\max_{\ell=1,\ldots,L}\frac{q^{\ell}_y}{n_y}~.
  \label{eq:limit_randomz_L4b}
  \end{equation}
  Then, for all $n_y\geq N(\varepsilon)$, we have
  \begin{align*}
  P\{B_{\ell,y}<q_{y}^{\ell}\}
  &=P\left\{{B_{\ell,y}}/{n_y}-P\{|Z_{i}-z_{\ell}|\leq\varepsilon/2\}
  <{q^{\ell}_y}/{n_y}-P\{|Z_{i}-z_{\ell}|\leq\varepsilon/2\}\right\}\\
  &\overset{(1)}{\leq}P\left\{\left\vert{B_{\ell,y}}/{n_y}
  -P\{|Z_{i}-z_{\ell}|\leq\varepsilon/2\}\right\vert
  >\tfrac{1}{2}P\{|Z_{i}-z_{\ell}|\leq\varepsilon/2\}\right\}
  \overset{(2)}{\rightarrow}0~,
  \end{align*}
  as desired, where (1) holds by \eqref{eq:limit_randomz_L4b} and (2) holds by the LLN as $n_y\to\infty$.

  We are now ready to prove \eqref{eq:limit_randomz_L0}. For any $(s_{1}^{\ell},\dots,s_{q_{y}^{\ell}}^{\ell},s_{q_{y}^{\ell}+1}^{\ell},\dots, s_{q^{\ell}}^{\ell})\in\mathbf{R}^{q^{\ell}}$ for each $\ell=1,\dots,L$ with $q^{\ell}=q_{y}^{\ell}+q_{x}^{\ell}$, we have
  \begin{align}
  &P\left\{\bigcap_{\ell=1}^{L}\bigcap_{i=1}^{q^{\ell}}\left\{\hat S_{n,i}^{\ell}\leq
  s_{i}^{\ell}\right\}\right\}\notag\\
  &{\overset{(1)}{=}E\left[\left.P\left\{\bigcap_{\ell=1}^{L}\left\{
  \bigcap_{i=1}^{q_{y}^{\ell}}\left\{\hat Y_{n,\left[i\right]}^{\ell}\leq s_{i}^{\ell}
  \right\}\bigcap\bigcap_{j=1}^{q_{x}^{\ell}}\left\{\hat X_{n,\left[j\right]}^{\ell}\leq
  s_{j+q_{y}^{\ell}}^{\ell}\right\}\right\}\right\vert\mathcal{A}\right\}\right]}\notag\\
  &\overset{(2)}{=}E\left[\left.P\left\{\bigcap_{\ell=1}^{L}\left\{
  \bigcap_{i=1}^{q_{y}^{\ell}}\left\{\hat Y_{n,\left[i\right]}^{\ell}\leq s_{i}^{\ell}
  \right\}\bigcap\bigcap_{j=1}^{q_{x}^{\ell}}\left\{\hat X_{n,\left[j\right]}^{\ell}\leq
  s_{j+q_{y}^{\ell}}^{\ell}\right\}\right\}\right\vert\mathcal{A}\right\}I\{E_{n}\}\right]
  +o_{n}(1)\notag\\
  &\overset{(3)}{=}E\left[\prod_{\ell=1}^{L}\left.P\left\{\left\{
  \bigcap_{i=1}^{q_{y}^{\ell}}\left\{\hat Y_{n,\left[i\right]}^{\ell}\leq s_{i}^{\ell}
  \right\}\bigcap\bigcap_{j=1}^{q_{x}^{\ell}}\left\{\hat X_{n,\left[j\right]}^{\ell}\leq
  s_{j+q_{y}^{\ell}}^{\ell}\right\}\right\}\right\vert\mathcal{A}\right\}I\{E_{n}\}\right]
  +o_{n}(1)\notag\\
  &\overset{(4)}{=}E\left[\prod_{\ell=1}^{L}\left.P\left\{\left\{
  \bigcap_{i=1}^{q_{y}^{\ell}}\left\{\hat Y_{n,\left[i\right]}^{\ell}\leq s_{i}^{\ell}
  \right\}\bigcap\bigcap_{j=1}^{q_{x}^{\ell}}\left\{\hat X_{n,\left[j\right]}^{\ell}\leq
  s_{j+q_{y}^{\ell}}^{\ell}\right\}\right\}\right\vert\mathcal{A}\right\}\right]
  +o_{n}(1)\notag\\
  &\overset{(5)}{=}E\left[\prod_{\ell=1}^{L}\left\{\prod_{i=1}^{q_{y}^{\ell}}
  F_{Y}(s_{i}^{\ell}|\hat Z_{n_{y},\left(i\right)}^{\ell})\prod_{j=1}^{q_{x}^{\ell}}
  F_{X}(s_{j+q_{y}^{\ell}}^{\ell}|\hat Z_{n_{x},\left(j\right)}^{\ell})\right\}\right]+o_{n}(1)~,
  \label{eq:limit_randomz_L5}
  \end{align}
where (1) holds by the LIE, (2) and (4) by \eqref{eq:limit_randomz_L1} and the conditional probability being bounded by one, (3) by Assumption \ref{ass:zhat} and the fact that, conditional on $E_n \in \mathcal{A}$, the $\ell$-th bracket term is a function of $\{Y_i:i\in M_y^{\ell}\}$ and $\{X_j:j\in M_x^{\ell}\}$ only, which are pairwise disjoint observations across $\ell = 1,\dots,L$, and (5) by Assumption \ref{ass:zhat} and the $\mathcal A$-measurable tie-breaking rule, the indices of the selected observations and their ordering by $\hat g_\ell$ are $\mathcal A$-measurable, so conditional on $\mathcal A$ the vector $(\hat Y_{n,[1]}^{\ell},\dots,\hat Y_{n,[q_y^\ell]}^{\ell}, \hat X_{n,[1]}^{\ell},\dots,\hat X_{n,[q_x^\ell]}^{\ell})$ has independent components with conditional conditional CDFs $F_{Y}(\cdot|\hat Z_{n_{y},(i)}^{\ell})$ and $F_{X}(\cdot|\hat Z_{n_{x},(j)}^{\ell})$.

  Next, we show that
  \begin{align}
  \hat Z_{n_{y},\left(i\right)}^{\ell}\overset{p}{\rightarrow}z_{\ell}~\text{ for all }~
  i&=1,\dots,q_{y}^{\ell}\text{ and }~\ell=1,\ldots,L~,\label{eq:limit_randomz_L6}\\
  \hat Z_{n_{x},\left(j\right)}^{\ell}\overset{p}{\rightarrow}z_{\ell}~\text{ for all }~
  j&=1,\dots,q_{x}^{\ell}\text{ and }~\ell=1,\ldots,L~.\label{eq:limit_randomz_L7}
  \end{align}
  We only show \eqref{eq:limit_randomz_L6}, as \eqref{eq:limit_randomz_L7} can be shown analogously. Fix $\varepsilon\in(0,\bar\varepsilon)$ arbitrarily, with $G_n$ and $B_{\ell,y}$ defined as before. By \eqref{eq:limit_randomz_L2}-\eqref{eq:limit_randomz_L4}, we conclude that
  \begin{equation*}
  P\Big\{\bigcap_{\ell=1}^{L}\bigcap_{i=1}^{q_y^\ell}
  \{|\hat Z_{n_{y},\left(i\right)}^{\ell}-z_{\ell}|\leq 3\varepsilon/2\}\Big\}\rightarrow 1~.
  \end{equation*}
  Since the choice of $\varepsilon\in(0,\bar\varepsilon)$ was arbitrary, \eqref{eq:limit_randomz_L6} follows.

  By \eqref{eq:limit_randomz_L6}, \eqref{eq:limit_randomz_L7}, and Assumption \ref{ass:cond-continuous}, $F_{Y}(s_{i}^{\ell}|\hat Z_{n_{y},(i)}^{\ell})\overset{p}{\rightarrow} F_{Y}(s_{i}^{\ell}|z_{\ell})$ and $F_{X}(s_{j+q_{y}^{\ell}}^{\ell}|\hat Z_{n_{x},(j)}^{\ell})\overset{p}{\rightarrow} F_{X}(s_{j+q_{y}^{\ell}}^{\ell}|z_{\ell})$ for every $i$, $j$, $\ell$, and every $(s_{1}^{\ell},\dots,s_{q^{\ell}}^{\ell})$. By this and the fact that the integrand in \eqref{eq:limit_randomz_L5} is bounded by one and converges in probability to a constant, it converges in $L^1$ and so 
  \begin{equation}
  \lim_{n\rightarrow\infty}P\left\{\bigcap_{\ell=1}^{L}\bigcap_{i=1}^{q^{\ell}}\left\{
  \hat S_{n,i}^{\ell}\leq s_{i}^{\ell}\right\}\right\}~
  =~\prod_{\ell=1}^{L}\left\{
  \prod_{i=1}^{q_{y}^{\ell}}F_{Y}(s_{i}^{\ell}|z_{\ell})\prod_{j=1}^{q_{x}^{\ell}}
  F_{X}(s_{j+q_{y}^{\ell}}^{\ell}|z_{\ell})\right\}~.
  \label{eq:limit_randomz_L9}
  \end{equation}
  By definition of convergence in distribution, \eqref{eq:limit_randomz_L9} implies \eqref{eq:limit_randomz_L0}, as desired.
\end{proof}

\begin{theorem}\label{thm:limit-cont}
Let $S$ be the random variable in Theorem \ref{thm:limit-Sn}. Then, for any $P\in \mathbf P_0$ and $\alpha\in(0,1)$, we obtain 
\[
E_P[ \phi (S)] \leq \bar{\alpha}\le \alpha ~,
\]
where 
\begin{equation}
    \bar{\alpha} := P\left\lbrace \sup_{u\in\mathbf (0,1)} \left(\frac{1}{q_y}\sum_{j=1}^{q_y} I\{U_{j} \le u\}-\frac{1}{q_x}\sum_{j=q_y+1}^{q} I\{U_{j} \le u\} \right) > c_{\alpha}(q_y,q) \right\rbrace ~.
\end{equation}
Moreover, the first inequality becomes an equality under $P$ such that  $F_Y(t|z_{0})=F_X(t|z_{0})$ for all $t\in \mathbf R$ and these are continuous functions of $t\in \mathbf R$.
\end{theorem}
\begin{proof}
Recall that $S = (S_1, \dots,S_{q_y},S_{q_y+1},\dots,S_{q})$ are independent, and such that $S_j \sim F_Y(t|z_{0})$ for all $j =1,\dots, q_y$ and $S_{q_y+j} \sim F_X(t|z_{0})$ for all $j=1,\dots, q_x$. Denote by $Q_Y(\cdot|z_{0})$ and $Q_X(\cdot|z_{0})$ the quantile functions of $F_Y(t|z_{0})$ and $ F_X(t|z_{0})$.
Consider the following argument:
{\small \begin{align*}
&E_P[ \phi (S)]  \\
&= P\left\lbrace \max_{k \in \{1,\dots,q_y\}} \left(\frac{1}{q_y}\sum_{j=1}^{q_y} I\{S_j \le S_k\}-\frac{1}{q_x}\sum_{j=q_y+1}^{q} I\{S_j \le  S_k\} \right) > c_{\alpha}(q_y,q) \right\rbrace  \\
&\overset{(1)}{=} P\left\lbrace \sup_{t\in\mathbf Y} \left(\frac{1}{q_y}\sum_{j=1}^{q_y} I\{Q_Y(U_j|z_{0}) \le t\}-\frac{1}{q_x}\sum_{j=q_y+1}^{q} I\{Q_X(U_j|z_{0}) \le t\} \right) > c_{\alpha}(q_y,q) \right\rbrace  \\
&\overset{(2)}{=} P\left\lbrace \sup_{t\in\mathbf Y} \left(\frac{1}{q_y}\sum_{j=1}^{q_y} I\{U_{j} \le F_Y(t|z_{0})\}-\frac{1}{q_x}\sum_{j=q_y+1}^{q} I\{U_{j} \le F_X(t|z_{0})\} \right) > c_{\alpha}(q_y,q) \right\rbrace \\
&\overset{(3)}{\leq}  P\left\lbrace \sup_{t\in\mathbf Y} \left(\frac{1}{q_y}\sum_{j=1}^{q_y} I\{U_{j} \le F_X(t|z_{0})\}-\frac{1}{q_x}\sum_{j=q_y+1}^{q} I\{U_{j} \le F_X(t|z_{0})\} \right) > c_{\alpha}(q_y,q) \right\rbrace  \\
&\overset{(4)}{=} P\left\lbrace \sup_{u\in \mathcal{U}} \left(\frac{1}{q_y}\sum_{j=1}^{q_y} I\{U_{j} \le u\}-\frac{1}{q_x}\sum_{j=q_y+1}^{q} I\{U_{j} \le u\} \right) > c_{\alpha}(q_y,q) \right\rbrace \\
&\overset{(5)}{\leq} P\left\lbrace \sup_{u\in\mathbf (0,1)} \left(\frac{1}{q_y}\sum_{j=1}^{q_y} I\{U_{j} \le u\}-\frac{1}{q_x}\sum_{j=q_y+1}^{q} I\{U_{j} \le u\} \right) > c_{\alpha}(q_y,q) \right\rbrace \\
&\overset{(6)}{=}\bar{\alpha}~,
\end{align*}}
where (1) follows from the quantile transformation and the fact that replacing $\max_{k \in \{1,\dots,q_y\}}$ with $\sup_{t\in \mathbf Y}$ does not affect the magnitude of the test statistic, (2) follows from  \citet[][Eq.\ 36]{pollard:02}, (3) from $P\in \mathbf P_0$, (4) from a simple change of variables and $\mathcal{U} := \cup_{t\in \mathbf Y}\{u =  F_X(t|z_{0}) \}$, (5) from $ \mathcal{U}\subseteq [0,1]$, and (6) from the definition of $c_{\alpha}(q_y,q)$ and $\bar{\alpha}$, and the fact that $\{U_{j}:j=1,\dots,q\}$ are i.i.d.\ distributed as $U( 0,1) $. 

To conclude the proof, it suffices to show that:  
(i) inequality (3) holds as an equality under the condition \( F_Y(t \mid z_{0}) = F_X(t \mid z_{0}) \) for all \( t \in \mathbf{R} \), and  
(ii) inequality (5) holds as an equality when these functions are continuous in \( t \).  
The first claim is immediate. For the second, it follows from the fact that the continuity of the CDFs implies \( \mathcal{U} = (0,1) \). By elementary properties of CDFs, $
\lim_{t \to -\infty} F_X(t \mid z_{0}) = 0$ and $\lim_{t \to \infty} F_X(t \mid z_{0}) = 1$. By the intermediate value theorem, for any \( u \in (0,1) \), there exists \( t \in \mathbf{R} \) such that \( u = F_X(t \mid z_{0}) \), implying \( u \in \mathcal{U} \), as desired.
\end{proof}

\subsection{Auxiliary lemmas}

\begin{lemma}
\label{lemma:equal-value-at-limit-S} 
Suppose that $S_n \sim P$ with $P$ satisfying Assumptions \ref{ass:Z}, \ref{ass:cond-continuous}, and \ref{ass:FiniteDisc} hold. Consider a sequence of random vectors $\{\tilde{S}_{n}:1\le n< \infty \}$ and a random vector $\tilde{S}$ defined on a common probability space $ (\Omega ,\mathcal{A},\tilde{P})$ such that
\begin{equation}
\Tilde{S}_{n}\overset{d}{=}S_{n},\quad \Tilde{S}\overset{d}{=}S,\quad \text{ and } \quad \Tilde{S}_{n}\overset{a.s.}{\to }\Tilde{S}~.
\label{eq:DCTimplication}
\end{equation}
Then, for any $j\in \{1,\dots ,q\}$,
\begin{equation*}
\tilde{P}\{\tilde{S}_{n,j}\neq \tilde{S}_{j},~\tilde{S}_{j}\in \mathcal{D} \}=o_{n}(1)~. 
\end{equation*}
\end{lemma}

\begin{proof}
We focus on an arbitrary \( j \in \{1, \dots, q_{y}\} \). The argument for \( j \in \{q_{y} + 1, \dots, q\} \) follows analogously by replacing \( Y \) with \( X \). By Lemma \ref{lemma:equal-value-at-limit}, it suffices to show that  
\begin{equation}\label{eq:equallimit_1}
    \sup_{t \in \mathbb{R}} \big| F_{\tilde{S}_{n,j}}(t) - F_{\tilde{S}_{j}}(t) \big| = o_{n}(1)~, 
\end{equation}  
where \( F_{\tilde{S}_{n,j}} \) and \( F_{\tilde{S}_{j}} \) denote the distribution functions of \( \tilde{S}_{n,j} \) and \( \tilde{S}_{j} \) in \( (\Omega, \mathcal{A}, \tilde{P}) \).

By the proof in Theorem \ref{thm:limit-Sn-general} (with $L=1$ and $\mathcal{Z}=\{z_{0}\}$), we have $ F_{S_{j}}(t)=F_{Y}(t|z_{0})$ and $ F_{S_{n,j}}(t)=E[F_{Y}(t|Z_{n_{y},(j)})]$ where $ Z_{n_{y},(j)}\overset{p}{\to }z_{0}$. By these and \eqref{eq:DCTimplication}, \eqref{eq:equallimit_1} follows from
\begin{equation}
\sup_{t\in \mathbf{R}}\vert E[F_{Y}(t|Z_{n_{y},(j)})]-F_{Y}(t|z_{0})\vert =o_{n}(1)~. \label{eq:equallimit_2}
\end{equation}

For any $\varepsilon >0$,
\begin{align*}
E[F_{Y}(t|Z_{n_{y},(j)})] &=\int_{\vert z-z_{0}\vert \leq \varepsilon }F_{Y}(t|z)dP_{Z_{n_{y},(j)}}( z) +\int_{\vert z-z_{0}\vert >\varepsilon }F_{Y}(t|z)dP_{Z_{n_{y},(j)}}( z) \\
&\leq \int_{\vert z-z_{0}\vert \leq \varepsilon }F_{Y}(t|z)dP_{Z_{n_{y},(j)}}( z) +P( \vert Z_{n_{y},(j)}-z_{0}\vert >\varepsilon ) \\
&\overset{(1)}{=} \int_{\vert z-z_{0}\vert \leq \varepsilon }F_{Y}(t|z)dP_{Z_{n_{y},(j)}}( z) +o_{n}(1)~,
\end{align*}
where (1) holds by $Z_{n_{y},(j)}\overset{p}{\to }z_{0}$. Then,
\begin{equation*}
\sup_{t\in \mathbf{R}}\vert E[F_{Y}(t|Z_{n_{y},(j)})]-F_{Y}(t|z_{0})\vert \leq \sup_{t\in \mathbf{R} }\sup_{\vert z-z_{0}\vert \leq \varepsilon }\vert F_{Y}(t|z)-F_{Y}(t|z_{0})\vert +o_{n}(1)~.
\end{equation*}
Fix $\delta >0$ arbitrarily. By Assumption \ref{ass:cond-continuous}, $\exists \varepsilon >0$ such that $ \sup_{t\in \mathbf{R}}\sup_{\vert z-z_{0}\vert \leq \varepsilon }\vert F_{Y}(t|z)-F_{Y}(t|z_{0})\vert <\delta /2$. For all large enough $n_y$, the right-hand side is bounded above by $\delta $. Since the choice of $\delta $ was arbitrary, \eqref{eq:equallimit_2} follows.
\end{proof}

\begin{lemma}\label{lem:aux_lemma4}
Let $V_{1}$ and $V_{2}$ be independent random variables that are discontinuous at a finite set of points $\mathcal{D}_{1}$ and $\mathcal{D}_{2}$, respectively. Then, for any $\delta >0$, $\exists \varepsilon >0$ small enough s.t.
\begin{align*}
P\{ \cup _{d\in \mathcal{D}_{1}}\{ |V_{1}-d|<\varepsilon \} \cap \{V_{1}\in \mathcal{D}_{1}^{c}\}\}  &<\delta ~, \\
P\{ \{ |V_{1}-V_{2}|<\varepsilon \} \cap \{V_{1}\in \mathcal{D}_{1}^{c}\}\cap \{ V_{2}\in \mathcal{D}_{2}^{c}\} \} &<\delta~ .
\end{align*}
\end{lemma}
\begin{proof}
Set $\bar{\varepsilon}:=\{ \min \vert d-\tilde{d}\vert :d,\tilde{d}\in \mathcal{D}_{1}\cap d\not=\tilde{d}\} >0$. For any $\varepsilon \in (0,\bar{\varepsilon}/2),$
\begin{align*}
&P\{ \cup _{d\in \mathcal{D}_{1}}\{ |V_{1}-d|<\varepsilon \}
\cap \{V_{1}\in \mathcal{D}_{1}^{c}\}\} \\
&\leq \sum_{d\in \mathcal{D}_{1}}P\{ \{|V_{1}-d|<\varepsilon \}\cap \{V_{1}\in \mathcal{D}_{1}^{c}\}\}  \\
&\leq _{(1)}\sum_{d\in \mathcal{D}_{1}}P\{ V_{1}\in (d-\varepsilon ,d)\cup (d,d+\varepsilon )\} \overset{(2)}{=} o_{\varepsilon }(1)~,
\end{align*}%
where (1) holds by $\varepsilon \in (0,\bar{\varepsilon}/2)$ and (2) by the fact that $(d-\varepsilon ,d)\cap (d,d+\varepsilon )$ has no discontinuities in the CDF. The right-hand side is less than $\delta $ by making $\varepsilon $
arbitrarily small, as desired. Also, for any $\varepsilon \in (0,\bar{\varepsilon}/2)$,
\begin{align*}
&P\{ \{ |V_{1}-V_{2}|<\varepsilon \} \cap \{V_{1}\in \mathcal{D}_{1}^{c}\}\cap \{ V_{2}\in \mathcal{D}_{2}^{c}\} \}  \\
&\overset{(1)}{=} \int_{\mathcal{D}_{2}^{c}}P\{ \{ |V_{1}-v_{2}|<\varepsilon \} \cap \{V_{1}\in \mathcal{D} _{1}^{c}\}\} dP_{V_{2}}( v_{2}) 
\overset{(2)}{=} o_{\varepsilon }(1)~,
\end{align*}
where (1) holds by $V_{1}\perp V_{2}$, and (2) holds by Lemma \ref{lem:aux_lemma3} and the dominated convergence theorem. The right-hand side is less than $\delta $ by making $\varepsilon $ arbitrarily small, as desired.
\end{proof}

\begin{lemma}\label{lem:aux_lemma3} 
Consider a random variable $V$ whose CDF is discontinuous at a finite set of points $\mathcal{D}$. Then, for any $v\in \mathbf{R}$,
\begin{equation*}
P\{ \{ |V-v|<\varepsilon \} \cap \{V\in \mathcal{D}^{c}\}\} =o_{\varepsilon}(1)~.
\end{equation*}
\end{lemma}
\begin{proof}
Set $\bar{\varepsilon}:=\{ \min \vert d-\tilde{d}\vert :d,\tilde{d}\in \mathcal{D}_{1}\cap d\not=\tilde{d}\} >0$. Fix $\varepsilon <\bar{\varepsilon}/2$. There are two possibilities: $v\in \mathcal{D}$ or $v\not\in \mathcal{D}$. First, consider $v\in \mathcal{D}$. In this case,
\begin{equation*}
P\{ |V-v|<\varepsilon \} \cap \{V\in \mathcal{D}^{c}\} \leq P\{ V\in ( v-\varepsilon ,v) \cap ( v,v+\varepsilon ) \} \overset{(1)}{=}  o_{\varepsilon}(1)~,
\end{equation*}
where (1) holds because $(d-\varepsilon ,d)\cap (d,d+\varepsilon )$ has no mass points. Second, consider $ v\not\in \mathcal{D}$.  Then,
\begin{equation*}
P\{ \{ |V-v|<\varepsilon \} \cap \{V\in \mathcal{D} ^{c}\}\} \leq F_{V}(v+\varepsilon )-F_{V}(v-\varepsilon )\overset{(2)}{=} o_{\varepsilon }(1)~ ,
\end{equation*}
where (1) by the fact that $( v-\varepsilon ,v+\varepsilon ] $ has no mass points, so $F_{V}$ is continuous on that interval.
\end{proof}

\begin{lemma}\label{lemma:equal-value-at-limit} 
Let $\{V_{n}:n\ge 1\}$ be a sequence of random variables that satisfy $V_{n}\overset{p}{\to }V$, where $V$ is a random variable whose CDF is discontinuous at a finite set of points $\mathcal{D}$. Furthermore, assume $\sup_{t\in \mathbf{R}}|F_{V_{n}}(t)-F_{V}(t)|\to 0$. Then, 
\begin{equation*}
P\{\{V\in \mathcal{D}\} \cap \{V_{n}\neq V\}\}\to 0~.
\end{equation*}
\end{lemma}
\begin{proof}
Fix $\delta >0$ arbitrarily. It suffices to find $N=N(\delta) $
such that $P\{V\in \mathcal{D},V_{n}\neq V\}<\delta $ for all $n>N$.

Set $\bar{\varepsilon}:=\{ \min \vert d-\tilde{d}\vert :d, \tilde{d}\in \mathcal{D}\cap d\not=\tilde{d}\} >0$. For any $ \varepsilon \in (0,\bar{\varepsilon}/2),$ consider the following argument. 
\begin{align*}
&P\{\{V\in \mathcal{D}\} \cap \{V_{n}\neq V\}\} \\
&=\sum_{d\in \mathcal{D}}P\{\{V=d\} \cap \{V_{n}\neq d\} \} \\
& \overset{(1)}{=} \sum_{d\in \mathcal{D}}P\{\{V=d\} \cap \{V_{n}\neq d\} \cap \{|V_{n}-d|<\varepsilon\}\}+o_{n}(1) \\
& \leq \sum_{d\in \mathcal{D}}P\{V_{n}\in (d-\varepsilon ,d)\cap (d,d+\varepsilon )\}+o_{n}(1) \\
& \overset{(2)}{=} 
\sum_{d\in \mathcal{D}}P\{V \in (d-\varepsilon ,d)\cap (d,d+\varepsilon )\}+o_{n}(1) \\
& \overset{(3)}{=} o_{\varepsilon }(1)+o_{n}(1)~,
\end{align*}
where (1) holds by $V_{n}\overset{p}{\to }V$, (2) by $\sup_{t\in \mathbf{R}}|F_{V_{n}}(t)-F_{V}(t)|=o_{n}(1) $,
and (3) by the fact that $(d-\varepsilon ,d)\cap (d,d+\varepsilon )$ has no discontinuities in the CDF. 
For all large enough $n$ and small enough $\varepsilon $, the right-hand side is bounded by $\delta $, as desired.
\end{proof}

\begin{lemma}\label{lem:ideal-power}
  Let $\alpha\in(0,1)$ and $P'\in\mathbf P_1:=\mathbf{P}\backslash\mathbf{P}_0$ be given, and let $t^*\in\mathbf R$ and
  $\varepsilon:=F_Y(t^*|z_0)-F_X(t^*|z_0)>0$, which exist by the definition of $\mathbf P_1$. Let
  $\kappa_\alpha:=\sqrt{\log(2/\alpha)/2}$. Then:
  \begin{enumerate}
  \item[(i)] $c_\alpha(a,a+b)\le \kappa_\alpha\big(a^{-1/2}+b^{-1/2}\big)$ for all
  $(a,b)\in\mathbf N^2$;
  \item[(ii)] for every $(a,b)\in\mathbf N^2$ such that
  $\kappa_\alpha(a^{-1/2}+b^{-1/2})<\varepsilon/3$,
  \begin{equation*}
   E_{P'}\big[\phi(S^{\circ}(a,b))\big]\;\ge\;1-e^{-2a\varepsilon^2/9}-e^{-2b\varepsilon^2/9}~.
  \end{equation*}
  \end{enumerate}
  In particular, $E_{P'}[\phi(S^{\circ}(a_n,b_n))]\to1$ along any sequences $(a_n,b_n)$ with
  $\min\{a_n,b_n\}\to\infty$.
  \end{lemma}
\begin{proof}
  We first prove (i). By \eqref{eq:our-cv}, $c_\alpha(a,a+b)$ is the $1-\alpha$ quantile of $\sup_{u\in(0,1)}\Delta_{a,a+b}(u)$, where $\Delta_{a,a+b}(u)=\hat G_a(u)-\hat H_b(u)$ and $\hat G_a$ and $\hat H_b$ denote the empirical CDFs of two independent i.i.d.\ samples from the uniform distribution on $(0,1)$, of sizes $a$ and $b$; see \eqref{eq:Delta-u}. For every $u\in(0,1)$,
  \begin{equation*}
   \Delta_{a,a+b}(u)=\big(\hat G_a(u)-u\big)+\big(u-\hat H_b(u)\big)
   \le \sup_{u\in(0,1)}\big(\hat G_a(u)-u\big)+\sup_{u\in(0,1)}\big(u-\hat H_b(u)\big)~.
  \end{equation*}
  By the one-sided Dvoretzky--Kiefer--Wolfowitz inequality with the tight constant \citep[Theorem 1]{massart:90}, $P\{\sup_u(\hat G_a(u)-u)>\xi\}\le e^{-2a\xi^2}$ for every $\xi>0$, and the same bound holds for $\sup_u(u-\hat H_b(u))$ by symmetry. Setting $\xi_a:=\kappa_\alpha a^{-1/2}$ and $\xi_b:=\kappa_\alpha b^{-1/2}$, so that $e^{-2a\xi_a^2}=e^{-2b\xi_b^2}=\alpha/2$, the union bound gives
  \begin{equation*}
   P\Big\{\sup_{u\in(0,1)}\Delta_{a,a+b}(u)>\xi_a+\xi_b\Big\}\le\alpha~,
  \end{equation*}
  and therefore $c_\alpha(a,a+b)\le \xi_a+\xi_b=\kappa_\alpha(a^{-1/2}+b^{-1/2})$, as desired.

  We now prove (ii). Let $\hat F^{\circ}_Y(t):=a^{-1}\sum_{i\le a}I\{Y^{\circ}_i\le t\}$ and $\hat F^{\circ}_X(t):=b^{-1}\sum_{j\le b}I\{X^{\circ}_j\le t\}$ denote the empirical CDFs of the two blocks of $S^{\circ}(a,b)$, and recall from \eqref{eq:ks-stat} that the test statistic satisfies $T(S^{\circ}(a,b))\ge \hat F^{\circ}_Y(t^*)-\hat F^{\circ}_X(t^*)$. Since $a\hat F^{\circ}_Y(t^*)$ is a Binomial random variable with success probability $F_Y(t^*|z_0)$, Hoeffding's inequality yields
  \begin{equation*}
   P'\Big\{\hat F^{\circ}_Y(t^*)\le F_Y(t^*|z_0)-\varepsilon/3\Big\}\le e^{-2a\varepsilon^2/9}~,
   \qquad
   P'\Big\{\hat F^{\circ}_X(t^*)\ge F_X(t^*|z_0)+\varepsilon/3\Big\}\le e^{-2b\varepsilon^2/9}~.
  \end{equation*}
  On the intersection of the complements of these two events,
  \begin{equation*}
   T(S^{\circ}(a,b))\;\ge\;\hat F^{\circ}_Y(t^*)-\hat F^{\circ}_X(t^*)
   \;\ge\;\big(F_Y(t^*|z_0)-\varepsilon/3\big)-\big(F_X(t^*|z_0)+\varepsilon/3\big)
   \;=\;\varepsilon/3
   \;>\;c_\alpha(a,a+b)~,
  \end{equation*}
  where the last inequality holds by part (i) and the hypothesis $\kappa_\alpha(a^{-1/2}+b^{-1/2})<\varepsilon/3$; hence the test rejects. Part (ii) then follows from the union bound. The final claim follows from (i) and (ii), since $\kappa_\alpha(a_n^{-1/2}+b_n^{-1/2})\to0$ and both exponential terms vanish as $\min\{a_n,b_n\}\to\infty$.
  \end{proof}

\section{Extensions and Supplementary Analyses}\label{app:supp-material}
\subsection{Multiple target points}\label{app:multiple-target}

Throughout the paper, we have focused on testing the null hypothesis \eqref{eq:H0} with $L = 1$ to emphasize the main ideas and maintain clarity. This section describes how our method extends when $L > 1$. In this setting, we address the intersection nature of the null hypothesis by using a maximum test. 

To formally define the test, we need to update our notation to explicitly reflect dependence on both $\ell$ and $\alpha$. Consider the point $z_{\ell}\in \mathcal Z$ for some $\ell\in \{1,\dots,L\}$. Let $g_{\ell}(Z) := |Z - z_{\ell}|$ and, for any two values \( z, z' \in \mathcal{Z} \), define the ordering \(\le_{\ell}\) as follows:
\[
z \le_{\ell} z' \quad \text{if and only if} \quad g_{\ell}(z) \le g_{\ell}(z')~.
\]
For each $\ell$, this leads to $g$-order statistics \( Z_{\ell,(i)} \) in each sample, as well as induced order statistics for $Y$ and $X$, that we denote by 
\begin{equation*}
  Y^{\ell}_{n_y,[1]},Y^{\ell}_{n_y,[2]}, \dots,Y^{\ell}_{n_y,[q_y]}~\text{ and } X^{\ell}_{n_x,[1]},X^{\ell}_{n_x,[2]}, \dots,X^{\ell}_{n_x,[q_x]}~.
\end{equation*}
As with $L=1$, any ties in the $g$-ordering can be resolved arbitrarily. Finally, if we let \(n := n_y + n_x\) and \(q := q_y + q_x\), the effective (pooled) sample is given by
\begin{equation}\label{eq:Sn_ell}
  S^{\ell}_n = (S^{\ell}_{n,1}, \dots, S^{\ell}_{n,q}) := (Y^{\ell}_{n_y,[1]}, \dots, Y^{\ell}_{n_y,[q_y]}, X^{\ell}_{n_x,[1]}, \dots, X^{\ell}_{n_x,[q_x]})~.
\end{equation}
Our test is entirely based on \(S^{\ell}_n\), with the default test statistic being 
\begin{equation*}\label{eq:ks-stat-ell}
  T(S^{\ell}_n) = \sup_{t\in \mathbf R} \left( \hat F^{\ell}_{n,Y}(t) - \hat F^{\ell}_{n,X}(t) \right) = \max_{k \in \{1,\dots,q_y\}} \left( \hat F^{\ell}_{n,Y}(S^{\ell}_{n,k}) - \hat F^{\ell}_{n,X}(S^{\ell}_{n,k}) \right) ~,
\end{equation*}
where the empirical CDFs are, 
\begin{equation*}
  \hat F^{\ell}_{n,Y}(t) := \frac{1}{q_y} \sum_{j=1}^{q_y} I\{S^{\ell}_{n,j} \le t\} \quad \text{ and }\quad \hat F^{\ell}_{n,X}(t) := \frac{1}{q_x} \sum_{j=1}^{q_x} I\{S^{\ell}_{n,q_y+j} \le t\}~.
\end{equation*}
In order to define our test below in \eqref{eq:phi-multiple}, we first introduce $\phi_{\ell}(\alpha)$, where 
\begin{equation*}\label{eq:our-test-ell}
  \phi_{\ell}(S^{\ell}_n,\alpha) := I\{T(S^{\ell}_n)>c_{\alpha}(q_y,q)\}~.
\end{equation*}
The test \( \phi_{\ell}(S^{\ell}_n,\alpha) \) depends on the point \( z_{\ell} \), or equivalently the index \( \ell \), in two ways. First, the random variable \( S^{\ell}_n \) depends on \( z_{\ell} \) through the induced order statistics, meaning the test statistic \( T(S^{\ell}_n) \) is influenced by the choice of the target point \( z_{\ell} \). Second, the critical value \( c_{\alpha}(q_y, q) \) is a function of \( (q_y, q) \), which, through the data-dependent rules introduced in Section \ref{sec:tuning-parameters}, implicitly depends on \( z_{\ell} \). The dependence on \( \alpha \) is directly evident from the definition of the critical value \( c_{\alpha}(q_y, q) \).

To test the null hypothesis in \eqref{eq:H0} when $L>1$, we propose the following test:
\begin{equation}\label{eq:phi-multiple}
    \phi(S^{1}_n,S^{2}_n,\dots,S^{L}_n) := \max_{\ell \in \{1,\dots,L\}} \phi_{\ell}\left(S^{\ell}_n,1-(1-\alpha)^{1/L}\right)~.
\end{equation}
In other words, the test $\phi(S^{1}_n,S^{2}_n,\dots,S^{L}_n)$ is the maximum of the $L$ individual tests, each computed at a target point $z_{\ell}$. However, each of these individual tests is performed at a level of $1 - (1 - \alpha)^{1/L}$, rather than at the nominal level $\alpha$. Since $\mathcal{Z}$ consists of a finite number of points, the asymptotic validity of the test $\phi(\alpha)$ follows from the validity of each individual test $\phi_{\ell}$ and the fact that they are asymptotically independent. We formalize this result in Theorem \ref{thm:limit-Sn-general}.

\subsection{Results for other test statistics}\label{app:other-stats}

Our paper proposes the test in \eqref{eq:our-test} based on the one-sided KS test statistic in \eqref{eq:ks-stat} and a critical value in \eqref{eq:our-cv} constructed by simulating i.i.d.\ uniform random variables applied to this statistic. In this section, we investigate the properties of analogous tests that replace the KS statistic with other commonly used statistics in the stochastic dominance literature, such as the Cramér--von Mises (CvM) and Anderson--Darling (AD) statistics.

The main takeaway of this section is that the validity of our approach does not extend to tests that replace the KS statistic with the CvM or AD statistics. Specifically, we provide examples of data-generating processes that satisfy the null hypothesis in \eqref{eq:H0} and the test overrejects. Our examples involve discretely distributed data, which is allowed by our assumptions. That said, the CvM-based version of our test remains valid under the assumption of continuously distributed data, as we show later.

For concreteness, we now define the CvM and AD analogs of our test for the null hypothesis in \eqref{eq:H0}. Given the pooled sample $S_n$ in \eqref{eq:Sn}, the one-sided CvM test statistic is given by 
\begin{equation}\label{eq:cvm-stat}
  T_{\rm CvM}(S_n) =\int \left( \hat F_{n,Y}(s) - \hat F_{n,X}(s) \right)^{+}d\hat F_{n,S}(s) = \frac{1}{q}\sum_{j=1}^{q}\left( \hat F_{n,Y}(S_{n,j}) - \hat F_{n,X}(S_{n,j})\right)^{+}~,
\end{equation}
where $\hat{F}_{n,S}(s)$ denotes the empirical CDF of $S_n$ and $x^{+} = \max\{x, 0\}^2$. In turn, the one-sided AD test statistic is given by
\begin{equation}\label{eq:ad-stat}
    T_{\rm AD}(S_n) = \int \frac{\left( \hat{F}_{n,Y}(s) -\hat{F}_{n,X}(s) \right)^{+} }{\hat{F}_{n,S}(s)( 1-\hat{F}_{n,S}(s))} d\hat{F}_{n,S}(s) = \frac{1}{q}\sum_{j=1}^q \frac{\left( \hat{F}_{n,Y}(S_{n,j}) -\hat{F}_{n,X}(S_{n,j}) \right)^{+} }{\hat{F}_{n,S}(S_{n,j}) \left( 1-\hat{F}_{n,S}(S_{n,j}) \right) }~.
\end{equation}
These statistics align with their standard textbook definitions, as discussed in \citet[][p. 101]{hajek/sidak/sen:99}. For both of these statistics, we can define the analog test as in our paper. For the CvM statistic, the corresponding test is 
\begin{equation}\label{eq:cvm-test}
    \phi_{\rm CvM}(S)=I\{T_{\rm CvM}(S)>c_{\alpha,{\rm CvM}}(q_y, q)\} .
\end{equation}
where $c_{\alpha,{\rm CvM}}(q_y,q) = \inf_{x\in\mathbf R}\left\lbrace P\left\lbrace T_{\rm CvM}(U) \le x \right\rbrace\ge 1-\alpha \right\rbrace $ and $U=(U_1, \cdots, U_q)$ is a vector of i.i.d.\ $U(0,1)$. The testing procedure based on the AD statistic is defined analogously.

We are now ready to argue that the CvM and AD analogs of our test are invalid under our assumptions. Since the CvM and AD statistics are both rank statistics, we can repeat arguments in Theorem \ref{thm:SD-AsySz} to show that, for any $P \in \mathbf P_{0}$,
\begin{align*}
\lim_{n \to \infty} E_{P}[\phi_{\rm CvM} (S_{n})]  = E_{\tilde{P}}[\phi_{\rm CvM} (\tilde{S})]~~~\text{ and }~~~\lim_{n \to \infty} E_{P}[\phi_{\rm AD} (S_{n})]  = E_{\tilde{P}}[\phi_{\rm AD} (\tilde{S})].
\end{align*}
To prove that these tests are invalid, it then suffices to find any $P \in \mathbf P_{0}$ such that $E_{\tilde{P}}[\phi_{\rm CvM} (\tilde{S})]>\alpha$ and $E_{\tilde{P}}[\phi_{\rm AD} (\tilde{S})]>\alpha$. To this end, we consider $q_y=2$, $q_x=1$, $\{Y|Z=z_0\}$ and $\{X|Z=z_0\}$ distributed according to $Bernoulli(0.5)$. For $\alpha = 5\%$, we get $E_{\tilde{P}}[\phi_{\rm CvM} (\tilde{S})] \approx 12.5\%$ and $E_{\tilde{P}}[\phi_{\rm AD} (\tilde{S})] \approx 12.5\%$, as desired.

A notable feature of the examples above is that the data are discrete. This naturally raises the question of whether the validity of these tests can be recovered in settings with continuously distributed data. While we were unable to establish this result for the AD test, we were able to do so for the CvM test.

\begin{theorem}\label{thm:SD-AsySz_CvM}
Let $\mathbf P$ the space of distributions that satisfy Assumptions \ref{ass:Z}, \ref{ass:cond-continuous}, and $Y|Z=z_0$ and $X|Z=z_0$ are continuously distributed, and let $\mathbf P_{0}$ be the subset of $\mathbf P$ such that \eqref{eq:H0} holds. Let $\alpha\in(0,1)$ be given, $T_{\rm CvM}$ be the CvM test statistic in \eqref{eq:cvm-stat}, and $\phi_{\rm CvM}(\cdot)$ be the test in \eqref{eq:cvm-test}. Then, 
  \begin{equation} \label{eq:Asysize_cvm}
    \limsup_{n \to \infty} E_{P}[\phi_{\rm CvM}(S_n)]\le \alpha
  \end{equation}
  whenever $P \in \mathbf P_{0}$.
\end{theorem}
\begin{proof}
    This result follows from repeating arguments that prove Theorem \ref{thm:SD-AsySz}. The only difference in the proof is that Theorem \ref{thm:limit-cont} is replaced by Theorem \ref{thm:CvM-validity}. 
\end{proof}

The previous theorem relies on the following auxiliary result. 

\begin{theorem}\label{thm:CvM-validity}
Let $S$ be the random variable in Theorem \ref{thm:limit-Sn} and let $\mathbf P_0$ be as in Theorem \ref{thm:SD-AsySz_CvM}. Then, for any $P\in \mathbf P_0$ and $\alpha\in(0,1)$, we obtain 
\[
E_P[ \phi_{\rm CvM} (S)] \leq \bar{\alpha}_{\rm CvM}\le \alpha ~,
\]
where 
$$\Bar{\alpha}_{\rm CvM} = P\left\{T_{\rm CvM}(U) > c_{\alpha,{\rm CvM}}(q_y, q)\right\}.$$
Moreover, the first inequality becomes an equality under $P$ such that  $F_Y(t|z_{0})=F_X(t|z_{0})$ for all $t\in \mathbf R$.
\end{theorem}
\begin{proof}
For any $t \in \mathbf R$, consider the CDF
\[
F_{M}(t) =\frac{q_{Y}}{q}F_{Y| Z= z_0}(t) +\frac{q_{X}}{q}F_{X| Z = z_0}(t)~,
\]
and let $Q_M$ denote its corresponding quantile function. Let $Q_Y$ and $Q_X$ are the quantiles functions associated with $F_Y(\cdot| z_0)$ and $F_X(\cdot| z_0)$, respectively.

Since $Y|Z=z_0$ and $X|Z=z_0$ are continuously distributed, $F_M$ is continuous. In turn, this implies that $Q_M$ is strictly increasing. By $F_{Y| Z= z_0}(t)\leq F_{X| Z= z_0}(t)$ for all $t \in \mathbf R$, we also have
\begin{align}\label{eq:cvm-monotonicity}
 Q_Y(u) \geq Q_M(u) \geq Q_X(u) ~~\forall u\in(0,1)~.
\end{align}

For brevity notation, we denote the critical value as $ c_{\alpha}\equiv c_{\alpha,{\rm CvM}}(q_y, q)$. Then,
\begin{align}
&E_{P}[\phi_{\rm CvM}(S)]\notag\\
&=P\left\{ \frac{1}{q}\sum_{u=1}^{q}\left( \frac{1}{q_{y}}\sum_{i=1}^{q_{y}}I\left\{
S_{i}\leq S_{u}\right\} -\frac{1}{q_{x}}\sum_{j=q_y+1}^{q}I\left\{ S_{j}\leq
S_{u}\right\} \right)^{+}> c_{\alpha}\right\} \notag\\
&\overset{(1)}{=}P\left\{ 
\left[\begin{array}{c}
\frac{1}{q}\sum\limits_{u=1}^{q_{y}} \bigg( \frac{1}{q_{y}}\sum\limits_{i=1}^{q_{y}} I\left\{ Q_{Y}\left( U_{i}\right) \leq Q_{Y}\left( U_{u}\right) \right\} -\frac{1}{q_{x}}\sum\limits_{j=q_y+1}^{q} I\left\{ Q_{X}\left( U_{j}\right) \leq Q_{Y}\left(U_{u}\right) \right\} \bigg)^{+}+\\ 
\frac{1}{q}\sum\limits_{u=q_y+1}^{q} \bigg( \frac{1}{q_{y}}\sum\limits_{i=1}^{q_{y}} I\left\{ Q_{Y}\left( U_{i}\right) \leq Q_{X}\left( U_{u}\right) \right\}-\frac{1}{q_{x}}\sum\limits_{j=q_y+1}^{q} I\left\{ Q_{X}\left( U_{j}\right) \leq Q_{X}\left(U_{u}\right) \right\}\bigg)^{+}
\end{array}
\right]
> c_{\alpha}\right\}  \notag\\
&\overset{(2)}{\leq} P\left\{ 
\left[
\begin{array}{c}
\frac{1}{q}\sum\limits_{u=1}^{q_{y}} \bigg( \frac{1}{q_{y}}\sum\limits_{i=1}^{q_{y}} I\left%
\{ Q_{M}\left( U_{i}\right) \leq Q_{M}\left( U_{u}\right) \right\} -\frac{1}{q_{x}}\sum\limits_{j=q_y+1}^{q} I\left\{ Q_{M}\left( U_{j}\right) \leq Q_{M}\left( U_{u}\right) \right\} \bigg)^{+}+ \\ 
\frac{1}{q}\sum\limits_{u=q_y+1}^{q} \bigg( \frac{1}{q_{y}}\sum\limits_{i=1}^{q_{y}} I\left\{ Q_{M}\left( U_{i}\right) \leq Q_{M}\left( U_{u}\right) \right\}-\frac{1}{q_{x}}\sum\limits_{j=q_y+1}^{q} I\left\{ Q_{M}\left( U_{j}\right) \leq Q_{M}\left(U_{u}\right) \right\} \bigg)^{+}%
\end{array}%
\right]
> c_{\alpha}\right\} \notag\\
&\overset{(3)}{=} P\left\{ 
\left[
\begin{array}{c}
\frac{1}{q}\sum\limits_{u=1}^{q_{y}} \left( \frac{1}{q_{y}}\sum\limits_{i=1}^{q_{y}} I\left%
\{  U_{i} \leq U_{u} \right\} -\frac{1}{%
q_{x}}\sum\limits_{j=q_y+1}^{q_{x}} I\left\{  U_{j} \leq 
U_{u}\right\} \right)^{+}+ \\ 
\frac{1}{q}\sum\limits_{u=q_y+1}^{q} \left( \frac{1}{q_{y}}\sum\limits_{i=1}^{q_{y}} I\left%
\{  U_{i} \leq  U_{u} \right\} -\frac{1}{%
q_{x}}\sum\limits_{j=q_y+1}^{q} I\left\{ U_{j} \leq 
U_{u} \right\} \right)^{+}%
\end{array}%
\right]
> c_{\alpha}\right\} \notag\\
&\overset{(4)}{=} P\left\{ 
T_{\rm CvM}(U)
> c_{\alpha}\right\} \overset{(5)}{=}  \Bar{\alpha}_{\rm CvM} \overset{(6)}{\leq} \alpha~,\label{eq:cvm-monotonicity3}
\end{align}
where (1) holds by the quantile transformation with $(U_{1},\cdots ,U_{q})$ i.i.d.\ $U(0,1)$, (2) by
\begin{equation}\label{eq:cvm-monotonicity2}
\{ Q_{Y}(U_{a}) \leq Q_{X}(U_{b}) \}
\subseteq \{ Q_{M}(U_{a}) \leq Q_{M}(U_{b})
\} \subseteq \{ Q_{X}(U_{a}) \leq Q_{Y}(
U_{b}) \}~
\end{equation}
for all $a,b=1,\dots,q$, (3) by the fact that $Q_{M}$ is strictly increasing, and (4), (5), and (6) by the definitions of $T_{\mathrm{CvM}}$ and $\Bar{\alpha}_{\mathrm{CvM}}$. We now show that \eqref{eq:cvm-monotonicity2} holds. For any $a,b=1,\dots,q$, suppose that $Q_{Y}(U_{a}) \leq Q_{X}(U_{b}) $. By this and \eqref{eq:cvm-monotonicity}, we have that $ Q_{M}(U_{a}) \leq Q_{Y}(U_{a}) \leq Q_{X}(U_{b}) \leq Q_{M}(U_{b}) $, which implies that $Q_{M}(U_{a}) \leq Q_{M}(U_{b})$, as desired by the first inclusion. In turn, by this and \eqref{eq:cvm-monotonicity}, we have that $Q_{X}(U_{a}) \leq Q_{M}(U_{a}) \leq Q_{M}(U_{b}) \leq Q_{Y}(U_{b}) $, as desired by the second inclusion.

Since $P\in \mathbf{P}_0$ was arbitrary, \eqref{eq:cvm-monotonicity3} implies that $\sup_{P \in \mathbf{P}_0} E_{P}[\phi_{\rm CvM}(S)] \leq \Bar{\alpha}$. Furthermore, under $F_Y(\cdot|z_0) = F_X(\cdot|z_0)$, we have that $Q_Y=Q_X=Q_M$, and so the inequality (2) in \eqref{eq:cvm-monotonicity3} holds with equality, as desired.
\end{proof}

\subsection{Empirical application}\label{app:empirical-application}

In this section, we consider two empirical applications. The first is an RDD in which the target point is known. The second is a two-independent-sample setting in which the target point is the median of the conditioning variable and must be estimated from the data. In both applications, the outcome variable contains ties. These applications therefore illustrate the use of our procedure in settings where the outcome distributions need not be atomless and, in the second application, where the target conditioning point is estimated from data.

\subsubsection{The effect of retirement on household consumption}
We apply our methodology to assess the distributional effect of retirement on household expenditures. Our empirical analysis follows \citet{battistin/Brugianivi/Rettore/Weber:2009} and \cite{shen/zhang:16} (SZ, henceforth). As in these studies, we use data from the 1993--2004 Italian Survey on Household Income and Wealth (SHIW), a repeated cross-section survey that covers roughly 8,000 households across more than 3,000 municipalities; see \cite{shen/zhang:2016data}. Because the outcomes are survey-reported, they exhibit ties, although the largest empirical point mass does not exceed 1.3\%

We use the SHIW data to evaluate whether retirement shifts the distribution of household expenditures downward. As SZ explains, this hypothesis can be examined by framing the problem as an RDD, comparing the distributions of household expenditures at ages just before and just after retirement. To formulate our problem, let $Z$ denote the age of the household head, $z_0$ the retirement age, $Y$ household expenditure before retirement, and $X$ household expenditure after retirement. Under continuity of the conditional distributions in the retirement age of the household head, the test of a negative distributional effect of retirement can be formulated as the following CSD testing problem:
\begin{align}
H_0 : F_{Y}(t|z_0)\le F_{X}(t|z_0)~\text{for all}~ t \in \mathbf{R},\qquad
H_1:F_{Y}(t|z_0)> F_{X}(t|z_0)~\text{for some}~ t \in \mathbf{R}~.
\label{eq:empApplic_H0}
\end{align}
We note that $H_0$ in \eqref{eq:empApplic_H0} is a special case of \eqref{eq:H0} with $\mathcal{Z} = \{z_0\}$. Following the logic of RDDs, conditioning on $Z = z_0$ is key to controlling for confounding factors.

\begin{table}[h]\small 
    \centering
    \begin{tabular}{llcccccc}
      \hline
      \multicolumn{1}{l}{{sample}} &
      \multicolumn{1}{l}{{expenditure}} &
      \multicolumn{1}{c}{{p-value}} &
      \multicolumn{1}{c}{{$q$}} &
      \multicolumn{1}{c}{{p-value}} &
      \multicolumn{1}{c}{{$q$}} &
      \multicolumn{1}{c}{{$c_L$}} &
      \multicolumn{1}{c}{{$c_U$}} \\
      \multicolumn{1}{l}{{($n$)}} &
      \multicolumn{1}{l}{{type}} &
      \multicolumn{1}{c}{{(KS)}} &
      \multicolumn{1}{c}{{(KS)}} &
      \multicolumn{1}{c}{{(SZ)}} &
      \multicolumn{1}{c}{{(SZ)}} &
      \multicolumn{1}{c}{{(KS)}} &
      \multicolumn{1}{c}{{(KS)}}\\
      \hline
       all households & nondurable & 22.64 & 479 & 13.69 & 8,739 & 0 & 52.30 \\
       (26,914) & food & 86.05 & 466 & 71.74 & 8,739 & 0 & 42.10 \\
       \hline
       wife already eligible & nondurable & 99.55 & 185 & 66.72 & 1,822 & 0 & $\infty$ \\
       (4,798) & food & 85.96 & 188 & 99.69 & 1,414 & 0 & 25.70 \\
       \hline
       wife ineligible & nondurable & 16.14 & 232 & 0.15$^{*}$ & 1,460 & 0 & $\infty$ \\
       (7,305) & food & 13.95 & 228 & 3.94$^{*}$ & 1,460 & 0.75 & 3.60 \\
       \hline
       single men & nondurable & 43.23 & 296 & 65.82 & 2,884 & 0 & 16.30 \\
       (9,001) & food & 41.33 & 292 & 31.19 & 2,537 & 0 & 18.60 \\
       \hline
    \end{tabular}
    \caption{\small Empirical results for the SHIW data. The table reports p-values and effective sample sizes for our KS test and the SZ test. All p-values are reported in percent. The last two columns report the sensitivity analysis for the KS test, where the baseline effective sample sizes $(q_y,q_x)$ are rescaled as $(cq_y,cq_x)$. The values $c_L$ and $c_U$ are the lower and upper breakpoints at which the 10\% test decision changes relative to the baseline choice $c=1$. Values of $0$ and $\infty$ indicate no change below or above $c=1$, respectively, over the range considered. The symbol $^{*}$ denotes rejection at the 10\% level.}
     \label{tab:empApplication}
\end{table}

Table \ref{tab:empApplication} reports the results of implementing the test in \eqref{eq:empApplic_H0} using both our methodology (KS) and the procedure proposed by SZ. Following SZ, we consider expenditures on nondurables and food, and we report results for four samples: the full sample, the subsample of households with wives eligible for retirement, the subsample of households with wives not yet eligible for retirement, and the subsample of households consisting of single men. The table presents the p-values for both tests and their corresponding effective sample sizes. While our test (KS) does not reject $H_0$ across all subsamples and expenditures, SZ rejects $H_0$ for both nondurable and food expenditures (at the 1\% and 5\% levels, respectively) for the subsample of households in which the wife is ineligible for retirement. Consistent with our simulation evidence, the SZ test uses substantially larger effective sample sizes than our KS test. 

Furthermore, the last two columns of Table \ref{tab:empApplication} report a sensitivity analysis with respect to the choice of effective sample sizes. Given the baseline choice $(q_y,q_x)$ computed by our rule of thumb, we rerun the test using $(cq_y,cq_x)$, where $c>0$ is a common scaling factor that varies over an evenly spaced grid with step size $0.05$. The lower breakpoint $c_L$ is the first grid point below $1$ at which the $10\%$ test decision changes, while the upper breakpoint $c_U$ is defined analogously above $1$. Except for food expenditure in the wife ineligible subsample, the interval $[c_L,c_U]$ is wide, with $c_L=0$ and $c_U$ ranging from $16.30$ to $\infty$, indicating that the empirical conclusions are robust to the choice of effective sample sizes. For food expenditure in the wife ineligible subsample, the interval is narrower, $[0.75,3.60]$, implying that the rejection at the $10\%$ level is somewhat more sensitive to the choice of effective sample sizes, although the conclusion remains unchanged over a substantial range around the baseline choice.

\subsubsection{The effect of class size on student achievement}
We next illustrate our methodology using data from the Tennessee Student/Teacher Achievement Ratio (STAR) experiment; see \cite{achilles/etal:2008data}. Following \citet{qu/yoon:2015}, we study the distributional effect of class size on student achievement conditional on teacher experience. The target conditioning point is chosen as the median teacher experience and estimated by its sample analogue. Our analysis focuses on conditional first-order stochastic dominance of the outcome distributions at this target value, rather than the conditional quantile treatment effects considered by \citet{qu/yoon:2015}.

Let $Y$ and $X$ denote the student achievement scores of students assigned to small and regular classes, respectively, and let $Z$ denote teacher experience, measured in years. We consider three achievement outcomes: reading scores, mathematics scores, and their sum. All three outcomes are integer-valued, so ties arise naturally in the recorded outcomes. Table~\ref{tab:star-score-point-masses} reports the largest empirical point masses in both the full samples and the effective samples used by our test. Following \citet{qu/yoon:2015}, we estimate the median teacher experience by its sample analogue. We consider both dominance directions:
\begin{align}\label{eq:empApplication2}
\begin{aligned}
H_0^{+}:&\ F_Y(t\mid z_0)\le F_X(t\mid z_0)\ \text{for all }t\in\mathbf R,
&
H_1^{+}:&\ F_Y(t\mid z_0)>F_X(t\mid z_0)\ \text{for some }t\in\mathbf R,\\
H_0^{-}:&\ F_X(t\mid z_0)\le F_Y(t\mid z_0)\ \text{for all }t\in\mathbf R,
&
H_1^{-}:&\ F_X(t\mid z_0)>F_Y(t\mid z_0)\ \text{for some }t\in\mathbf R.
\end{aligned}
\end{align}
Here, $H_0^{+}$ corresponds to the hypothesis that small classes first-order stochastically dominate regular classes, while $H_0^{-}$ corresponds to the reverse dominance hypothesis.
\begin{table}[h]\small 
    \centering
    \begin{tabular}{lcccccccc}
      \hline
      null 
      & outcome
      & $q_{\mathrm{small}}$ 
      & $q_{\mathrm{regular}}$ 
      & $n_{\mathrm{small}}$ 
      & $n_{\mathrm{regular}}$ 
      & p-value 
      & $c_L$ 
      & $c_U$\\
      \hline
       $H^{+}_0$  
       & \multirow{2}{*}{sum} 
       & \multirow{2}{*}{94} 
       & \multirow{2}{*}{98} 
       & \multirow{2}{*}{301} 
       & \multirow{2}{*}{338} 
       & 95.5 
       & 0 
       & $\infty$\\
       $H^{-}_0$ 
       &  
       &
       &  
       &  
       &  
       & 0.06$^{*}$ 
       & 0 
       & $\infty$\\
      \hline
      $H^{+}_0$  
       & \multirow{2}{*}{reading} 
       & \multirow{2}{*}{93} 
       & \multirow{2}{*}{96} 
       & \multirow{2}{*}{302} 
       & \multirow{2}{*}{338} 
       & 98.9
       & 0 
       & $\infty$\\
       $H^{-}_0$ 
       &  
       &
       &  
       &  
       &  
       & 0.05$^{*}$ 
       & 0.6 
       & $\infty$\\
      \hline
      $H^{+}_0$  
       & \multirow{2}{*}{math} 
       & \multirow{2}{*}{95} 
       & \multirow{2}{*}{101} 
       & \multirow{2}{*}{301} 
       & \multirow{2}{*}{339} 
       & 96.9 
       & 0 
       & $\infty$\\
       $H^{-}_0$ 
       &  
       &
       &  
       &  
       &  
       & 0.75$^{*}$ 
       & 0 
       & $\infty$\\
      \hline
    \end{tabular}
    \caption{\small Empirical results for the Project STAR data. The table reports p-values for the two dominance hypotheses in \eqref{eq:empApplication2}. All p-values are reported as percentages. We consider three outcomes: reading, math, and their sum. The effective sample sizes $(q_{\mathrm{small}},q_{\mathrm{regular}})$ are rescaled as $(cq_{\mathrm{small}},cq_{\mathrm{regular}})$ in the sensitivity analysis. The values $c_L$ and $c_U$ are the lower and upper breakpoints at which the 10\% test decision changes relative to the baseline choice $c=1$. Values of $0$ and $\infty$ indicate no change below or above $c=1$, respectively, over the range considered. The symbol $^{*}$ denotes rejection at the 10\% level.}
    \label{tab:qyStarFsdExact}
\end{table}

Table \ref{tab:qyStarFsdExact} reports the results for both dominance directions and all three outcomes. For the sum of reading and math scores, we do not reject the hypothesis that the conditional distribution of student achievement in small classes first-order stochastically dominates that in regular classes; the corresponding p-value is $95.5\%$. By contrast, the reverse dominance hypothesis is rejected at the 10\% level, with a p-value of $0.06\%$. The same qualitative conclusion holds for reading and math scores separately: $H^+_0$ is not rejected, whereas $H^-_0$ is rejected at the 1\% level for both outcomes. These results provide evidence that, conditional on median teacher experience, the distribution of student achievement under small classes first-order stochastically dominates that under regular classes. The last two columns of Table \ref{tab:qyStarFsdExact} report the same sensitivity analysis as in the previous application. For five of the six outcome--hypothesis combinations, the breakpoints satisfy $c_L=0$ and $c_U=\infty$. For the reverse dominance hypothesis for reading, $c_L=0.6$ and $c_U=\infty$. Thus, the conclusions are unchanged over a wide range of effective sample sizes around the baseline choice.

\begin{table}[h]\small
    \centering
    \begin{tabular}{cccccccc}
        \hline
        Outcome & Class & Sample & $q$ & $n$
        & Score & Count & Largest point mass (\%) \\
        \hline
        reading & small   & full      & --  & 1,739 & 434         & 48       & 2.8 \\
        reading & regular & full      & --  & 2,006 & 419         & 57       & 2.8 \\
        reading & small   & effective & 93  & --   & 434         & 6        & 6.5 \\
        reading & regular & effective & 96  & --   & 398 and 413         & 5 each     & 5.2 \\
        \hline
        math    & small   & full      & --  & 1,762 & 489         & 87       & 4.9 \\
        math    & regular & full      & --  & 2,032 & 478         & 100      & 4.9 \\
        math    & small   & effective & 95  & --   & 500          & 7       & 7.4\\
        math    & regular & effective & 101 & --   & 449         & 10       & 9.9 \\
        \hline
    \end{tabular}
    \caption{\small Largest empirical point masses in the full and effective
    samples for the reading and mathematics scores. Score denotes the score attaining the maximum observed frequency, and Count denotes the corresponding frequency.}
    \label{tab:star-score-point-masses}
\end{table}

\subsection{Details on the Simulations and the data dependent rule for \emph{q}}\label{app:simulations}
The parameter values used in the simulations of Section \ref{sec:simulations} are reported in the following tables. The parameters for Designs 1-4 are summarized in Table~\ref{tab:Simu-Designs-14}, those for Designs 5-7 in Table~\ref{tab:Simu-Designs-57}, and the parameter values for case (d)---power results---in Table~\ref{tab:Simu-cased}. Table \ref{tab:q-values} reports the mean values of \(q_y^*\) and \(q_x^*\) across simulations.

\subsubsection{Details on the implementation of the data dependent rule for \emph{q}}

Here we also discuss some numerical details on the implementation of the data dependent rule for $q$ introduced in \eqref{eq:rot-qy} and \eqref{eq:rot-qx}. We do this for a generic sample $\{(W_i,Z_i):1\le i\le n\}$, suppressing the sample-specific $Y$ and $X$ subscripts. Let $\hat\mu_Z$ be the sample median, $\hat\sigma_Z=\operatorname{IQR}(Z)/(2\Phi^{-1}(3/4))$, $\hat f_0=\phi_{\hat\mu_Z,\hat\sigma_Z}(\hat z_n)$, and
\[
 \hat C_Z=\frac{1}{\hat\sigma_Z^2\sqrt{2\pi e}},\qquad
 \hat C_W(\rho)=\frac{|\rho|}{\hat\sigma_Z\sqrt{1-\rho^2}}\frac{1}{\sqrt{2\pi}}~.
\]
For $|\rho|<1$, define
\begin{equation}\label{eq:q_n_rho}
 \hat c_n(\rho) :=2^{-1/3}
 \left(\frac{4\hat f_0^2}{\hat C_Z/2+\hat f_0\hat C_W(\rho)}\right)^{2/3}~,
 \qquad
 q_n(\rho) :=\max\left\{q_{\min}, \lceil n^{1/2}\hat c_n(\rho)\rceil\right\}~.
\end{equation}
Our codes uses $q_{\min}=10$ and $\hat\rho=\operatorname{sgn}(\tilde\rho) \min\{|\tilde\rho|,0.99\}$, where $\tilde\rho=\sin(\pi\hat\tau/2)$ and $\hat\tau$ is Kendall's rank correlation. Thus $\hat\rho\in[-0.99,0.99]$. The implementation rounds upward and imposes the lower bound $q_{\min}$.

The rule has a deterministic $O(n^{1/2})$ envelope. Since $\hat C_W(\rho)\ge0$, $\hat f_0\le(\hat\sigma_Z\sqrt{2\pi})^{-1}$, and $\hat C_Z=(\hat\sigma_Z^2\sqrt{2\pi e})^{-1}$,
\[
 \hat c_n(\rho)
 \le 2^{-1/3}\left(\frac{8\hat f_0^2}{\hat C_Z}\right)^{2/3}
 \le \left(\frac{16e}{\pi}\right)^{1/3}
 =:\kappa\approx2.4012~.
\]
Consequently,
\begin{equation}\label{eq:q-rule-envelope}
 q_n(\rho)\le \bar q_n:=
 \max\{q_{\min},\lceil\kappa n^{1/2}\rceil\}~.
\end{equation}
Hence $\bar q_n=o(n^{2/3})$ and the envelope in Assumption~\ref{ass:A-meas} holds with $\eta_n=0$. If $\delta_n\asymp n^{-1/2}$, then $\sqrt{\bar q_n}\delta_n=O(n^{-1/4})$, so the growth conditions in \eqref{eq:growth} hold.

The rank correlation $\hat\rho$ in \eqref{eq:rot-qy}--\eqref{eq:rot-qx} depends on the outcomes, so estimating it from the full sample is not $\mathcal A$-measurable and the resulting rule does not satisfy Assumption~\ref{ass:A-meas}. As anticipated in Section~\ref{sec:tuning-parameters}, we therefore estimate $\hat\rho$ from the observations that cannot enter the test neighborhood. Let the pilot size be $m_0=\bar q_n$, with $\bar q_n$ as in \eqref{eq:q-rule-envelope}, and let $M_0$ contain the $m_0$ observations closest to $\hat z_n$, breaking distance ties by an $\mathcal A$-measurable rule. We estimate Kendall's correlation from $\{(W_i,Z_i):i\notin M_0\}$, transform and clip it to obtain $\hat\rho_{\rm lo}$, and set $q_{\rm lo}=q_n(\hat\rho_{\rm lo})$. Since $q_{\rm lo}\le\bar q_n=m_0$, the selected neighborhood is contained in $M_0$, so the outcomes that determine the test observations and those used to compute $\hat\rho_{\rm lo}$ are disjoint.

This construction meets the requirement of Theorem~\ref{thm:large-q-general} under a minor extension. Define 
\[
 \mathcal G_n=\mathcal A\vee\sigma\{W_i:i\notin M_0\}~,
\] 
taking the join of the corresponding $\sigma$-fields for the two samples. While $q_{\rm lo}$ is strictly speaking not $\mathcal A$-measurable, it is indeed $\mathcal G_n$-measurable. The selected indices are also $\mathcal G_n$-measurable, while the selected outcomes are not included in $\mathcal G_n$ and, conditional on $\mathcal G_n$, remain mutually independent with conditional laws $P_{W,Z_i}$. The conditional product representation \eqref{eq:cond-product} therefore continues to hold with $\mathcal G_n$ in place of $\mathcal A$, and the proof of Theorem~\ref{thm:large-q-general} extends by conditioning on $\mathcal G_n$. Because $m_0=O(n^{1/2})$, the leave-out step discards only a vanishing fraction of the sample, so $\hat\rho_{\rm lo}$ and the full-sample estimate have the same probability limit under standard regularity conditions. In finite samples the rules based on the leave-out and full-sample estimates are numerically very close; we compute both and report them side by side in the robustness checks below.

A simpler $\mathcal A$-measurable alternative sets $\rho=0$: since $\hat C_W(\rho)\ge\hat C_W(0)=0$ gives $q_n(\rho)\le q_n(0)$, this dispenses with the leave-out step and yields the largest tuning parameter in the family we consider, at the cost of ignoring the estimated dependence. We report both this variant and the full-sample rule in the robustness checks below.

\begin{table}[h]\footnotesize
  \centering
  \begin{tabular}{ccccccc}\hline
      Design &$\mu_Y(z)$ &$\mu_X(z)$  & $\sigma_Y(Z)$ & $\sigma_X(Z)$  & $U,V$ & $z_{\ell}$  \\ \hline
       1a&  $z$&  $z$&   $z^2$& $z^2$ & $N(0,1)$ & 0.5\\
       1b&  $1.05z$&  $z$&  $z^2$&  $z^2$&  $N(0,1)$& 0.5\\
       1c&  $z$&  $z$&  $z^2$&  $z^2$& $N(0,1)$& 0.25, 0.75\\ \hline
       2a&  $z$ & \footnotesize{$z^2+0.25$}&   $z^2$&  $z^2$&  $N(0,1)$&  0.5\\
       2b&  $1.05z$ & \footnotesize{$0.5z + 0.25$} &   $z^2$&  $z^2$&  $N(0,1)$&  0.5\\
       2c&  $z$ & \footnotesize{$z-(z-0.25)(z-0.75)$} &   $z^2$&  $z^2$&  $N(0,1)$& 0.25, 0.75\\ \hline
       3a&  $z$&  $z$&   $z^2$&   $z^2$&  $U[0,1]$&  0.5\\
       3b&  $z+0.1z^2$& $z$&     $0.95z^2$& $z^2$&     $U[0,1]$&  0.5\\
       3c&  $z$&  $z$&   $z^2$&   $z^2$&  $U[0,1]$& 0.25, 0.75\\ \hline
       4a&  $\mu(z)$ & $\mu(z)$&  $1$ &  $1$&  $N(0,1)$&  0\\
       4b&  $\mu(z)+0.1$ & $\mu(z)$&  $1$ &  $1$& $N(0,1)$ & 0 \\
       4c&  $\mu(z)$ & $\mu(z)$&  $1$ &  $1$ &  $N(0,1)$& -0.5, 0.5\\ \hline 
  \end{tabular}
  \caption{\footnotesize Parameter values for the continuously distributed designs. Designs 1 to 3 are based on the location-scale model in \eqref{eq:loc-scale}, while Designs 4 is based on the RDD model in \eqref{eq:rdd-design}.}
  \label{tab:Simu-Designs-14}
\end{table}

\begin{table}[h]\footnotesize
  \centering
  \begin{tabular}{ccccccc}\hline
      Design &$\mu_Y(z)$ &$\mu_X(z)$  & $\sigma_Y(Z)$ & $\sigma_X(Z)$  & $U,V$ & $z_{\ell}$  \\ \hline
       5a&  0& 0&  $z^2$&  $z^2$      &  C-logN&  0.5\\
       5b&  0& 0&  $1.05z^2$&  $z^2$  & C-logN & 0.5 \\
       5c&  0& 0&  $z^2$&  $z^2$      &  C-logN& 0.25, 0.75\\ \hline
       6a&  0 & --&  -- &   --        & discrete & 0.5 \\
       6b&  1 & --&  -- &  --         & discrete & 0.5 \\
       6c&  0 & --&  -- &  --         &  discrete& 0.25, 0.75\\ \hline
       7a&  0 & --&  --&  --          & Binomial & 0.5 \\
       7b&  1 & --&  --&  --          & Binomial & 0.5 \\
       7c&  0 & --&  --&  --          &  Binomial& 0.25, 0.75\\ \hline 
  \end{tabular}
  \caption{\footnotesize Parameter values for the discrete and mixed designs. Designs 5 is based on the location-scale model in \eqref{eq:loc-scale}, while Designs 6 and 7 are discrete designs as described in the main text.}
  \label{tab:Simu-Designs-57}
\end{table}

\begin{table}[h]\footnotesize
  \centering
  \begin{tabular}{ccccccc}\hline
      Design &$\mu_Y(z)$ &$\mu_X(z)$  & $\sigma_Y(Z)$ & $\sigma_X(Z)$  & $U,V$ & $z_{\ell}$  \\ \hline
1d&  $0.95z$&  $z$&  $z^2$&  $z^2$&         $N(0,1)$& 0.5\\ 
2d&  $z$ & $0.6z + 0.25$&   $z^2$&  $z^2$&  $N(0,1)$&  0.5\\ 
3d&  $z$&  $z$&     $0.90z^2$& $z^2$&       $U[0,1]$&  0.5\\ 
4d&  $\mu(z)$ & $\mu(z)$&  $0.5 + z^2$ &  $1$&  $N(0,1)$&  0 \\ 
5d&  0& 0&  $0.90z^2$&  $z^2$               & C-logN & 0.5 \\ 
6d&  -1/2 & --&  -- &  --                   & discrete & 0.5 \\ 
7d&  -1 & --&  --&  --                      & Binomial & 0.5 \\ \hline 
  \end{tabular}
  \caption{\footnotesize Parameter values for all designs under the alternative hypothesis $H_1$.}
  \label{tab:Simu-cased}
\end{table}

\begin{table}[h]\footnotesize
  \begin{center}
  \begin{tabular}{rrrrrrrrrrrrrrrr}
  \hline
  \multirow{2}{*}{\bf Design} \footnotesize
      & \multicolumn{3}{c}{\bf 1} & 
      & \multicolumn{3}{c}{\bf 2} & 
      & \multicolumn{3}{c}{\bf 3} &
      & \multicolumn{3}{c}{\bf 4} \\
  \cline{2-4} \cline{6-8} \cline{10-12} \cline{14-16}
   & a & b & c 
   & & a & b & c 
   & & a & b & c
   & & a & b & c \\
  \hline
  $q_y^*$ 
   & 46.34 & 45.53 & 29.28 
   & & 46.32 & 45.53 & 29.27 
   & & 21.25 & 20.09 & 15.17 
   & & 52.00 & 51.97 & 40.22 \\
  $q_x^*$ 
   & 46.33 & 46.34 & 29.28 
   & & 48.36 & 56.79 & 29.18 
   & & 21.25 & 21.26 & 15.17 
   & & 52.35 & 52.35 & 40.22 \\
  \hline
  
  \multirow{2}{*}{\bf Design} 
      & \multicolumn{3}{c}{\bf 5} & 
      & \multicolumn{3}{c}{\bf 6} & 
      & \multicolumn{3}{c}{\bf 7} &
      & \multicolumn{3}{c}{\bf } \\
  \cline{2-4} \cline{6-8} \cline{10-12} \cline{14-16}
   & a & b & c 
   & & a & b & c 
   & & a & b & c
   & &  &  &  \\
  \hline
  $q_y^*$ 
   & 39.18 & 39.18 & 25.94 
   & & 66.88 & 74.26 & 37.72 
   & & 33.10 & 33.85 & 22.60 
   & & -- & -- & -- \\
  $q_x^*$ 
   & 39.16 & 39.20 & 25.92 
   & & 66.92 & 66.90 & 37.74 
   & & 33.11 & 33.11 & 22.59 
   & & -- & -- & -- \\
  \hline
  \end{tabular}
  \caption{\footnotesize Mean values of $q_y^*$ and $q_x^*$ across simulations: $n=1,000$, $\alpha=10\%$, $MC=10,000$.}\label{tab:q-values}
  \end{center}
  \end{table}

\subsubsection{Robustness checks}

In this appendix, we assess the sensitivity of our procedure to three choices: the sample size, the nominal level, and the rule used to select the tuning parameters $(q_y,q_x)$. We consider Design 4 (continuous, RDD) and Design 7 (discrete, binomial) from Section~6, together with $\alpha \in \{0.05,0.10\}$ and three $q$-selection rules: the data-dependent rule proposed in the paper (\emph{leave-out}); a version of the same rule that estimates $\rho$ without leaving observations out (\emph{full}); and a variant that ignores the estimated dependence altogether by imposing $\rho=0$ (\emph{rho0}). We further consider two additional sample sizes ($n=500$ and $n=2{,}000$, alongside the paper's $n=1{,}000$) and two unbalanced sample sizes, $(n_y,n_x)=(500,1{,}500)$ and $(1{,}500,500)$.

Table~\ref{tab:extended-size} reports rejection rates under $H_0$ and Table~\ref{tab:extended-power} under $H_1$, with $MC=10{,}000$ replications per cell. Each cell reports the rejection rate (\%) under the three $q$-rules, separated by slashes, in the order full/leave-out/rho0.

The results confirm the conclusions in the main text. For Design 4, size is controlled closely at both nominal levels and across all sample sizes, and the choice of $q$-rule is essentially immaterial: the dependence between the outcome and the running variable is weak in this design, so $\hat\rho$ is close to zero regardless of how it is estimated, and the three rules select nearly identical values of $q$. For Design 7, rejection rates under $H_0$ remain below the nominal level throughout -- the test never over-rejects in this grid -- and the default rule is conservative, as in the main text; the refined critical value brings rejection rates closer to the nominal level. Imposing $\rho=0$ increases $q$ substantially, which improves power. 
Overall, the procedure's size and power properties are robust to the sample size, the nominal level, and reasonable variation in how the tuning parameter is selected.

\begin{table}[h!]\footnotesize
\begin{center}
\begin{tabular}{lcccccc}
\hline
$(n_y,n_x)$ & D4 $\alpha{=}.05$ & D4 $\alpha{=}.10$ & D7 $\alpha{=}.05$ def & D7 $\alpha{=}.05$ ref & D7 $\alpha{=}.10$ def & D7 $\alpha{=}.10$ ref \\
\hline
(500,500) & 5.1/5.1/4.8 & 9.7/9.7/7.4 & 1.6/1.6/1.3 & 2.6/2.9/1.4 & 3.5/3.7/2.5 & 4.8/4.8/6.6 \\
(1000,1000) & 4.8/4.8/3.4 & 9.6/9.4/7.5 & 1.4/1.5/1.1 & 2.3/2.4/2.6 & 3.5/3.4/2.6 & 4.8/4.7/3.9 \\
(2000,2000) & 4.9/5.1/4.9 & 9.8/9.6/9.3 & 1.6/1.5/1.5 & 2.1/2.1/1.5 & 3.2/3.1/3.1 & 4.5/4.6/3.2 \\
(500,1500) & 5.2/5.2/5.3 & 10.4/10.3/10.4 & 1.6/1.4/1.5 & 2.4/2.3/2.6 & 3.6/3.8/3.3 & 5.8/5.9/5.9 \\
(1500,500) & 4.5/4.5/4.6 & 9.4/9.6/9.6 & 1.6/1.6/1.4 & 2.4/2.5/2.8 & 3.3/3.5/3.5 & 5.7/5.8/5.9 \\
\hline
\end{tabular}
\caption{\footnotesize Rejection rates (\%) under $H_0$ for Designs 4 (RDD) and 7 (binomial), $\alpha \in \{.05,.10\}$, default and refined critical values (D7 only), $MC=10{,}000$. Each cell reports \texttt{full/leaveout/rho0}, the rejection rate under the three $q$-selection rules.}
\label{tab:extended-size}
\end{center}
\end{table}

\begin{table}[h!]\footnotesize
\begin{center}
\begin{tabular}{lcccccc}
\hline
$(n_y,n_x)$ & D4 $\alpha{=}.05$ & D4 $\alpha{=}.10$ & D7 $\alpha{=}.05$ def & D7 $\alpha{=}.05$ ref & D7 $\alpha{=}.10$ def & D7 $\alpha{=}.10$ ref \\
\hline
(500,500) & 36.0/35.9/35.7 & 53.3/53.0/46.2 & 9.1/8.4/16.7 & 13.3/12.9/16.9 & 16.1/15.8/23.6 & 20.1/19.3/31.1 \\
(1000,1000) & 52.6/52.3/46.3 & 69.1/69.1/64.5 & 12.2/12.3/19.8 & 15.3/15.0/32.6 & 21.5/21.1/32.6 & 26.9/26.6/40.0 \\
(2000,2000) & 72.4/72.2/73.9 & 85.3/85.3/85.2 & 15.8/15.5/34.4 & 20.4/20.1/34.7 & 25.4/25.1/47.0 & 32.9/32.3/47.3 \\
(500,1500) & 35.6/35.7/36.5 & 53.0/53.2/54.5 & 10.8/10.7/21.8 & 15.0/14.7/29.0 & 19.2/19.4/33.4 & 26.3/26.0/44.1 \\
(1500,500) & 36.4/36.3/37.7 & 52.7/52.5/53.4 & 10.9/11.0/22.0 & 14.7/14.9/30.5 & 19.2/19.1/33.7 & 26.2/26.5/44.5 \\
\hline
\end{tabular}
\caption{\footnotesize Rejection rates (\%) under $H_1$ for Designs 4 (RDD) and 7 (binomial), $\alpha \in \{.05,.10\}$, default and refined critical values (D7 only), $MC=10{,}000$. Each cell reports \texttt{full/leaveout/rho0}, the rejection rate under the three $q$-selection rules.}
\label{tab:extended-power}
\end{center}
\end{table}

\subsubsection{Finite-sample effects of estimating the conditioning point}
We investigate the finite-sample effects of estimating the conditioning point using two designs based on Designs 1 and 2 of Section \ref{sec:simulations}. In both, $Z \sim \mathrm{Beta}(2,2)$, the two samples are independent with $n_y = n_x = n$, the conditioning point is the median of $Z$, $z_0 = 1/2$, and $\hat z_n$ is the sample median of the pooled $Z$ observations. The tuning parameters $(q_y, q_x)$ are chosen by the rule in Section \ref{sec:tuning-parameters}, evaluated at the same point used to select the nearest neighbors.

As before, Design (D1a) is such that $F_Y(\cdot|z) = F_X(\cdot|z)$ at every $z$ and the null hypothesis in \eqref{eq:H0} holds with equality at every conditioning point. In Design (D1b) the null hypothesis holds in its interior. In Design (D1d) the null hypothesis is violated at every $z$. The second design (D2a) sets $\mu_y(z) = z$ and $\mu_x(z) = z^2 + 1/4$, so that $\mu_x(z) - \mu_y(z) = (z - 1/2)^2$ with equal conditional variances. The null hypothesis then holds, with equality, only at $z_0 = 1/2$, and is violated at every other conditioning point. This is the most demanding configuration for an estimated conditioning point: any deviation of $\hat z_n$ from $z_0$ centers the neighborhoods at points where the null hypothesis fails. Table \ref{tab:est-z0} reports the rejection rates at $\alpha = 10\%$ using the known conditioning point $z_0$ versus the estimated $\hat{z}$ (pooled sample median of $Z$), for $n = 500, 1{,}000, 2{,}000$. D1a/D1b/D2a report size, D1d reports power. Mean $|\hat{z} - z_0|$ measures how well the estimated center concentrates around the true one as $n$ grows.

\begin{table}[h!]\footnotesize
  \begin{center}
  \begin{tabular}{lcccccccccccc}
  \hline
   & \multicolumn{4}{c}{$n = 500$} & \multicolumn{4}{c}{$n = 1{,}000$} & \multicolumn{4}{c}{$n = 2{,}000$} \\
  \cline{2-13}
   & D1a & D1b & D1d & D2a & D1a & D1b & D1d & D2a & D1a & D1b & D1d & D2a \\
  \hline
  KS ($z_0$)      & 9.5 & 5.1 & 17.9 & 10.3 & 9.9 & 4.6 & 19.1 & 9.6 & 10.2 & 3.5 & 21.1 & 10.2 \\
  KS ($\hat{z}$)       & 9.4 & 5.3 & 18.1 & 10.4 & 9.6 & 4.6 & 19.4 & 9.7 & 10.0 & 3.7 & 21.0 & 10.8 \\
  Mean $|\hat{z} - z_0|$ & 0.009 & 0.008 & 0.008 & 0.008 & 0.006 & 0.006 & 0.006 & 0.006 & 0.004 & 0.004 & 0.004 & 0.004 \\
  \hline
  \end{tabular}
  \caption{\footnotesize Rejection rates (\%) at $\alpha = 10\%$ using the known conditioning point $z_0$ versus the estimated $\hat{z}$, for $n = 500, 1{,}000, 2{,}000$. D1a/D1b/D2a report size, D1d reports power.}
  \label{tab:est-z0}
  \end{center}
  \end{table}

\normalsize
\bibliography{CSD_Ref.bib}

\begin{thebibliography}{52}
\newcommand{\enquote}[1]{``#1''}
\expandafter\ifx\csname natexlab\endcsname\relax\def\natexlab#1{#1}\fi

\bibitem[\protect\citeauthoryear{Abadie}{Abadie}{2002}]{abadie:2002}
\textsc{Abadie, A.} (2002): \enquote{Bootstrap tests for distributional treatment effects in instrumental variable models,} \emph{Journal of the American statistical Association}, 97, 284--292.

\bibitem[\protect\citeauthoryear{Anderson}{Anderson}{1996}]{anderson:1996}
\textsc{Anderson, G.} (1996): \enquote{Nonparametric tests of stochastic dominance in income distributions,} \emph{Econometrica: Journal of the Econometric Society}, 1183--1193.

\bibitem[\protect\citeauthoryear{Andrews and Shi}{Andrews and Shi}{2017}]{andrews-shi-2017-JOE}
\textsc{Andrews, D.~W. and X.~Shi} (2017): \enquote{Inference based on many conditional moment inequalities,} \emph{Journal of Econometrics}, 196, 275--287.

\bibitem[\protect\citeauthoryear{Armstrong and Koles{\'a}r}{Armstrong and Koles{\'a}r}{2018}]{armstrong/kolesar:18}
\textsc{Armstrong, T.~B. and M.~Koles{\'a}r} (2018): \enquote{Optimal inference in a class of regression models,} \emph{Econometrica}, 86, 655--683.

\bibitem[\protect\citeauthoryear{Barcellos, Carvalho, and Turley}{Barcellos et~al.}{2023}]{barcellos/carvalho/turley:2023}
\textsc{Barcellos, S.~H., L.~S. Carvalho, and P.~Turley} (2023): \enquote{Distributional effects of education on health,} \emph{Journal of Human Resources}, 58, 1273--1306.

\bibitem[\protect\citeauthoryear{Barrett and Donald}{Barrett and Donald}{2003}]{barrett/donald:2003}
\textsc{Barrett, G.~F. and S.~G. Donald} (2003): \enquote{Consistent tests for stochastic dominance,} \emph{Econometrica}, 71, 71--104.

\bibitem[\protect\citeauthoryear{Battistin, Brugiavini, Rettore, and Weber}{Battistin et~al.}{2009}]{battistin/Brugianivi/Rettore/Weber:2009}
\textsc{Battistin, E., A.~Brugiavini, E.~Rettore, and G.~Weber} (2009): \enquote{The retirement consumption puzzle: evidence from a regression discontinuity approach,} \emph{American Economic Review}, 99, 2209--2226.

\bibitem[\protect\citeauthoryear{Bhattacharya}{Bhattacharya}{1974}]{bhattacharya:74}
\textsc{Bhattacharya, P.} (1974): \enquote{Convergence of sample paths of normalized sums of induced order statistics,} \emph{The Annals of Statistics}, 1034--1039.

\bibitem[\protect\citeauthoryear{Blackman}{Blackman}{1956}]{blackman:1956}
\textsc{Blackman, J.} (1956): \enquote{An extension of the Kolmogorov distribution,} \emph{The Annals of Mathematical Statistics}, 513--520.

\bibitem[\protect\citeauthoryear{Bugni and Canay}{Bugni and Canay}{2021}]{bugni/canay:21}
\textsc{Bugni, F.~A. and I.~A. Canay} (2021): \enquote{Testing Continuity of a Density via g-order statistics in the Regression Discontinuity Design,} \emph{Journal of Econometrics}, 221, 138--159.

\bibitem[\protect\citeauthoryear{Bugni, Canay, and Kim}{Bugni et~al.}{2025}]{bugni/canay/kim:26b}
\textsc{Bugni, F.~A., I.~A. Canay, and D.~Kim} (2025): \enquote{On the Rate of Convergence of Induced Ordered Statistics and their Applications,} Tech. rep., Northwestern University.

\bibitem[\protect\citeauthoryear{Bugni, Canay, and Shaikh}{Bugni et~al.}{2018}]{bugni/canay/shaikh:18}
\textsc{Bugni, F.~A., I.~A. Canay, and A.~M. Shaikh} (2018): \enquote{Inference under Covariate Adaptive Randomization,} \emph{Journal of the American Statistical Association}, 113, 1784--1796.

\bibitem[\protect\citeauthoryear{Canay and Kamat}{Canay and Kamat}{2018}]{canay/kamat:18}
\textsc{Canay, I.~A. and V.~Kamat} (2018): \enquote{Approximate Permutation Tests and Induced Order Statistics in the Regression Discontinuity Design,} \emph{The Review of Economic Studies}, 85, 1577--1608.

\bibitem[\protect\citeauthoryear{Canay, Romano, and Shaikh}{Canay et~al.}{2017}]{canay/romano/shaikh:17}
\textsc{Canay, I.~A., J.~P. Romano, and A.~M. Shaikh} (2017): \enquote{Randomization Tests under an Approximate Symmetry Assumption,} \emph{Econometrica}, 85, 1013--1030.

\bibitem[\protect\citeauthoryear{Caughey, Dafoe, Li, and Miratrix}{Caughey et~al.}{2023}]{caughey/etal:2023}
\textsc{Caughey, D., A.~Dafoe, X.~Li, and L.~Miratrix} (2023): \enquote{Randomisation inference beyond the sharp null: bounded null hypotheses and quantiles of individual treatment effects,} \emph{Journal of the Royal Statistical Society Series B: Statistical Methodology}, 85, 1471--1491.

\bibitem[\protect\citeauthoryear{Chang, Lee, and Whang}{Chang et~al.}{2015}]{chang-lee-whang-2015-EJ}
\textsc{Chang, M., S.~Lee, and Y.-J. Whang} (2015): \enquote{Nonparametric tests of conditional treatment effects with an application to single-sex schooling on academic achievements,} \emph{The Econometrics Journal}, 18, 307--346.

\bibitem[\protect\citeauthoryear{Chung and Romano}{Chung and Romano}{2013}]{chung/romano:13}
\textsc{Chung, E. and J.~P. Romano} (2013): \enquote{Exact and asymptotically robust permutation tests,} \emph{The Annals of Statistics}, 41, 484--507.

\bibitem[\protect\citeauthoryear{{[dataset]}~Achilles, Bain, Bellott, Boyd-Zaharias, Finn, Folger, Johnston, and Word}{{[dataset]}~Achilles et~al.}{2008}]{achilles/etal:2008data}
\textsc{{[dataset]}~Achilles, C.~M., H.~P. Bain, F.~Bellott, J.~Boyd-Zaharias, J.~Finn, J.~Folger, J.~Johnston, and E.~Word} (2008): \enquote{Tennessee's Student Teacher Achievement Ratio ({STAR}) Project,} .

\bibitem[\protect\citeauthoryear{{[dataset]}~Shen and Zhang}{{[dataset]}~Shen and Zhang}{2016}]{shen/zhang:2016data}
\textsc{{[dataset]}~Shen, S. and X.~Zhang} (2016): \enquote{Replication data for: Distributional Tests for Regression Discontinuity: Theory and Empirical Examples,} Harvard Dataverse, version 1.1.

\bibitem[\protect\citeauthoryear{David and Galambos}{David and Galambos}{1974}]{david/galambos:74}
\textsc{David, H. and J.~Galambos} (1974): \enquote{The asymptotic theory of concomitants of order statistics,} \emph{Journal of Applied Probability}, 762--770.

\bibitem[\protect\citeauthoryear{Davidson and Duclos}{Davidson and Duclos}{2000}]{davidson/duclos:2000}
\textsc{Davidson, R. and J.-Y. Duclos} (2000): \enquote{Statistical inference for stochastic dominance and for the measurement of poverty and inequality,} \emph{Econometrica}, 68, 1435--1464.

\bibitem[\protect\citeauthoryear{Davidson and Duclos}{Davidson and Duclos}{2013}]{davidson/duclos:2013}
---\hspace{-.1pt}---\hspace{-.1pt}--- (2013): \enquote{Testing for Restricted Stochastic Dominance,} \emph{Econometric Reviews}, 32, 84--125.

\bibitem[\protect\citeauthoryear{Delgado and Escanciano}{Delgado and Escanciano}{2013}]{delgado/escaciano:2013}
\textsc{Delgado, M.~A. and J.~C. Escanciano} (2013): \enquote{Conditional stochastic dominance testing,} \emph{Journal of Business \& Economic Statistics}, 31, 16--28.

\bibitem[\protect\citeauthoryear{Donald, Hsu, and Barrett}{Donald et~al.}{2012}]{donald/hsu/barrett:2012}
\textsc{Donald, S.~G., Y.-C. Hsu, and G.~F. Barrett} (2012): \enquote{Incorporating covariates in the measurement of welfare and inequality: methods and applications,} \emph{The Econometrics Journal}, 15, C1--C30.

\bibitem[\protect\citeauthoryear{Dubhashi and Panconesi}{Dubhashi and Panconesi}{2009}]{dubhashi/panconesi:09}
\textsc{Dubhashi, D.~P. and A.~Panconesi} (2009): \emph{Concentration of Measure for the Analysis of Randomized Algorithms}, Cambridge University Press.

\bibitem[\protect\citeauthoryear{Durbin}{Durbin}{1973}]{durbin:1973}
\textsc{Durbin, J.} (1973): \emph{Distribution theory for tests based on the sample distribution function}, SIAM.

\bibitem[\protect\citeauthoryear{Gnedenko and Korolyuk}{Gnedenko and Korolyuk}{1951}]{gnedenko/korolyuk:1951}
\textsc{Gnedenko, B.~V. and V.~S. Korolyuk} (1951): \enquote{On the maximum divergence of two empirical distributions (in Russian),} \emph{DAN}, 80, 525.

\bibitem[\protect\citeauthoryear{Goldman and Kaplan}{Goldman and Kaplan}{2018}]{goldman/kaplan:2018}
\textsc{Goldman, M. and D.~M. Kaplan} (2018): \enquote{Comparing distributions by multiple testing across quantiles or {CDF} values,} \emph{Journal of Econometrics}, 206, 143--166.

\bibitem[\protect\citeauthoryear{Gonzalo and Olmo}{Gonzalo and Olmo}{2014}]{gonzalo/olmo:2014}
\textsc{Gonzalo, J. and J.~Olmo} (2014): \enquote{Conditional stochastic dominance tests in dynamic settings,} \emph{International Economic Review}, 55, 819--838.

\bibitem[\protect\citeauthoryear{Hajek, Sidak, and Sen}{Hajek et~al.}{1999}]{hajek/sidak/sen:99}
\textsc{Hajek, J., Z.~Sidak, and P.~K. Sen} (1999): \emph{Theory of rank tests}, Academic press, 2nd edition ed.

\bibitem[\protect\citeauthoryear{Hodges}{Hodges}{1958}]{hodges:1958}
\textsc{Hodges, J.} (1958): \enquote{The significance probability of the Smirnov two-sample test,} \emph{Arkiv f{\"o}r matematik}, 3, 469--486.

\bibitem[\protect\citeauthoryear{Hájek and Šidák}{Hájek and Šidák}{1967}]{hajek/sidak:1967}
\textsc{Hájek, J. and Z.~Šidák} (1967): \emph{Theory of Rank Tests}, New York: Academic Press.

\bibitem[\protect\citeauthoryear{Kaufmann and Reiss}{Kaufmann and Reiss}{1992}]{kaufmann/reiss:92}
\textsc{Kaufmann, E. and R.-D. Reiss} (1992): \enquote{On conditional distributions of nearest neighbors,} \emph{Journal of Multivariate Analysis}, 42, 67--76.

\bibitem[\protect\citeauthoryear{Korolyuk}{Korolyuk}{1955}]{korolyuk:1955}
\textsc{Korolyuk, V.~S.} (1955): \enquote{On the discrepancy of empiric distributions for the case of two independent samples,} \emph{Izv. Akad. Nauk SSSR Ser. Mat.}, 19, 81--96.

\bibitem[\protect\citeauthoryear{Lehmann}{Lehmann}{1951}]{lehmann:1951}
\textsc{Lehmann, E.~L.} (1951): \enquote{Consistency and unbiasedness of certain nonparametric tests,} \emph{The annals of mathematical statistics}, 165--179.

\bibitem[\protect\citeauthoryear{Lehmann and Romano}{Lehmann and Romano}{2005}]{lehmann/romano:05}
\textsc{Lehmann, E.~L. and J.~P. Romano} (2005): \emph{Testing Statistical Hypotheses}, Springer, New York, 3rd ed.

\bibitem[\protect\citeauthoryear{Li and Sunder}{Li and Sunder}{2024}]{Li/Sunder:2024}
\textsc{Li, Y. and N.~Sunder} (2024): \enquote{Distributional effects of education on mental health,} \emph{Labour Economics}, 88, 102528.

\bibitem[\protect\citeauthoryear{Linton, Maasoumi, and Whang}{Linton et~al.}{2005}]{Linton/maasoumi/whang:2005}
\textsc{Linton, O., E.~Maasoumi, and Y.-J. Whang} (2005): \enquote{Consistent testing for stochastic dominance under general sampling schemes,} \emph{The Review of Economic Studies}, 72, 735--765.

\bibitem[\protect\citeauthoryear{Linton, Seo, and Whang}{Linton et~al.}{2023}]{linton/seo/whang:23}
\textsc{Linton, O., M.~H. Seo, and Y.-J. Whang} (2023): \enquote{Testing stochastic dominance with many conditioning variables,} \emph{Journal of Econometrics}, 235, 507--527.

\bibitem[\protect\citeauthoryear{Linton, Song, and Whang}{Linton et~al.}{2010}]{linton/song/whang:2010}
\textsc{Linton, O., K.~Song, and Y.-J. Whang} (2010): \enquote{An improved bootstrap test of stochastic dominance,} \emph{Journal of Econometrics}, 154, 186--202.

\bibitem[\protect\citeauthoryear{Linton, Whang, and Yen}{Linton et~al.}{2016}]{linton/whang/yen:2016}
\textsc{Linton, O., Y.-J. Whang, and Y.-M. Yen} (2016): \enquote{A Nonparametric Test of a Strong Leverage Hypothesis,} \emph{Journal of Econometrics}, 194, 153--186.

\bibitem[\protect\citeauthoryear{Lok and Tabri}{Lok and Tabri}{2021}]{lok/tabri:2021}
\textsc{Lok, T.~M. and R.~V. Tabri} (2021): \enquote{An Improved Bootstrap Test for Restricted Stochastic Dominance,} \emph{Journal of Econometrics}, 224, 371--393.

\bibitem[\protect\citeauthoryear{Massart}{Massart}{1990}]{massart:90}
\textsc{Massart, P.} (1990): \enquote{The Tight Constant in the {Dvoretzky--Kiefer--Wolfowitz} Inequality,} \emph{The Annals of Probability}, 18, 1269--1283.

\bibitem[\protect\citeauthoryear{McFadden}{McFadden}{1989}]{mcfadden:1989}
\textsc{McFadden, D.} (1989): \enquote{Testing for stochastic dominance,} in \emph{Studies in the economics of uncertainty: In honor of Josef Hadar}, Springer, 113--134.

\bibitem[\protect\citeauthoryear{Pollard}{Pollard}{2002}]{pollard:02}
\textsc{Pollard, D.} (2002): \emph{A User's Guide to Measure Theoretic Probability}, Cambrigde University Press, New York.

\bibitem[\protect\citeauthoryear{Qu and Yoon}{Qu and Yoon}{2015}]{qu/yoon:2015}
\textsc{Qu, Z. and J.~Yoon} (2015): \enquote{Nonparametric estimation and inference on conditional quantile processes,} \emph{Journal of Econometrics}, 185, 1--19.

\bibitem[\protect\citeauthoryear{Qu and Yoon}{Qu and Yoon}{2019}]{qu/yoon:2019}
---\hspace{-.1pt}---\hspace{-.1pt}--- (2019): \enquote{Uniform inference on quantile effects under sharp regression discontinuity designs,} \emph{Journal of Business \& Economic Statistics}, 37, 625--647.

\bibitem[\protect\citeauthoryear{Reiss}{Reiss}{1989}]{reiss:89}
\textsc{Reiss, R.-D.} (1989): \emph{Approximate distributions of order statistics: with applications to nonparametric statistics}, Springer-Verlag, New York.

\bibitem[\protect\citeauthoryear{Shen}{Shen}{2019}]{shen:2019}
\textsc{Shen, S.} (2019): \enquote{Estimation and inference of distributional partial effects: Theory and application,} \emph{Journal of Business \& Economic Statistics}, 37, 54--66.

\bibitem[\protect\citeauthoryear{Shen and Zhang}{Shen and Zhang}{2016}]{shen/zhang:16}
\textsc{Shen, S. and X.~Zhang} (2016): \enquote{Distributional Tests for Regression Discontinuity: Theory and Empirical Examples,} \emph{Review of Economics and Statistics}.

\bibitem[\protect\citeauthoryear{Whang}{Whang}{2019}]{whang:2019}
\textsc{Whang, Y.-J.} (2019): \enquote{Econometric analysis of stochastic dominance,} \emph{Cambridge Books}.

\bibitem[\protect\citeauthoryear{Wu and Kaplan}{Wu and Kaplan}{2025}]{wu/kaplan:2025}
\textsc{Wu, Q. and D.~M. Kaplan} (2025): \enquote{Multiple testing of stochastic monotonicity,} Tech. rep., Department of Economics, University of Missouri.

\end{thebibliography}
\end{document}